\pdfoutput=1 

\documentclass[a4paper,UKenglish,cleveref, autoref, thm-restate]{lipics-v2021}

\nolinenumbers 

\usepackage{xspace}
\usepackage{algorithm}
\usepackage[noend]{algpseudocode}
\usepackage{amsmath}
\usepackage{amsthm}
\usepackage{mathtools}
\usepackage{tikz}
\usepackage{textcomp}
\usepackage{stmaryrd}
\usepackage{varwidth}
\usepackage{graphicx}

\usepackage{BOONDOX-cal}

\DeclareMathAlphabet{\pazocal}{OMS}{zplm}{m}{n}

\usetikzlibrary{automata, graphs,positioning,chains,arrows,snakes,decorations.pathmorphing}

\algrenewcommand\algorithmicindent{1.3em}%
\algnewcommand\algorithmicforeach{\textbf{for each}}
\algdef{S}[FOR]{ForEach}[1]{\algorithmicforeach\ #1\ \algorithmicdo}

\usepackage{todonotes}

\newcommand{\cL}{\pazocal{L}}

\newcommand{\cO}{\pazocal{O}}

\newcommand{\Stream}{\pazocal{S}}

\newcommand{\cA}{\pazocal{A}}
\newcommand{\cT}{\pazocal{T}\!}

\newcommand{\oalph}{\Omega}

\newcommand{\qinit}{I}

\newcommand{\out}{\textsf{out}}
\newcommand{\sem}[1]{\llbracket #1\mkern1mu\rrbracket}

\newcommand{\trans}[1]{\overset{#1}{\longrightarrow}}

\newcommand{\msom}{\text{MSO}_{\textsf{match}}}

\newcommand{\inpAlph}{\Omega}
\newcommand{\inpw}{x}
\newcommand{\outw}{y}
\newcommand{\outtime}{\mathsf{time}}
\newcommand{\outdelay}{\mathsf{delay}}

\newcommand{\yield}[1]{\texttt{yield}[#1]}

\newcommand{\enumvpt}{\textsc{EnumVPT}}

\newcommand{\D}{\pazocal{D}}
\newcommand{\cD}{\pazocal{D}}

\newcommand{\spanset}{\mathsf{Spans}}
\newcommand{\varset}{\mathsf{Vars}}

\newcommand{\spanc}[2]{[#1,#2\rangle}
\newcommand{\smap}{\mu}
\newcommand{\varcaptures}[1]{\pazocal{C}_{#1}}
\newcommand{\gprod}[1]{\Rightarrow_{#1}}

\newcommand{\splain}{\textsf{plain}}

\newcommand{\dsname}{enumerable compact set\xspace}
\newcommand{\dsnamebigcaps}{Enumerable Compact Set\xspace}

\newcommand{\dsabbr}{ECS\xspace}

\newcommand{\dsepsabbr}{$\eps$-ECS\xspace}
\newcommand{\enumds}{\textsc{Enum-ECS}}

\newcommand{\odepth}{\mathsf{odepth}}
\newcommand{\add}{\mathsf{add}}
\renewcommand{\prod}{\mathsf{prod}}
\newcommand{\union}{\mathsf{union}}

\newcommand{\push}{\mathsf{push}}
\newcommand{\pop}{\mathsf{pop}}

\newcommand{\varop}[1]{\{_{#1}}
\newcommand{\varcl}[1]{\}_{#1}}

\newcommand{\vpa}{VPA\xspace}
\newcommand{\vptname}{visibly pushdown transducer\xspace}
\newcommand{\vptnames}{visibly pushdown transducers\xspace}
\newcommand{\vpt}{VPT\xspace}
\newcommand{\vpts}{VPT\xspace}
\newcommand{\lindelay}{output-linear delay\xspace}

\newcommand{\enumE}{\pazocal{E}}
\newcommand{\nat}{\mathbb{N}}
\newcommand{\eps}{\varepsilon}
\renewcommand{\epsilon}{\varepsilon}
\renewcommand{\L}{\pazocal{L}}
\newcommand{\opS}{\Sigma^{\texttt{<}}}
\newcommand{\clS}{\Sigma^{\texttt{>}}}
\newcommand{\noS}{\Sigma^{\texttt{|}}}
\newcommand{\tinyop}{{\scriptsize\texttt{<}}}
\newcommand{\tinycl}{{\scriptsize\texttt{>}}}
\newcommand{\wnS}{\Sigma^{\texttt{<*>}}}
\newcommand{\pwnS}{\textsf{prefix}(\wnS)}
\newcommand{\op}[1]{\tinyop #1}
\newcommand{\cl}[1]{#1\tinycl}
\newcommand{\br}[1]{\llbracket #1 \mkern1mu \rrbracket}
\newcommand{\depth}{\mathsf{depth}}

\DeclareMathOperator{\oout}{%
	\ooalign{\raisebox{0ex}{$o$}\cr\hidewidth\rotatebox[origin=c]{15}{\raisebox{.2ex}{\scalebox{.7}{$\ \boldsymbol{\smallsmile}$}}}\hidewidth}}%

\DeclareMathOperator{\ooutscr}{\scriptsize{\ooalign{\raisebox{0ex}{$o$}\cr\hidewidth\rotatebox[origin=c]{15}{\raisebox{.2ex}{\scalebox{.7}{$\ \,  \boldsymbol{\smallsmile}$}}}\hidewidth}}}

\newcommand{\ogapname}{output weight\xspace}
\newcommand{\outgap}{\mathsf{outputweight}}

\newcommand{\cQ}{\pazocal{Q}}
\newcommand{\cI}{\pazocal{I}}

\title{Streaming enumeration on nested documents} 
\titlerunning{Streaming enumeration on nested documents}

\author{Martín Muñoz}{Pontificia Universidad Católica de Chile \and Millennium Institute for Foundational Research on Data}{mmunos@uc.cl}{}{}
\author{Cristian Riveros}{Pontificia Universidad Católica de Chile \and Millennium Institute for Foundational Research on Data}{cristian.riveros@uc.cl}{}{}

\authorrunning{Martín Muñoz and Cristian Riveros}
\Copyright{Martín Muñoz and Cristian Riveros}



\hideLIPIcs  

\funding{This work was funded by ANID - Millennium Science Initiative Program - Code ICN17\_002.}

\EventEditors{Dan Olteanu and Nils Vortmeier}
\EventNoEds{2}
\EventLongTitle{25th International Conference on Database Theory (ICDT 2022)}
\EventShortTitle{ICDT 2022}
\EventAcronym{ICDT}
\EventYear{2022}
\EventDate{March 29--April 1, 2022}
\EventLocation{Edinburgh, UK}
\EventLogo{}
\SeriesVolume{220}
\ArticleNo{8}

\begin{document}
	
\maketitle


\begin{abstract}
	
Some of the most relevant document schemas used online, such as XML and JSON, have a nested format. In the last decade, the task of extracting data from nested documents over streams has become especially relevant. We focus on the streaming evaluation of queries with outputs of varied sizes over nested documents. We model queries of this kind as Visibly Pushdown Transducers (VPT), a computational model that extends visibly pushdown automata with outputs and has the same expressive power as MSO over nested documents. Since processing a document through a VPT can generate a massive number of results, we are interested in reading the input in a streaming fashion and enumerating the outputs one after another as efficiently as possible, namely, with constant-delay. This paper presents an algorithm that enumerates these elements with constant-delay after processing the document stream in a single pass. Furthermore, we show that this algorithm is worst-case optimal in terms of update-time per symbol and memory usage. 


\end{abstract}

\section{Introduction}\label{sec:intro}

Streaming query evaluation~\cite{altinel2000efficient,babcock2002models} is the task of processing queries over data streams in one pass and with a limited amount of resources. This approach is especially useful on the web, where servers share data, and they have to extract the relevant content as they receive it. For structuring the data, the de facto structure on the web are nested documents, like XML or JSON. For querying, servers use languages designed for these purposes, like XPath, XQuery, or JSON query languages.
As an illustrative example, suppose our data server (e.g. Web API) is continuously receiving XML documents of the form:
\vspace{-1mm}

\begin{center}
\texttt{
<doc> <a> <b/> <c/> <b/> </a> <c> <b/> <b/> </c> </doc> ...}
\end{center}
\vspace{-1mm}
and for each document it has to evaluate the query $\cQ = //a/b$ (i.e., to extract all $b$-tags that are surrounded by an $a$-tag). The streaming query evaluation problem consists on reading these documents and finding all $b$-tags without storing the entire document on memory, i.e., by making one pass over the data and spending constant time per tag. In our example, we need to retrieve the 3rd and 5th tag as soon as the last tag $\texttt{</doc>}$ is received. One could consider here that the server has to read an infinite stream and perform the query evaluation continuously, where it must enumerate partial outputs as one of the XML documents ends.

Researchers have studied the streaming query evaluation problem in the past, focusing on reducing the processing time or memory usage (see, e.g. \cite{BarYossefFJ07}). Hence, they spent less effort on understanding the enumeration time of such a problem, regarding delay guarantees between outputs.
Constant-delay enumeration is a new notion of efficiency for retrieving outputs~\cite{DurandG07,Segoufin13}.
Given an instance of the problem, an algorithm with constant-delay enumeration performs a preprocessing phase over the instance to build some indices and then continues with an enumeration phase. It retrieves each output, one-by-one, taking a delay that is constant between any two consecutive outcomes. These algorithms provide a strong guarantee of efficiency since a user knows that, after the preprocessing phase, {she} will access the output as if {the algorithm had} already computed {it}. These techniques have attracted researchers' attention, finding sophisticated solutions to several query evaluation problems~\cite{BaganDG07,BerkholzGS20,Bagan06,AmarilliBJM17,FlorenzanoRUVV20,AmarilliBMN19}. 

In this work, we investigate the streaming query evaluation problem over nested documents by including enumeration guarantees, like constant-delay. We study the evaluation of queries given by Visibly Pushdown Transducers (\vpts) over nested documents.  These machines are the natural ``output extension'' of visibly pushdown automata, 
{and have the same expressive power as MSO over nested documents. 
In particular, \vpts can define queries like $\cQ$ above or any fragment of query languages for XML or JSON included in MSO.}
Therefore, \vpts allow considering the streaming query evaluation from a more general perspective, without getting married to a specific language (e.g., XPath). 

We study the evaluation of \vpt over a nested document in a streaming fashion. Specifically, we want to find a streaming algorithm that reads the document sequentially and spends constant time per input symbol. 
Furthermore, whenever needed, the algorithm can enumerate all outputs with output-linear delay. 
The main contribution of the paper is an algorithm with such characteristics for the class of I/O-unambiguous \vpts. We can extend this algorithm by determinization to all \vpts (i.e., in data complexity). 
Regarding memory consumption, we bound the amount of memory used in terms of the nesting of the document and the output weight. We show that our algorithm is worst-case optimal in the sense that there are instances where the maximum amount of memory required by any streaming algorithm is at least one of these two measures. 


Our main result applies to the streaming evaluation of XML and JSON query languages. 
In the appendix, we also show an application in the context of document spanners~\cite{FaginKRV15}.

\smallskip
\noindent \textbf{Related work.} The problem of streaming query evaluation has been extensively studied in the last decades. Some work considered streaming verification, like schema validation~\cite{SegoufinV02} or type-checking~\cite{KumarMV07}, where the output is true or false. Other proposals~\cite{ChenDZ06,OlteanuFB04,JosifovskiFB05,GreenGMOS04,Olteanu07} provided streaming algorithms for XPath or XQuery's fragments; however, extending them for reaching constant-delay enumeration seems unlikely. Furthermore, most of these works~\cite{KumarMV07,GouC07,GauwinNT09} assumed outputs of fixed size (i.e., tuples). People have also considered other aspects of streaming evaluation with outputs like earliest query answering~\cite{GauwinNT09} or bounded delay~\cite{DBLP:conf/lata/GauwinNT09} (i.e., given the first visit of a node, find the earliest event that permits its selection). These aspects are orthogonal to the problem studied here.
Another line of research is~\cite{BarYossefFJ05,BarYossefFJ07}, which presents space lower bounds for evaluating fragments of XPath or XQuery over streams. These works do not consider restrictions on the delay to give outputs.  

Visibly pushdown automata~\cite{AlurM04} {are} a model usually used for streaming evaluation of boolean queries~\cite{KumarMV07}. In~\cite{FiliotGRS19,AlurFMRS20}, authors studied the evaluation of \vpt in a streaming fashion, but none of them {saw} enumeration problems. Other extensions~\cite{GauwinNR08} for streaming evaluation have been analyzed but restricted to fixed-size outputs, and constant-delay was not included. 

Constant-delay algorithms have been studied for several classes of query languages and structures~\cite{Segoufin13}, as we already discussed. In~\cite{Bagan06,AmarilliBJM17}, researchers considered query evaluation over trees (i.e., a different representation {for} nested documents), but their algorithms are for offline evaluation and {it is not clear how to extend this algorithm for the online setting}. 
This research is extended with updates in~\cite{AmarilliBMN19pods}, which can encode streams by inserting new data items to the left. However, their update-time is logarithmic, and our proposal can do it with constant time~(i.e., in data complexity).
Furthermore, to the best of our knowledge it is unclear how to modify the work in~\cite{AmarilliBMN19pods} to get constant update-time in our scenario.
Streaming evaluation with constant-delay enumeration was included in the context of dynamic query evaluation~{\cite{IdrisUV17,BerkholzKS17,NikolicO18,0002NOZ20}} or complex event processing~\cite{GrezRU19,GrezR20}. In both cases, the input cannot encode nested documents, and their results do not apply.

\section{Preliminaries}\label{sec:prelim}



%
%
%
%

\noindent \textbf{Well-nested words and streams.} As usual, given a set $\Sigma$ we denote by $\Sigma^*$ all finite words with symbols in~$\Sigma$ where $\epsilon \in \Sigma^*$ represents the empty word of length $0$.

We will work over a {\em structured alphabet} $\Sigma = (\opS, \clS, \noS)$ comprised of three disjoint sets $\opS$, $\clS$, and $\noS$ that contain {\it open}, {\it close}, and {\it neutral} symbols respectively (in \cite{AlurM04,FiliotRRST18} these sets  are named \emph{call}, \emph{return}, and \emph{local}, respectively). 
Furthermore, we will denote symbols in $\opS$, $\clS$ or $\noS$ by $\op{a}$, $\cl{a}$, and~$a$, respectively.
Instead, we will use $s$ to denote any symbol in $\opS$, $\clS$, or $\noS$.
The set of {\em well-nested words}
over $\Sigma$, denoted as $\wnS$, is defined as the closure of the following rules: 
$\noS \cup \{\eps\} \subseteq \wnS$,
if $w_1, w_2 \in \wnS\setminus\{\eps\}$ then $w_1 \cdot w_2 \in \wnS$, and if $w \in \wnS$ and $\op{a} \in \opS$ and $\cl{b} \in \clS$ then $\op{a}\cdot w\cdot\cl{b} \in \wnS$. 
In addition, we will work with prefixes of well-nested words, that we call {\em prefix-nested words}. We denote the set of prefixes of $\wnS$ as $\pwnS$.
Also, we will sometimes use $w[i]$ to refer to the $i$-th symbol in a word $w$.

A {\em stream} $\Stream = s_1 s_2 \cdots$ is an infinite sequence where $s_i\in \opS \cup \clS \cup \noS$. Given a stream $\Stream = s_1 s_2 \ldots$ and positions $i,j\in\nat$ such that $i \leq j$, the word $\Stream[i,j]$ is the sequence $s_i \cdots s_j$. We also use this notation to refer to subsequences of infinite sequences that are not composed of symbols in $\Sigma$. For a stream $\Stream$, we will always assume that for each $i \in \nat$, the word $\Stream[1,i]$ is a prefix of some nested word (i.e., it can be completed to form a nested word). We also consider a method $\yield{\Stream}$ which can be called to access each element of $\Stream$ sequentially.

\smallskip
\noindent\textbf{Visibly pushdown automata.} A {\em visibly pushdown automaton}~\cite{AlurM04} (\vpa) is a tuple $\cA = (Q, \Sigma, \Gamma, \Delta, \qinit, F)$ where $Q$ is a finite set of states, $\Sigma = (\opS, \clS, \noS)$ is the input alphabet, $\Gamma$ is the stack alphabet, $\Delta \subseteq (Q\times\opS\times Q\times\Gamma)\cup(Q\times\clS\times\Gamma\times Q)\cup(Q\times\noS\times Q)$ is the transition relation, $\qinit \subseteq Q$ is a set of initial states, and $F\subseteq Q$ is a set of final states.
A transition $(q,\op{a},q',\gamma)$ is a {\em push-transition} where on reading $\op{a}\in\opS$, $\gamma$ is pushed onto the stack and the current state switches from $q$ to $q'$. Conversely, $(q,\cl{a},\gamma,q')$ is a {\em pop-transition} where on reading $\cl{a}\in\clS$ from the input and $\gamma$ from the top of the stack, the current state changes from $q$ to $q'$, and the symbol $\gamma$ is popped. Lastly, we say that $(q,a,q')$ is a {\em neutral transition} if $a\in\noS$, where there is no stack operation.

A stack is a finite sequence $\sigma$ over $\Gamma$ where the top of the stack is the first symbol on~$\sigma$. For a well-nested word $w = s_1 \cdots s_n$ in $\wnS$, a run of $\cA$ on $w$ is a sequence $\rho = (q_1,\sigma_1) \trans{s_1} \ldots \trans{s_n} (q_{n+1},\sigma_{n+1})$, where each $q_i \in Q$, $\sigma_i\in\Gamma^{*}$, $q_1\in \qinit$, $\sigma_1 = \eps$, and for every $i\in[1,n]$ the following holds: 
(1) if $s_{i}\in\opS$, then there is $\gamma\in\Gamma$ such that $(q_i,s_{i},q_{i+1},\gamma) \in \Delta$ and $\sigma_{i+1} = \gamma\sigma_i$, 
(2) if $s_{i}\in\clS$, then there is $\gamma\in\Gamma$ such that $(q_i,s_{i},\gamma,q_{i+1}) \in \Delta$ and $\sigma_{i} = \gamma\sigma_{i+1}$, and
(3) if $s_{i}\in\noS$, then $(q_i,s_{i},q_{i+1}) \in \Delta$ and $\sigma_{i+1} = \sigma_i$.
A run $\rho$ is accepting if $q_{n+1}\in F$. A well-nested word $w\in\wnS$ is accepted by a \vpa $\cA$ if there is an accepting run of $\cA$ on $w$. The language $\cL(\cA)$ is the set of well-nested words accepted by $\cA$. Note that if $\rho$ is an accepting run of $\cA$ on a well-nested word $w$, then $\sigma_{n+1} = \eps$. A set of well-nested words $\L \subseteq \wnS$ is called a visibly pushdown language if there exists a \vpa $\cA$ such that $\L = \cL(\cA)$.

A \vpa $\cA = (Q, \Sigma, \Gamma, \delta, \qinit, F)$ is said to be \emph{deterministic} if $|\qinit| =1$ and $\delta$ is a function subset of $(Q\,\times\opS \to Q\times \Gamma) \cup
(Q\times\clS\times\Gamma \to Q)\cup
(Q\times\noS \to Q)$. We also say that $\cA$ is \emph{unambiguous} if, for every $w \in \cL(\cA)$, there exists exactly one accepting run of $\cA$ on $w$. In~\cite{AlurM04}, it is shown that for every \vpa there exists an equivalent deterministic \vpa of at most exponential size. 

\smallskip
\noindent\textbf{Model of computation.} 
As it is common in the enumeration algorithms literature~\cite{Bagan06,Courcelle09,Segoufin13}, for our algorithms we assume the computational model of Random Access Machines (RAM) with uniform cost measure, and addition and subtraction as basic operations~\cite{AhoHU74}. We assume that a RAM has read-only input registers where the machine places the input, read-write work registers where it does the computation, and write-only output registers where it gives the output (i.e., the enumeration of the results).

\section{Streaming evaluation with output-linear delay}\label{sec:streamenum}


We are interested in defining a notion of a streaming enumeration problem: to evaluate a query over a stream and to enumerate the outputs with bounded delay whenever there is such. Towards this goal, we want to restrict the amount of resources used (i.e., time and space) and impose strong guarantees on the delay. As our gold standard, we consider the notion of \emph{output-linear delay} defined in~\cite{FlorenzanoRUVV20}. This notion is a refinement of the definition of constant-delay~\cite{Segoufin13} or linear-delay~\cite{Courcelle09} enumeration that better fits our purpose. Altogether, our plan for this section is to define a streaming enumeration problem and then provide a notion of efficiency that a solution for this problem should satisfy. 

We adopt the setting of relations to formalize a streaming enumeration problem~\cite{jerrum1986random,ArenasCJR19}. First, we need to define what is an enumeration problem outside the stream setting. 
Let $\inpAlph$  be an alphabet. An enumeration problem is a relation $R\subseteq (\inpAlph^{*} \times \inpAlph^*) \times\inpAlph^{*}$. For each pair $((q, x), y) \in R$ we view $(q, x)$ as the input of the problem and $y$ as a possible output for $(q, x)$. 
Furthermore, we call $q$ the query and $x$ the data.
This separation allows for a fine-grained analysis of the query complexity and data complexity of the problem. 
For an instance $(q, x)$ we define the set $\sem{q}_{R}(\inpw) = \{\outw \mid ((q,\inpw), \outw)\in R\}$ of all outputs of evaluating $q$ over $x$.

A streaming enumeration problem is an extension of an enumeration problem $R$ where the input is a pair $(q, \Stream)$ such that $\Stream$ is an infinite sequence of elements in $\Omega$. 
We identify two ways of extending an enumeration problem $R$ that differ in the output sets that are desired at each position in the stream:
\begin{enumerate}
\item The {\em streaming full-enumeration problem} for $R$ is one where the objective is to enumerate the set $\sem{q}_{R}(\Stream[1,n])$ at each position $n \geq 1$. 
\item A {\em streaming $\Delta$-enumeration problem} for $R$ is one where the objective is to enumerate the set $\sem{q}^{\Delta}_{R}(\Stream[1,n]) = \sem{q}_{R}(\Stream[1,n]) \setminus \bigcup_{i < n} \sem{q}_{R}(\Stream[1,i])$ at each position $n \geq 1$. 
\end{enumerate}
These versions give us two different ways of returning the outputs. These notions have been studied previously in the context of incremental view maintenance~\cite{ChirkovaY12} and more recently, for dynamic query evaluation~\cite{IdrisUV17,BerkholzKS17}.
For the sake of simplification, in the following we provide all definitions for the full-enumeration scenario. All definitions can be extended to $\Delta$-enumeration by changing $\sem{q}_{R}$ to $\sem{q}^{\Delta}_R$.

We turn now to our notion of efficiency for solving a streaming enumeration problem. Let $f\colon\nat\to\nat$. We say that $\enumE$ is a {\em streaming evaluation algorithm} for $R$ with $f$-{\em update-time} if $\enumE$ operates in the following way: it receives a query $q$ and reads the stream $\Stream$ by calling the $\yield{\Stream}$ method sequentially. After the $n$-th call to $\yield{\Stream}$, the algorithm processes the $n$-th data symbol in two phases:
\begin{itemize}
	\item In the first phase, called the {\em update} phase, the algorithm updates a data structure $D$ with the read symbol and the time spent is bounded by $\cO(f(\vert q \vert))$.
	\item The second phase, called the {\em enumeration} phase, occurs immediately after each update phase and outputs $\sem{q}_{R}(\Stream[1,n])$ using $D$. During this phase the algorithm:
	(1) writes $\# \outw_1 \# \outw_2\#\cdots\#\outw_m\#$ to the output registers where \# is a distinct separator symbol not contained in $\inpAlph$, and $\outw_1,\outw_2,\ldots,\outw_m$ is an enumeration (without repetitions) of the set $\sem{q}_{R}(\Stream[1,n])$,
	(2) it writes the first \# as soon as the enumeration phase starts, and 
	(3) it stops immediately after writing the last \#.
\end{itemize}
The purpose of separating $\enumE$’s operation into an update and enumeration phase is to make an output-sensitive analysis of $\enumE$’s complexity. Moreover, from a user perspective, this separation allows running the enumeration phase without interrupting the update phase. That is, the user could execute the enumeration phase in a separate machine, and its running time only depends on how many outputs she wants to enumerate.  

For the enumeration phase, we measure the delay between two outputs as follows: For an input $\inpw\in\inpAlph^*$, let $\# \outw_1 \# \outw_2\#\cdots\#\outw_m\#$ be the output of the algorithm during any call to the enumeration phase. Furthermore, let $\outtime_i(\inpw)$ be the time in the enumeration phase when the algorithm writes the $i$-th $\#$ when running on $\inpw$ for $i \leq m+1$. 
Define $\outdelay_i(\inpw) = \outtime_{i+1}(\inpw) - \outtime_{i}(\inpw)$ for $i \leq m$. 
Then we say that $\enumE$ has {\em output-linear delay} if there exists a constant $k$ such that for every $\inpw\in\inpAlph^*$ and $i \leq m$ it holds that $\outdelay_i(\inpw) \leq k \cdot |y_i|$. In other words, the number of instructions executed by $\enumE$ between the time that the $i$-th and the $(i+1)$-th \# are written is linear on the size of $y_i$.
Note that, in particular, an output-linear delay implies that the enumeration phase ends in constant time if there is no output for enumerating. 

As the last ingredient, we define how to measure the memory space of a streaming evaluation. Note that after the $n$-th call a streaming evaluation algorithm with $f$-update time will necessarily use at most $\cO(n\cdot f(|q|))$ bits of space. As a refinement of this bound, we say that this algorithm uses $g$-space over a query $q$ and stream $\Stream$ if the number of bits used by it after the $n$-th call is in $\cO(g(|q|, \Stream[1, n]))$.

Given a streaming enumeration problem, we say that it can be solved with update-time $f$, output-linear delay, and in $g$-space if there exists an algorithm such as the one described above. For $\Delta$-enumeration, the notion of streaming evaluation algorithm also applies, even though it could be the case that one can find such an algorithm for full-enumeration but not for $\Delta$-enumeration, and vice versa. 
Finally, the enumeration problem and solutions provided here are a formal refinement of the algorithmic notions proposed in the literature of streaming evaluation~\cite{GauwinNT09}, dynamic query evaluation~\cite{BerkholzKS17,IdrisUV17}, and complex event processing~\cite{GrezRU19,GrezR20}.

\section{Visibly pushdown transducers and main result}\label{sec:vpawo}

In this section, we present the definition of visibly pushdown transducers~\cite{FiliotRRST18} (\vpt), which are an extension of visibly pushdown automata to produce outputs. We use \vpt as our computational model to represent queries with output. This model is general enough to include any query language for nested documents, like XML or JSON, whose expressive power is in MSO.
After the setting is formalized, we state the main result of the paper. 

A \emph{\vptname} (\vpt) is a tuple $\cT = (Q, \Sigma, \Gamma, \oalph, \Delta, \qinit, F)$ where $Q$, $\Sigma$, $\Gamma$, $\qinit$, and $F$ are the same as for \vpa, $\oalph$ is the output alphabet with $\eps \notin \oalph$, and 
$
\Delta \subseteq  
(Q \times \opS \times (\oalph \cup \{\eps\}) \times Q \times \Gamma)  \cup 
(Q \times \clS \times (\oalph \cup \{\eps\}) \times \Gamma \times Q)  \cup 
(Q \times \noS \times (\oalph \cup \{\eps\}) \times Q)
$
\noindent is the transition relation. As usual for transducers, a symbol $s \in \opS \cup \clS \cup \noS$ is an input symbol that the machine reads and $\oout \in \oalph \cup \{\eps\}$ is a symbol that the machine prints in an output tape. Furthermore, $\eps$ represents that no symbol is printed for that transition.
A run $\rho$ of $\cT$ over a well-nested word $w = s_1s_2\cdots s_n \in\wnS$ 
is a sequence of the form
$
\rho = (q_1, \sigma_1) \xrightarrow{s_1/\!\ooutscr_1} \ldots  \xrightarrow{s_n/\!\ooutscr_n} (q_{n+1}, \sigma_{n+1})
$
where $q_i \in Q$, $\sigma_i\in \Gamma^{*}$, $q_1 \in I$, $\sigma_1 = \eps$ and for every $i\in[1,n]$ the following holds:
(1)~if $s_{i}\in \opS$, then $(q_i, s_{i}, \oout_{i},q_{i+1},\gamma) \in \Delta$ for some $\gamma\in\Gamma$ and $\sigma_{i+1} = \gamma\sigma_i$,
(2)~if $s_{i}\in\clS$, then $(q_i, s_{i}, \oout_{i}, \gamma, q_{i+1}) \in \Delta$ for some $\gamma\in\Gamma$ and $\sigma_i = \gamma\sigma_{i+1}$, and
(3)~if $s_{i}\in\noS$, then $(p_i, s_{i}, \oout_{i},q_{i+1})\in \Delta$ and $\sigma_i = \sigma_{i+1}$. We say that the run is accepting if $q_{n+1}\in F$. 
We call a pair $(q_i, \sigma_i)$ a configuration of $\rho$. 
Finally, the output of an accepting run $\rho$ is defined as:
$
\out(\rho) =  \out(\oout_1,1)\cdot \ldots \cdot \out(\oout_n, n)
$
where $\out(\oout, i) = \eps$ when $\oout = \eps$ and $(\oout, i)$ otherwise. Note that in $\oout_1 \cdots \oout_n$ we use $\eps$ as a symbol, and in $\out(\rho)$ we use $\eps$ as the empty string. Given a \vpt $\cT$ and a $w \in\wnS$, we define the set $\sem{\cT}(w)$ of all outputs of $\cT$ over $w$ as:
$
\sem{\cT}(w) = \{ \out(\rho) \, \mid \, \text{$\rho$ is an accepting run of $\cT$ over $w$}\}.
$


Strictly speaking, our definition of \vpt is richer than the one studied in~\cite{FiliotRRST18}. In our definition of \vpt each output element is a tuple  $(\oout,i)$ where $\oout$ is the symbol and $i$ is the output position, where for a standard \vpt~\cite{FiliotRRST18} an output element is just the symbol $\oout$. 
The extension presented here is indeed important for practical applications like in document spanners~\cite{FlorenzanoRUVV20,AmarilliBMN19} or in XML query evaluation~\cite{BarYossefFJ05, ShalemB08}. 

A first reasonable question is to understand what is the expressive power of VPT, namely, as a formalism for non-boolean query evaluation over nested words. For the Boolean case, it was shown~\cite{AlurM04} that VPA describe the same class of queries as MSO over nested words, called $\msom$.
Formally, fix a structured alphabet $\Sigma$ and let $w \in \wnS$ be a word of length~$n$. 
We encode $w$ as a structure:
$$
\big(\, [1,n], \, \leq,\, \{P_a\}_{a\in\Sigma}, \, \sf{match} \, \big)
$$
where $[1,n]$ is the domain, $\leq$ 
is the total order over $[1, n]$, $P_a = \{i \mid w[i] = a\}$, and $\sf{match}$ is a binary relation over $[1,n]$ that corresponds to the matching relation of open and close symbols: $\textsf{match}(i, j)$ is true iff $w[i]$ is an open symbol and $w[j]$ is its matching close symbol.
By some abuse of notation, we also use $w$ to denote its corresponding logical structure. 
A $\msom$ formula $\varphi$ over $\Sigma$ is given by:
\[
\ \ \ \ \ \ \, \varphi \ :=\ P_a(x) \ \mid\  x \in X \ \mid \ x \leq y \ \mid \ \textsf{match}(x, y) \ \mid \ \neg \varphi \ \mid \ \varphi \vee \varphi \ \mid \ \exists x.\varphi \ \mid \ \exists X.\varphi 
\]
where $a\in\Sigma$, $x$ and $y$ are first-order variables and $X$ is a monadic second order (MSO) variable. 
We write $\varphi(X_1, \ldots, X_n)$ where $X_1, \ldots, X_n$ are the free MSO variables of $\varphi$ (first-order variables are a special case of MSO variables). 
Then we write $w \models \varphi(A_1, \ldots, A_n)$ for $A_1, \ldots, A_n \subseteq [1, n]$ when $w$ satisfies $\varphi$ by replacing each variable $X_i$ with the set $A_i$. Here, we assume the standard semantics for MSO logic~\cite{libkin2004elements}. 

Given that VPT is an extension of VPA, it should not be a surprise that we can translate these results to VPT. In particular, the result in~\cite{AlurM04} can be easily extended to link VPT with formulas expressible in $\msom$. 

\begin{proposition}\label{prop:msovpt}
	Let $\varphi(X_1,\ldots,X_m)$ be a $\msom$ formula with $m$ free variables $X_1,\ldots,X_m$. 
	There is a VPT $\cT$ for which there is a one-to-one correspondence between the set 
	$\sem{\cT}(w)$ 
	and the set $\{(A_1,\ldots,A_m)\mid w\models \varphi(A_1,\ldots,A_m)\}$ for any word $w\in\wnS$.
	Moreover, for every VPT $\cT$ there is an $\msom$ formula $\varphi(X_1,\ldots,X_m)$ for which the same one-to-one correspondence holds.
\end{proposition}

In other words, VPT has the same expressive power as MSO over nested words. Given that fragments of query languages over nested documents (e.g., navigational XPath~\cite{CateM07}, JSON Navigational Logic~\cite{BourhisRV20}) are usually included in MSO, this shows that VPT is an expressive formalism for query evaluation over nested documents. 
As an example, in the appendix we show how to translate some XPath queries into \vpt, including $\cQ$.

%
%

We say that a \vpt $\cT = (Q, \Sigma, \Gamma, \oalph, \Delta, \qinit, F)$ is \emph{input/output deterministic} (I/O-deterministic for short) if $|\qinit| = 1$ and $\Delta$ is a partial function of the form $\Delta: (Q\times\opS \times \oalph\to Q\times \Gamma) \cup 
(Q\times\clS\times \oalph\times\Gamma \to Q) \cup
(Q\times\noS\times \oalph \to Q)$.
On the other hand, we say that $\cT$ is \emph{input/output unambiguous} (I/O-unambiguous for short) if for every $w\in \wnS$ and every $\mu \in \sem{\cT}(w)$ there is exactly one accepting run $\rho$ of $\cT$ over $w$ such that $\mu = \out(\rho)$. 
Notice that an I/O-deterministic \vpt is also I/O-unambiguous and in both models for each output there exists at most one run. The definition of I/O-deterministic is in line with the notion of I/O-deterministic variable automata of~\cite{FlorenzanoRUVV20} and I/O-unambiguous is a generalization of this idea that is enough for the purpose of our enumeration algorithm. One can show that for every \vpt $\cT$ there exists an equivalent I/O-deterministic \vpt and, therefore, an equivalent I/O-unambiguous \vpt. 
\begin{lemma}\label{vpawo:det}
	For every \vpt $\cT$ there exists an I/O-deterministic \vpt $\cT'$ of size $\cO(2^{|Q|^2|\Gamma|})$ such that $\br{\cT}(w) = \br{\cT'}(w)$ for every $w\in\wnS$.
\end{lemma}

In this paper, we are interested on the following streaming enumeration problem for \vpt. Let $\pazocal{C}$ be a class of \vpt (e.g. I/O-deterministic \vpt). 
\begin{center}
	\framebox{
		\begin{tabular}{rl}
			\textbf{Problem:} & $\enumvpt[\pazocal{C}]$\\
			\textbf{Input:} & a \vpt $\cT \in \pazocal{C}$ and $w\in \wnS$ \\
			\textbf{Output:} & Enumerate $\sem{\cT}(w)$
		\end{tabular}
	}
\end{center}
The main result of the paper is that for the class of I/O-unambiguous \vpt, the streaming full-enumeration version of this problem can be solved efficiently. 
\begin{theorem}\label{theo:main}
	The streaming full-enumeration problem of $\enumvpt$ for the class of I/O-unambiguous \vpt can be solved with update-time $\cO(|Q|^2\vert\Delta\vert)$ and output-linear delay. For the general class of \vpt, it can be solved with update-time $\cO(2^{|Q|^2\vert\Delta\vert})$ and output-linear \nolinebreak delay. 
\end{theorem} 

The result for the class of all \vpt is a consequence of Lemma~\ref{vpawo:det} and the enumeration algorithm for I/O-unambiguous \vpt (see Section~\ref{sec:ds} and~\ref{sec:eval}). For both cases, if the \vpt is fixed (i.e., in data complexity), then the update-time of the streaming algorithm is constant. 

For the streaming version of $\enumvpt$, one can have $\Delta$-enumeration with a small loss of efficiency by solving the full-enumeration problem. Specifically, one can show that for any I/O-unambiguous \vpt $\cT$ there is an I/O-unambiguous \vpt $\cT'$ of linear size with respect to $|\cT|$ such that $\br{\cT'}(w) = \br{\cT}(w) \setminus \bigcup \, \{ \br{\cT}(w[1, i]) \mid i < |w|, w[1,i] \in\wnS\}$ for each $w\in\wnS$.

\begin{theorem}\label{vpawo:deltamain}
	The streaming $\Delta$-enumeration problem of $\enumvpt$ for the class of I/O-unambiguous \vpt can be solved with update-time $\cO(|Q|^2\vert\Delta\vert)$ and output-linear delay. For the general class of \vpt, it can be solved with update-time $\cO(2^{|Q|^2\vert\Delta\vert})$ and output-linear \nolinebreak delay. 
\end{theorem}

We could have considered a more general definition of VPT to produce outputs for prefix-nested words. This would be desirable for having some sort of \emph{earliest query answering}~\cite{GauwinNT09} which is important in practical scenarios. We remark that the algorithm of Theorem~\ref{theo:main} can be extended for this case at the cost of making the presentation more complicated. For the sake of presentation, we defer this extension to the full version of this paper.

\smallskip
\noindent \textbf{Space lower bounds of evaluating a \vpt.}  
This subsection deals with the space used by the streaming evaluation algorithm of Theorem~\ref{theo:main}. 
Indeed, this algorithm could use linear space in the worst case. In the following we explore some lower bounds in the space needed by any algorithm, and show that this bound is tight for a certain type of \vpt.

To study the minimum number of bits needed to solve $\enumvpt$ we need to introduce some definitions.
Fix a \vpt $\cT$ and $w \in \pwnS$. Let 
$\outgap(\cT, w)$ be the number of positions less than $|w|$ that appear in some output of $\sem{\cT}(w \cdot w')$ for some $w \cdot w'\in \wnS$. 
Furthermore, for a well-nested word $u$ let $\depth(u)$ be the maximum number of nesting pairs inside $u$, formally, $\depth(a) = 0$ for $a \in  \noS \cup \{\eps\}$, $\depth(u_1 \cdot u_2) = \max\{\depth(u_1), \depth(u_2)\}$, and $\depth(\op{a}\cdot u\cdot\cl{b}) = \depth(u) + 1$. For $w \in \pwnS$, we define $\depth(w) = \min\{\depth(w') \mid w \text{ is a prefix of } w'\}$.
We can now state some worst-case space lower bounds for \enumvpt{}.
\begin{proposition}\label{alg:spacebound}
	\begin{enumerate}
		\item There exists a \vpt $\cT$ such that every streaming evaluation algorithm for \enumvpt{} with input $\cT$ and $\Stream$ requires at least $\Omega(\depth(\Stream[1, n]))$ bits of space.
		\item 
		There exists a \vpt $\cT$ such that every streaming evaluation algorithm for \enumvpt{} with input $\cT$ and $\Stream$ requires at least $\Omega(\outgap(\cT, \Stream[1, n]))$ bits of space.
	\end{enumerate}
\end{proposition}

In~\cite{BarYossefFJ05,BarYossefFJ07}, the authors provide lower bounds on the amount of space needed for evaluating XPath in terms of the nesting and the concurrency (see~\cite{BarYossefFJ05} for a definition). 
One can show that  the \ogapname of $\cT$ and $w$ is always above the concurrency of $\cT$ and $w$. Despite this, one can check that both notions coincide for the space lower bound given in Proposition~\ref{alg:spacebound}.

The previous results show that, in the worst case, any streaming evaluation algorithm for VPT will require space of at least the depth of the document or the \ogapname.
To show that Theorem~\ref{theo:main} is optimal in the worst-case, we need to consider a further assumption of our \vpt. We say that a \vpt $\cT$ is \emph{trimmed}~\cite{caralp2015trimming} if for every $w\in \pwnS$ and every (partial) run $\rho$ of $\cT$ over $w$, there exists $w'$ and an accepting run $\rho'$ of $\cT$ over $w \cdot w'$ such that $\rho$ is a prefix of $\rho'$. 
This notion is the analog of trimmed non-deterministic automata. Similarly to Lemma~\ref{vpawo:det}, one can show that for every \vpt $\cT$ there exists a trimmed I/O-deterministic \vpt $\cT'$ equivalent to $\cT$ (i.e., by extending the construction in~\cite{caralp2015trimming} to VPT). 
The next result shows that, if the input to \enumvpt{} is a trimmed I/O-unambiguous \vpt, then the memory footprint is at most the maximum between the depth and \ogapname of the input. 

\begin{proposition}\label{prop:space}
	The streaming enumeration problem of \enumvpt\ for the class of trimmed I/O-unambiguous \vpt can be solved with update-time $\cO(|Q|^2\vert\Delta\vert)$, output-linear delay and $\cO(\max\{\depth(\Stream[1, n]), \outgap(\cT,\Stream[1, n])\}\times|Q|^2|\Delta|)$ space for every stream $\Stream$.
\end{proposition} 

Unfortunately, the algorithm provided in Theorem~\ref{theo:main} is not \emph{instance optimal}, in the sense of using the lowest number of bits needed for each specific \vpt (see the appendix).
Note that an instance optimal algorithm for the streaming enumeration problem of \vpts will imply a solution to the weak evaluation problem, stated by Segoufin and Vianu~\cite{SegoufinV02}. This is an open problem in the area (see \cite{Barloy21} for some recent results), so we leave this for future work. 


\section{Enumerable compact sets: a data structure for output-linear delay}\label{sec:ds}

This section presents a data structure, called \dsnamebigcaps{} (\dsabbr{}), which is the cornerstone of our enumeration algorithm for VPT. This data structure is strongly inspired by the work in~\cite{AmarilliBJM17,AmarilliBMN19}. Indeed, \dsabbr{} can be considered a refinement of the d-DNNF circuits used in~\cite{AmarilliBJM17} or of the set circuits used in~\cite{AmarilliBMN19}.
Several papers~\cite{OlteanuZ15,AmarilliBJM17,AmarilliBMN19pods,Torunczyk20}   have considered circuits-like structures for encoding outputs and enumerate them with constant delay. The novelty of ECS is twofold. First, we use ECS for solving a streaming evaluation problem. Although people have studied streaming query evaluation with enumeration before~\cite{IdrisUV17,BerkholzKS17}, this is the first work that uses a circuit-like data structure in an online setting. 
Second and more important, there is a difference in performance if we compare ECS to the previous approaches.
In offline evaluation, constant delay algorithms usually create an initial circuit from the input, making several passes over the structure, building indices, and then running the enumeration process. Given time restrictions for the online evaluation, we cannot create a circuit and do this linear-time preprocessing before enumerating. On the contrary, we must extend the circuit-like data structure for each data item in constant time and then be ready to start the enumeration. This requirement justifies the need for a new data structure for representing and enumerating outputs.  Therefore, ECS differs from previous proposals because each operation must take constant time, and we can run the enumeration process with output-linear delay, at any time and without any further preprocessing. In the following, we present \dsabbr{} step-by-step to use it later in the next section.


Let $\Sigma$ be a (possibly infinite) alphabet. We define an \emph{\dsnamebigcaps{}} (\dsabbr) as a tuple $\D = (\Sigma, V, I, \ell, r, \lambda)$ such that $V$ and $I \subseteq V$ are finite sets of nodes, $\ell\colon I\to V$ and $r\colon I\to V$ are the {\em left} and {\em right} functions, and $\lambda\colon V\to\Sigma\cup\{\cup,\odot\}$ is a label function such that $\lambda(v)\in\{\cup,\odot\}$ if, and only if, $v\in I$.
Further, we assume that the directed graph $(V, \{(v,\ell(v)),(v,r(v))\mid v\in V\})$ is acyclic.
We call the nodes in $I$ {\em inner nodes} and the nodes in $V \setminus I$ {\em leaves}. Furthermore, for $v\in I$ we say that $v$ is a {\em product node} if $\lambda(v) = \odot$, and a {\em union node} if $\lambda(v) = \cup$. We define the size of $\D$ as $|\D| = |V|$.
For each node $v$ in $\D$, we associate a set of words $\L_{\D}(v)$ recursively as follows: (1) $\L_{\D}(v) = \{a\}$ whenever $\lambda(v) = a\in\Sigma$, (2) $\L_{\D}(v) = \L_{\D}(\ell(v)) \cup \L_{\D}(r(v))$ whenever $\lambda(v) = \cup$, and (3) $\L_{\D}(v) = \L_{\D}(\ell(v)) \cdot \L_{\D}(r(v))$ whenever $\lambda(v) = \odot$,
where $L_1 \cdot L_2 = \{w_1\cdot w_2 \mid w_1\in L_1\text{ and }w_2\in L_2\}$.

The size $|\L_{\D}(v)|$ can be exponential with respect to $|\D|$. For this reason, we say that $\D$ is a \emph{compact} representation of $\L_{\D}(v)$ for any $v \in V$.
Although $\L_{\D}(v)$ is very large, the goal is to enumerate all of its elements efficiently. Specifically, we consider the following problem:
\vspace{.1cm}
\begin{center}
	\framebox{
		\begin{tabular}{rl}
			\textbf{Problem:} & $\enumds$\\
			\textbf{Input:} & An \dsabbr{} $\D = (\Sigma, V, I, \ell, r, \lambda)$ and $v \in V$.  \\
			\textbf{Output:} & Enumerate the set $\L_\D(v)$ without repetitions. \\
		\end{tabular}
	}
\end{center}
\vspace{.1cm}
Plus, we want to solve $\enumds$ with output-linear delay. To reach this goal we need to impose two additional restrictions on $\D$. The first restriction is to guarantee that $\D$ is not ambiguous, namely, for each $w \in \L_\D(v)$ there is at most one way to retrieve $w$ from~$\D$. Formally, we say that $\D$ is \emph{unambiguous} if $\D$ satisfies the following two properties: (1) for every union node $v$ it holds that $\L_{\D}(\ell(v))$ and $\L_{\D}(r(v))$ are disjoint, and (2) for every product node $v$ and for every $w \in \L_{\D}(v)$, there exists a unique way to decompose $w = w_1 \cdot w_2$ such that $w_1 \in \L_{\D}(\ell(v))$ and $w_2 \in \L_{\D}(r(v))$. Thus, if $\D$ is unambiguous, there will be no duplicates if we enumerate $\L_{\D}(v)$ directly, given that there is no way of producing the same element in two different ways. 

The second restriction is to guarantee that, for each node $v$, there exists an output or, more specifically, a symbol of an output {\em close} to~$v$, in the sense that it can be reached in a bounded number of steps. This is not always the case for an \dsabbr{}. For example, consider a balanced tree of union nodes where all the outputs are at the leaves. One has to traverse a logarithmic number of nodes from the root to reach the first output.
Note that product nodes do not pose this problem since the number of nodes that have to be traversed to produce a certain output is proportional to its length.
For this reason, we define the notion of \emph{$k$-bounded}~\dsabbr{}. 
Given an \dsabbr{} $\D$, define the (left) output-depth of a node $v\in V$, denoted by $\odepth_{\D}(v)$, recursively as follows:
$\odepth_{\D}(v) = 0$ whenever $\lambda(v)\in\Sigma$ or $\lambda(v) = \odot$, and $\odepth_{\D}(v) = \odepth_{\D}(\ell(v))+1$ whenever $\lambda(v) = \cup$.
Then, for a fixed $k\in\nat$ we say that $\D$ is $k$-bounded if $\odepth_{\D}(v)\leq k$ for all $v\in V$.

Given the definition of output-depth, we say that $v$ is an output node of $\D$ if $v$ is a leaf or a product node. 
Note that if $\D$ only has output nodes, then it is 0-bounded, and one can easily check that $\L_\D(v)$ can be enumerated with output-linear delay.
Indeed, for a fixed $k$ the same happens with every unambiguous and $k$-bounded \dsabbr{}.

\begin{proposition}\label{ds:lindelay}
	Fix $k\in\nat$. Let $\D = (\Sigma, V, I, \ell, r, \lambda)$ be an unambiguous and $k$-bounded \dsabbr{}. Then the set $\L_\D(v)$ can be enumerated with output-linear delay for any $v\in V$.
\end{proposition}

The enumeration algorithm above does not require any preprocessing over $\D$ and the main idea is to perform some sort of DFS traversal over the nodes. By this proposition, from now we assume that all \dsabbr{} are unambiguous and $k$-bounded for some fixed~$k$.


The next step is to provide a set of operations that allow extending an \dsabbr{} $\D$ while maintaining $k$-boundedness. Furthermore, we require these operations to be fully-persistent: a data structure is called \emph{fully-persistent} if every version can be both accessed and modified~\cite{driscoll1986making}. In other words, the previous version of the data structure is always available after each operation.
To satisfy the last requirement, the strategy will consist in extending $\D$ to $\D'$ for each operation, by always adding new nodes and maintaining the previous nodes untouched. Then $\cL_{\D'}(v) = \cL_{\D}(v)$ for each node $v \in V$, and so, the structure is fully-persistent. 

More precisely, fix an \dsabbr{} $\D = (\Sigma, V, I, \ell, r, \lambda)$. In the following, we say that $\D' = (\Sigma, V', I', \ell', r', \lambda')$ is an extension of $\D$ if, and only if, $\mathsf{obj} \subseteq \mathsf{obj}'$ for every $\mathsf{obj} \in \{V,I, \ell, r, \lambda\}$. Further, we write $\mathsf{op}(I) \rightarrow O$ to define the signature of an operation $\mathsf{op}$ where $I$ is the input and $O$ is the output. 
Then for any $a \in \Sigma$ and $v_1, \ldots, v_4 \in V$, we define the operations:
\[
\begin{array}{rclrclrcl}
\add(\D,a) & \!\!\!\!\rightarrow\!\!\!\! & (\D',v')   \ \ \ \ \ \	 & \prod(\D,v_1, v_2) & \!\!\!\! \rightarrow \!\!\!\! & (\D',v')   \ \ \ \ \ \ \ &  \union(\D,v_3, v_4) & \!\!\!\!\rightarrow\!\!\!\! & (\D',v')
\end{array}
\]
\noindent such that $\D'$ is an extension of $\D$  and $v' \in V' \setminus V$ is a fresh node such that $\cL_{\D'}(v') = \{a\}$, $\cL_{\D'}(v') = \cL_{\D}(v_1) \cdot \cL_{\D}(v_2)$, and $\cL_{\D'}(v') = \cL_{\D}(v_3) \cup \cL_{\D}(v_4)$, respectively.
We assume that the $\union$ and $\prod$ respect properties (1) and (2) of an unambiguous \dsabbr{}, that is, $\cL_{\D}(v_1)$ and $\cL_{\D}(v_2)$ are disjoint and, for every $w \in \L_{\D}(v_3) \cdot \L_{\D}(v_4)$, there exists a unique way to decompose $w = w_1 \cdot w_2$ such that $w_1 \in \L_{\D}(v_3)$ and $w_2 \in \L_{\D}(v_4)$. 

Next, we show how to implement each operation. In fact, the case of $\add$ and $\prod$ are straightforward. For $ \add(\D,a) \rightarrow (\D',v')$ define $V' := V \cup \{v'\}$, $I' := I$, and $\lambda'(v') = a$. One can easily check that $\L_{\D'}(v') = \{a\}$ as expected. For $\prod(\D,v_1, v_2) \rightarrow (\D',v')$ we proceed in a similar way: define $V' := V \cup\{v'\}$, $I' := I \cup\{v\}$, $\ell'(v') :=v_1$, $r'(v') = v_2$, and $\lambda'(v') = \odot$. 
Then $\L_{\D'}(v') = \L_{\D}(v_1)\cdot\L_{\D}(v_2)$.
Furthermore, one can check that each operation takes constant time, $\D'$ is a valid \dsabbr{} (i.e. unambiguous and $k$-bounded), and the operations are fully-persistent (i.e. the previous version $\D$ is available).

To define the union, we need to be a bit more careful to guarantee output-linear delay, specifically, the  $k$-bounded property.
For a node $v\in V$, we say that $v$ is {\em safe} if (1) $\odepth_{\D}(v) \leq 1$, and (2) if $\odepth_{\D}(v) = 1$, then $\odepth_{\D}(r(v)) \leq 1$. In other words, $v$ is safe if $v$ is an output node, or its left child is an output node, and the right child is either an output node or has output depth 1.
Note that a leaf or a product node are safe nodes by definition and, thus, the $\add$ and $\prod$ operations always produce safe nodes. The trick then is to show that, if $v_3$ and $v_4$ are safe nodes, then we can implement $\union(\D,v_3, v_4) \rightarrow (\D',v')$ and produce a safe node $v'$. For this define $(\D',v')$ as follows:
\begin{itemize}
	\item  If $v_3$ or $v_4$ are output nodes then $V' := V\cup\{v'\}$, $I' := I\cup\{v'\}$, and $\lambda(v') := \cup$. Moreover, if $v_3$ is the output node, then $\ell'(v') := v_3$ and $r'(v') := v_4$. Otherwise, we connect $\ell'(v') := v_4$ and~$r'(v') := v_3$.
	\item If $v_3$ and $v_4$ are not output nodes (i.e. both are union nodes), then $V' := V \cup\{v',u_1,u_2\}$, $I' := I\cup\{v',u_1,u_2\}$, $\ell'(v') := \ell(v_3)$, $r'(v') := u_1$, and $\lambda'(v') := \cup$; $\ell'(u_1) := \ell(v_4)$, $r'(u_1) := u_2$, and $\lambda'(u_1) := \cup$; $\ell'(u_2) := r(v_3)$, $r'(u_2) := r(v_4)$, and $\lambda'(u_2) := \cup$.
\end{itemize}
This gadget is depicted in Figure~\ref{fig-union-def} (note that a similar trick is used in~\cite{AmarilliBJM17} for computing an index over a circuit). 
This construction has several properties.
First, one can easily check that $\L_{\D'}(v') = \L_{\D}(v_1)\cup\L_{\D}(v_2)$ and so the semantics is well-defined. 
Second, $\union$ can be computed in constant time in $|\D|$ given that we only need to add three fresh nodes, and the operation is fully-persistent given that we connect them without modifying~$\D$. 
Furthermore, the produced node $v'$ is safe in $\D'$, although nodes $u_1$ and $u_2$ are not necessarily safe. 
Finally, $\D'$ is 2-bounded whenever $\D$ is $2$-bounded. This is straightforward to see for first case when $v_3$ or $v_4$ are output nodes. For the second case (i.e., Figure~\ref{fig-union-def}), we have to notice that $v_3$ and $v_4$ are safe, therefore $\ell(v_3)$ and $\ell(v_4)$ are output nodes, and then $\odepth_{\D'}(v') = \odepth_{\D'}(u_1) = 1$. Further, given that $v_3$ is safe, we know that $\odepth_{\D}(r(v_3)) \leq 1$, so $\odepth_{\D'}(u_2)\leq 2$. Given that the output depths of all fresh nodes in $\D'$ are bounded by $2$ and $\D$ is $2$-bounded, then we conclude that $\D'$ is $2$-bounded as well.

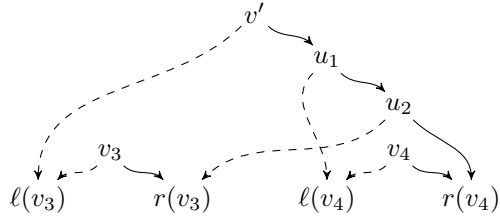
\begin{figure}[t]
	\centering
	\begin{tikzpicture}[->,>=stealth',roundnode/.style={circle,inner sep=1pt},squarednode/.style={rectangle,inner sep=2pt}, scale=0.95]
	\node[squarednode] (0) at (0, 0) {$\ell(v_3)$};
	\node[squarednode] (1) at (2, 0) {$r(v_3)$};
	\node[roundnode] (2) at (1, 0.65) {$v_3$};
	\node[squarednode] (3) at (4, 0) {$\ell(v_4)$};
	\node[squarednode] (4) at (6, 0) {$r(v_4)$};
	\node[roundnode] (5) at (5, 0.65) {$v_4$};
	\node[roundnode] (6) at (5, 1.3) {$u_2$};
	\node[roundnode] (7) at (4, 1.95) {$u_1$};
	\node[roundnode] (8) at (3, 2.6) {$v'$};
	\draw[dashed] (2) to[out=-135,in=45] (0);
	\draw (2) to[out=-45,in=135] (1);
	\draw[dashed] (5) to[out=-135,in=45] (3);
	\draw (5) to[out=-45,in=135] (4);
	\draw (6) to[out=-45,in=90] (4);
	\draw[dashed] (6) to[out=-135,in=45] (1);
	\draw (7) to[out=-45,in=135] (6);
	\draw[dashed] (7) to[out=-135,in=90] (3);
	\draw (8) to[out=-45,in=135] (7);
	\draw[dashed] (8) to[out=-135,in=90] (0);
	\end{tikzpicture}
	\caption{Gadget for $\union(\D,v_3,v_4)$. Nodes $v',u_1,u_2,v_3$ and $v_4$ are labeled as $\cup$. Dashed and solid lines denote the mappings in $\ell'$ and $r'$ respectively.}
	\label{fig-union-def}
	\vspace{-3mm}
\end{figure}

By the previous discussion, if we start with an \dsabbr\ $\D$ which is 2-bounded (or empty) and we apply the $\add$, $\prod$ and $\union$ operators between safe nodes (which also produce safe nodes), then the result is 2-bounded as well. Finally, by Proposition~\ref{ds:lindelay}, the result can be enumerated with output-linear delay.

\begin{theorem}\label{theo:data-structure}
	The operations $\add$, $\prod$, and $\union$ require constant time and are fully-persistent. Furthermore, if we start from an empty \dsabbr\ $\D$ and apply $\add$, $\prod$, and $\union$ over safe nodes, the partial results $(\D', v')$ satisfy that $v'$ is always a safe node and the set $\L_{\D'}(v)$ can be enumerated with output-linear delay for every node $v$.
\end{theorem}  

It is important to remark that restricting these operations only over safe nodes is a mild condition. Given that we will usually start from an empty \dsabbr\ and apply these operations over previously returned nodes, the whole algorithm will always use safe nodes during its computation, satisfying the conditions of Theorem~\ref{theo:data-structure}. 

For technical reasons, our algorithm of the next section needs a slight extension of \dsabbr{} by allowing leaves that produce the empty string $\eps$. Let $\eps\not\in\Sigma$ be a symbol representing the empty string (i.e. $w\cdot\eps = \eps\cdot w = w$). We define an \dsname{} with $\eps$ (called \dsepsabbr) as a tuple $\D = (\Sigma, V, I, \ell, r, \lambda)$ defined identically to an \dsabbr{} except that $\lambda:V\to\Sigma\cup\{\cup, \odot,\eps\}$ and $\lambda(v)\in\{\cup,\odot\}$ if, and only, if $v\in I$. Also, if $\lambda(v) = \eps$, then $\L_{\D}(v) = \{\epsilon\}$. The unambiguity restriction is the same for \dsepsabbr and one has to slightly extend $k$-boundedness to consider $\epsilon$-nodes. 
However, to support the $\prod$ and $\union$ operations in constant time and to maintain the $k$-boundedness invariant, we need to extend the notion of safe nodes (called $\epsilon$-safe) and the gadgets for $\prod$ and $\union$. 
Given space restrictions, we show these extensions in the appendix and state here the main result, that will be used in the next section.
 
\begin{theorem}\label{theo:data-structure-eps}
	The operations $\add$, $\prod$, and $\union$ over \dsepsabbr{} take constant time and are fully-persistent. Furthermore, if we start from an empty \dsepsabbr{} $\D$ and apply $\add$, $\prod$, and $\union$ over $\epsilon$-safe nodes, the partial results $(\D', v')$ satisfy that $v'$ is always an $\epsilon$-safe node and the set $\L_{\D'}(v)$ can be enumerated with output-linear delay for every node $v$.
\end{theorem}  


\section{Evaluating \vptnames with \lindelay}\label{sec:eval}

\newcommand{\aS}{S}
\newcommand{\aT}{T}
\newcommand{\dets}{{\operatorname{det}}}
\newcommand{\Runs}{\operatorname{Runs}}

The goal of this section is to describe an algorithm that takes an I/O-unambiguous \vpt $\cT$ plus a stream $\Stream$, and enumerates the set $\br{\cT}(\Stream[1,n])$ for an arbitrary $n\geq 0$ with $\cO(|Q|^2 |\Delta|)$-update-time and output-linear delay.
We divide the presentation of the algorithm into two parts. The first part explains the determinization of a VPA, which is instrumental in understanding our update phase. The second part gives the algorithm and proves its correctness. 
Given that a neutral symbol $a$ can be represented as a pair $\op{a} \cdot \cl{a}$, in this section we present the algorithm and definitions without neutral letters, that is, the structured alphabet is $\Sigma = (\opS, \clS)$. Thus, from now on we use $a$ for denoting any symbol in $\opS \cup \clS$. 

\smallskip
\noindent \textbf{Determinization of visibly pushdown automata.} 
A significant result in Alur and Madhusudan's paper~\cite{AlurM04} that introduces \vpa was that one can always determinize them. We provide here an alternative proof for this result that requires a somewhat more direct construction. This determinization process is behind our update algorithm and serves to give some crucial notions of how it works. We start by providing the determinization construction, introducing some useful notation, and then giving some intuition.

Given a VPA $\cA = (Q, \Sigma, \Gamma, \Delta, \qinit, F)$, we define the following deterministic \vpa $\cA^\dets = (Q^\dets, q_0^\dets, \Gamma^\dets, \delta^\dets, F^\dets)$ with state set $Q^\dets = 2^{Q\times Q}$ and stack symbol set $\Gamma^\dets = 2^{Q\times \Gamma \times Q}$. The initial state is $q_0^\dets = \{(q,q)\mid q\in I\}$ and the set of final states is $F^\dets = \{\aS \in Q^\dets \mid \aS \cap (I\times F) \neq \emptyset\}$. 
Finally, we define the transition function $\delta^\dets$ such that if $\op{a}\in\opS$, then $\delta^\dets(\aS,\op{a}) = (\aS',\aT')$ where
$\aS' =  \{(q,q) \mid \exists p, p', \gamma. \ (p,p')\in \aS \, \wedge \,  (p',\op{a},q,\gamma)\in\Delta\}$ and $\aT' = \{ (p,\gamma, q) \mid \exists p'. \ (p,p')\in \aS  \, \wedge \,  (p',\op{a},q,\gamma)\in\Delta\}$; if $\cl{a}\in\clS$, then $\delta^\dets(\aS,\aT,\cl{a}) = \aS'$ where
$\aS' = \{ (p,q) \mid \exists p',q', \gamma. \ (p,\gamma,p')\in \aT \, \wedge \, (p',q')\in \aS \text \, \wedge \, (q',\cl{a},\gamma, q)\in\Delta \}$.


\newcommand{\clevel}{{\sf currlevel}}
\newcommand{\llevel}{{\sf lowerlevel}}

To understand the purpose of this construction, first we need to introduce some notation.
Fix a well-nested word $w = a_1 a_2 \cdots a_n$. A span $s$ of $w$ is a pair $\spanc{i}{j}$ of natural numbers $i$ and~$j$ with $1 \leq i \leq j \leq n+1$. 
We denote by $w[i,j\rangle$ the subword $a_i\cdots a_{j-1}$ of~$w$ and, when $i = j$, we assume that $w[i,j\rangle = \eps$.
Intuitively, spans are indexing $w$ with intermediate positions, like $\underset{\tiny\texttt 1}{}a_1 \underset{\tiny\texttt 2}{}a_2 \underset{\tiny\texttt 3}{} \ldots \underset{\tiny\texttt n}{} a_n \underset{\tiny\texttt{n+1}}{}$, where $i$ is between symbols $a_{i-1}$ and $a_i$. Then $\spanc{i}{j}$ represents an interval $\{i, \ldots, j\}$ that captures the subword $a_i \ldots a_{j-1}$.

Now, we say that a span $\spanc{i}{j}$ of $w$ is well-nested if $w\spanc{i}{j}$ is well-nested. Note that $\eps$ is well-nested, so $\spanc{i}{i}$ is a well-nested span for every $i$. 
For a position $k\in[1,n+1]$, we define the {\em current-level span} of $k$, $\clevel(k)$, as the well-nested span $\spanc{j}{k}$ such that $j = \min\{j' \mid \spanc{j'}{k} \text{ is well-nested}\}$. Note that $\spanc{k}{k}$ is always well-nested and thus $\clevel(k)$ is well defined. 
We also identify the {\em lower-level span} of $k$, $\llevel(k)$, defined as $\llevel(k) = \clevel(j-1) = \spanc{i}{j-1}$ whenever $\clevel(k) = \spanc{j}{k}$ and $j > 1$. 
In contrast to $\clevel(k)$, $\llevel(k)$ is not always well-defined given that it is ``one level below'' than $\clevel(k)$ and this may not exist. More concretely, for $\clevel(k) = \spanc{j}{k}$ and  $\llevel(k) = \spanc{i}{j-1}$, 

\noindent these spans will look as follows:
$$
\underset{\tiny\texttt 1}{} a_1 
\underset{\tiny\texttt 2}{} a_2 
\underset{\tiny\texttt 3}{} \ldots 
\op{a}_{i-1} \underset{\tiny\texttt{i}}{}  
\overbrace{a_{i} \ldots a_{j-2}}^{\llevel(k)} 
\underset{\tiny\texttt{j-1}}{} \op{a}_{j-1} 
\underset{\tiny\texttt{j}}{} 
\overbrace{a_{j} \ldots  a_{k-1}}^{\clevel(k)} 
\overset{\text{\LARGE $\downarrow$}}{\underset{\tiny\texttt{k}}{}} a_{k} \ldots
\underset{\tiny\texttt n}{} a_n \underset{\tiny\texttt{n+1}}{}
$$
As an example, consider the word $ 
\underset{\tiny\texttt 1}{} {\texttt (} \, 
\underset{\tiny\texttt 2}{} {\texttt (} \, 
\underset{\tiny\texttt 3}{} {\texttt )} \, 
\underset{\tiny\texttt 4}{} {\texttt (} \, 
\underset{\tiny\texttt 5}{} {\texttt (} \, 
\underset{\tiny\texttt 6}{} {\texttt )} \, 
\underset{\tiny\texttt 7}{} {\texttt )} \, 
\underset{\tiny\texttt 8}{} {\texttt )} \,
\underset{\tiny\texttt 9}{}$. 
The only well-nested spans besides the ones of the form $\spanc{i}{i}$ are $\spanc{1}{9}$, $\spanc{2}{4}$, $\spanc{2}{8}$, $\spanc{4}{8}$ and $\spanc{5}{7}$, therefore $\clevel(8) = \spanc{2}{8}$, and $\llevel(7) = \spanc{2}{4}$.

We are ready to explain the purpose of the determinization above. Let $w = a_1 a_2 \cdots a_n$ be a well-nested word and $\rho^\dets = (\aS_1, \tau_1) \trans{a_1} \ldots \trans{a_{k-1}} (\aS_k, \tau_k)$ be the (partial) run of $\cA^\dets$ until some $k$.
Furthermore, assume $\tau_k = \aT_k \cdot \tau$ for some $\aT_k \in \Gamma^\dets$ and $\tau \in (\Gamma^\dets)^*$.  
The connection between $\rho^\dets$ and the runs of $\cA$ over $a_1\ldots a_{k-1}$ is given by the following invariants:
\begin{enumerate}
	\item[(a)] $(p,q) \in \aS_k$ if, and only if, there exists a run $(q_1, \sigma_1) \trans{a_1} \ldots \trans{a_{k-1}} (q_k, \sigma_k)$ of $\cA$ over $a_1 \ldots a_{k-1}$ such that $q_j = p$, $q_k=q$, and $\clevel(k) = \spanc{j}{k}$. 
	\item[(b)] $(p,\gamma, q) \in \aT_k$ if, and only if, there exists a run $(q_1, \sigma_1) \trans{a_1} \ldots \trans{a_{k-1}} (q_k, \sigma_k)$ of $\cA$ over $a_1 \ldots a_{k-1}$ such that $q_i = p$, $q_j=q$, $\sigma_k = \gamma \sigma$ for some $\sigma$, and $\llevel(k) = \spanc{i}{j-1}$. 
\end{enumerate}
On one hand, (a) says that each pair $(p,q)\in \aS_k$ represents some non-deterministic run of $\cA$ over $w$ for which $q$ is the $k$-th state, and $p$ was visited on the step when the current symbol at the top of the stack was pushed. On the other hand, (b) says that $(p,\gamma,q)\in \aT_k$ represents some run of $\cA$ over $w$ for which $\gamma$ is at the top of the stack, $q$ was visited on the step when $\gamma$ was pushed, and $p$ was visited on the step when the symbol below $\gamma$ was pushed (see Figure~\ref{fig:delta-schema} (left)). More importantly, these conditions are exhaustive, that is, every run of $\cA$ over $a_1 \ldots a_{k-1}$ is represented by $\rho^\dets$. 


\begin{figure}[t]
	\centering
\begin{tikzpicture}[->,>=stealth',roundnode/.style={circle,draw,inner sep=1.2pt},squarednode/.style={rectangle,inner sep=3pt}]
		\node [squarednode] (0) at (0.1, 0.2) {};
		\node [squarednode] (1) at (0.5, 1) {$p$};
		\node [squarednode] (2) at (1.9, 1) {$q$};
		\node [squarednode] (3a) at (1.2, 2.7) {\scalebox{0.8}{$(p,\gamma, p') \in T_k$}};
		\node [squarednode] (3) at (2.4, 2) {$p'$};
		\node [squarednode] (4a) at (3.2, 2.7) {\scalebox{0.8}{$(p'\!,q') \in  S_k$}};
		\node [squarednode, inner sep=1pt] (4) at (3.6, 2) {$q'$};
		
		\draw (0) to node[left] {{\footnotesize {\sf push}\! $\delta$}} (1);
		\draw (2) to node[left] {{\footnotesize {\sf push}\! $\gamma$}} (3);
		\tikzset{}
		\draw[-,decorate,decoration={snake,amplitude=.4mm,segment length=2mm,post length=0mm,pre length=0mm}] (1) to (2);
		\draw[-,decorate,decoration={snake,amplitude=.4mm,segment length=2mm,post length=0mm,pre length=0mm}] (3) to (4);
		
		\draw[-,gray] (3a) to (3);	
		\draw[-,gray] (4a) to (4);
		
		\draw[-,decorate,decoration={brace, amplitude=5}] (1.8,0.6) -- node[below=1mm] {\scalebox{0.7}{$\llevel(k)$}} (0.6,0.6) ;
		\draw[-,decorate,decoration={brace, amplitude=5}] (3.5,1.3) -- node[below=1mm] {\scalebox{0.7}{$\clevel(k)$}} (2.5,1.3) ;
		
		\draw[-,dashed] (4.2,0) -- (4.2,3);

	\end{tikzpicture} \ \
\begin{tikzpicture}[->,>=stealth',roundnode/.style={circle,draw,inner sep=1.2pt},squarednode/.style={rectangle,inner sep=3pt}]
	\node [squarednode] (title) at (0.2,2.5) {{\em Open}:};
	\node [squarednode] (0) at (0.1, 0.2) {};
	\node [squarednode] (1) at (0.5, 1) {$p$};
	\node [squarednode] (2) at (1.8, 1) {$p$};
	\node [squarednode] (2p) at (1.95, 1.05) {$'$};
	\node [squarednode] (3) at (2.3, 2) {$q$};
	\draw (0) to node[left] {{\footnotesize {\sf push}\! $\delta$}} (1);
	\draw (2) to node[left] {{\footnotesize {\sf push}\! $\gamma$}} (3);
	\tikzset{decoration={snake,amplitude=.4mm,segment length=2mm,
			post length=0mm,pre length=0mm}}
	\draw[-,decorate] (1) to (2);
\end{tikzpicture} \!\!\!\!\!\!
\begin{tikzpicture}[->,>=stealth',roundnode/.style={circle,draw,inner sep=1.2pt},squarednode/.style={rectangle,inner sep=3pt}]
	\node [squarednode] (title) at (0.2,2.5) {{\em Close}:};
	\node [squarednode] (0) at (0.1, 0.2) {};
	\node [squarednode] (1) at (0.5, 1) {$p$};
	\node [squarednode] (2) at (1.9, 1) {$\bigcirc$};
	\node [squarednode] (3) at (2.4, 2) {$p'$};
	\node [squarednode] (4) at (3.6, 2) {$q'$};
	\node [squarednode] (5) at (4.1, 1) {$q$};
	
	\draw (0) to node[left] {{\footnotesize {\sf push}\! $\delta$}} (1);
	\draw (2) to node[left] {{\footnotesize {\sf push}\! $\gamma$}} (3);
	\draw (4) to node[right] {{\footnotesize {\sf pop}\! $\gamma$}} (5);
	\tikzset{}
	\draw[-,decorate,decoration={snake,amplitude=.4mm,segment length=2mm,post length=0mm,pre length=0mm}] (1) to (2);
	\draw[-,decorate,decoration={snake,amplitude=.4mm,segment length=2mm,post length=0mm,pre length=0mm}] (3) to (4);
\end{tikzpicture}
	\caption{Left: An example run of some \vpa $\cA$ at step $k$. Right: Illustration of two nondeterministic runs for some \vpa $\cA$, as considered in the determinization process.}
	\label{fig:delta-schema}
	\vspace{-3mm}
\end{figure}
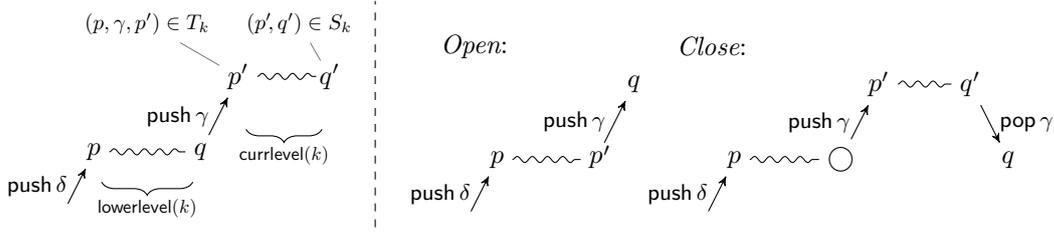



By these two invariants, the correctness of $\cA^\dets$ easily follows and the reader can get some intuition behind $\delta^\dets(\aS,\op{a})$ and $\delta^\dets(\aS,\aT,\cl{a})$ (see Figure~\ref{fig:delta-schema} (right) for a graphical description). 
Indeed, the most important consequence of these two invariants is that a tuple $(q_j,q_k) \in \aS_k$ represents the interval of some run over $w\spanc{j}{k}$ with $\clevel(k) = \spanc{j}{k}$ and the tuple $(q_i, \gamma, q_j) \in \aT_k$ represents the interval of some run over $w\spanc{i}{j-1}$ with $\llevel(k) = \spanc{i}{j-1}$, i.e., the level below. 
In other words, the configuration $(\aS_k, \tau_k)$ of $\cA^\dets$ forms a succinct representation of all the non-deterministic runs of $\cA$. This is the starting point of our update algorithm, that we discuss next.

\smallskip
\noindent \textbf{The streaming evaluation algorithm.} In Algorithm~\ref{alg:preprocessing} we present the update phase for solving the streaming version of $\enumvpt$. The main procedure is {\sc UpdatePhase}, that receives an I/O-unambiguous \vpt $\cT = (Q, \Sigma, \Gamma, \oalph, \Delta, \qinit, F)$ and a stream~$\Stream$, reads the next $k$-th symbol and computes the set of outputs $\br{\cT}(\Stream[1,k])$.
More specifically, it constructs an \dsepsabbr~$\cD$ and a vertex $v_{\text{out}}$ such that $\L_{\cD}(v_{\text{out}}) = \sem{\cT}(\Stream[1,k])$ if $\Stream[1,k]$ is well-nested and $\emptyset$ otherwise.
After the {\sc UpdatePhase} procedure is done, we can enumerate $\L_{\cD}(v_{\text{out}})$ with output-linear delay by calling the enumeration phase, that is, by applying Theorem~\ref{theo:data-structure-eps}. 

Towards this goal, in Algorithm~\ref{alg:preprocessing} we make use of the following data structures: First of all, we use an \dsepsabbr $\cD = (\Sigma, V, I, \ell, r, \lambda)$, nodes $v \in V$, and the functions $\add$, $\union$, and $\prod$ over $\cD$ and $v$ (see Section~\ref{sec:ds}). For the sake of simplification, we overload the notation of these operators slightly so that if $v = \emptyset$, then $\union(\cD,v,v') = \union(\cD,v',v) = (\cD,v')$. 
We use a hash table $\aS$ which indexes nodes $v$ in $\cD$ by pairs of states $(p,q) \in Q\times Q$. We denote the elements of $\aS$ as ``$(p,q) : v$'' where $(p,q)$ is the index and $v$ is the content. Furthermore, we write $\aS_{p,q}$ to access the node $v$. We also use a stack $\aT$ that stores hash tables: each element is a hash table which indexes vertices $v$ in $\cD$ by triples $(p, \gamma, q) \in Q \times \Gamma \times Q$. 
We assume that $\aT$ has the standard stack methods $\push$ and $\pop$ where if $\aT = t_k\cdots\,t_1$, then $\push(\aT,t) = t\, t_k\cdots\,t_1$ and $\pop(\aT) = t_{k-1}\cdots\,t_1$. We write $\emptyset$ for denoting the empty stack or for checking if $T$ is empty.
Similarly to $\aS$, we use the notation $\aT_{p,\gamma,q}$ to access the nodes in the topmost hash-table in $\aT$ (i.e. $\aT$ is a stack of hash tables). 
We assume that accessing a non-assigned index in these hash tables returns the empty set.
All variables (e.g., $S$, and $T$) are defined globally in Algorithm~\ref{alg:preprocessing} and they can be accessed by any of the subprocedures. 
given that we use the RAM model (see Section~\ref{sec:prelim}), each operation over hash tables or stacks takes constant time. 

Algorithm~\ref{alg:preprocessing} builds the \dsepsabbr $\cD$ incrementally, reading $\Stream$ one letter at a time by calling $\yield{\Stream}$ and keeping a counter $k$ for the position of the current letter. 
For every $k \in [1,n+1]$, {\sc UpdatePhase} builds the $k$-th iteration of table $\aS$ and stack $\aT$, which we note as $\aS^k$ and $\aT^k$ respectively. 
Before {\sc UpdatePhase} is called for the first time, it runs {\sc Intialize} (lines 1-4) to set the initial values of $k$, $\cD$, $\aS$, and $\aT$.
We consider the initial $\aS$ and $\aT$ as the $1$-st iteration, defined as $\aS^1 = \{(q,q): v_\eps \mid q \in I\}$ and $\aT^1 = \emptyset$ (i.e. the empty stack) where $v_\eps$ is a node in $\cD$ such that $\L_{\cD}(v_\eps) = \{\eps\}$ (lines 3-4).
In the $k$-th iteration, depending on whether the current letter is an open symbol or a close symbol, the {\sc OpenStep} or {\sc CloseStep} procedures are called, updating $\aS^{k-1}$ and $\aT^{k-1}$ to $\aS^{k}$ and $\aT^{k}$, respectively.  
More specifically, {\sc UpdatePhase} adds nodes to $\cD$ such that the nodes in $\aS^k$ represent the runs over $w\spanc{j}{k}$ where $\clevel(k)=\spanc{j}{k}$, and the nodes in the topmost table in $\aT^k$ represent the runs over $w\spanc{i}{j-1}$ where $\llevel(k) = \spanc{i}{j-1}$. 
Moreover, for a given pair $(p,q)$, the node $\aS^k_{p,q}$ represents all runs over $w\spanc{j}{k}$ with $\clevel(k)=\spanc{j}{k}$ that start on $p$ and end on~$q$. For a given triple $(p,\gamma, q)$ the node $\aT^k_{p,\gamma, q}$ represents all runs over $w\spanc{i}{j-1}$ with $\llevel(k) = \spanc{i}{j-1}$ that start on $p$, and end on~$q$ right after pushing $\gamma$ onto the stack. 
Here, the intuition gained in the determinization of VPA is crucial. Indeed, table $\aS^k$ and stack $\aT^k$ are the mirror of the configuration $(\aS_k, \tau_k)$ of~$\cA^\dets$ (recall invariants (a) and (b)).


\begin{algorithm}[t]
\caption{The update phase of the streaming evaluation algorithm for $\enumvpt$ given an I/O-unambiguous VPT $\cT = (Q, \Sigma, \Gamma, \oalph, \Delta, \qinit, F)$ and a stream $\Stream$.}\label{alg:preprocessing}
\begin{varwidth}[t]{0.537\textwidth}
\begin{algorithmic}[1]
\Procedure{{Initialize}}{$\cT, \Stream$}
\State $k\gets 1$, $\D\gets\emptyset$
\State $(\D,v_\epsilon)\gets\add(\D,\eps)$
\State $\aS\gets \{(q,q): v_\epsilon \mid q \in I\}$, $\aT \gets \emptyset$
\EndProcedure	

\State
	
\Procedure{{UpdatePhase}}{$\cT, \Stream$}
\State $a \gets \yield{\Stream}$
 \If{$a \in\opS$}
  \State $\D \gets $ {\sc OpenStep}$(\D, a,k)$  
 \ElsIf{$a \in\clS$}
  \State $\D \gets $ {\sc CloseStep}$(\D, a,k)$
 \EndIf 
 \State $k\gets k+1$
 \State $v_{\text{out}} \gets \emptyset$
 \If{$T = \emptyset$}
	 \ForEach{$p\in I, q\in F$ s.t.\,$\aS_{p,q}\!\neq\emptyset$}
	 \State $(\D,v_{\text{out}}) \gets \union(\D,v_{\text{out}},\aS_{p,q})$
 \EndFor
 \EndIf 
 \State {\sc EnumerationPhase}$(\cD, v_{\text{out}})$
\EndProcedure

\State

\Procedure{{IfProd}}{$\cD, v, \oout, k$}
\If{$\oout \neq \eps$}
\State $(\D',v') \gets \add(\D,(\oout, k))$
\State $(\D',v') \gets \prod(\D',v,v')$
\Else
\State $(\D',v') \gets (\cD, v)$
\EndIf
\State \Return $(\D',v')$
\EndProcedure
\algstore{myalg}
\end{algorithmic}	
\end{varwidth}
\begin{varwidth}[t]{0.6\textwidth}
\begin{algorithmic}[1]
\algrestore{myalg}
\Procedure{{OpenStep}}{$\D, \op{a},k$}
\State $\aS'\gets\emptyset$, $\aT\gets\push(\aT,\emptyset)$
\For{$p\in Q \text{ and } (p',\op{a},\oout,q,\gamma)\in\Delta$}
\If{$\aS_{p,p'}\neq\emptyset$}
\If{$\aS'_{q,q} = \emptyset$}
\State $(\D,v_{\eps}) \gets \add(\D,\eps)$
\State $\aS'_{q,q}\gets v_{\eps}$
\EndIf

\State $v\gets \aS_{p,p'}$
\State $(\D,v) \gets $ {\sc IfProd}$(\cD, v, \oout, k)$
\State $(\D,v) \gets \union(\D,v,\aT_{p,\gamma,q})$
\State $\aT_{p,\gamma,q}\gets v$
\EndIf
\EndFor
\State $\aS \gets \aS'$
\State \Return $\D$
\EndProcedure

\State

\Procedure{{CloseStep}}{$\D, \cl{a},k$}
\State $\aS'\gets\emptyset$
\For{$p,p'\in Q \text{ and } (q',\cl{a},\oout,\gamma,q)\in\Delta$}
\If{$\aS_{p',q'}\!\neq\emptyset \text{ and }\aT_{p,\gamma,p'}\!\neq \emptyset$}
\State $(\D,v)\gets\prod(\D,\aT_{p,\gamma,p'},\aS_{p',q'})$ \ \ \ \ \ \ \ \ \ \ \ \ \ 
\State $(\D,v) \gets ${\sc IfProd}$(\cD, v, \oout, k)$
\State $(\D,v)\gets\union(\D,v,\aS'_{p,q})$
\State $\aS'_{p,q}\gets v$
\EndIf
\EndFor
\State $\aT\gets\pop(\aT)$
\State $\aS \gets \aS'$
\State \Return $\D$
\EndProcedure

\end{algorithmic}
\end{varwidth}
\end{algorithm}


Before formalizing these notions, we will describe in more detail what the procedures {\sc OpenStep} and {\sc CloseStep} exactly do. Recall that the operation $\add(\D, a)$ simply creates a node in $\D$ labeled as $a$; the operation $\prod(\cD,v_1,v_2)$ returns a pair $(\cD',v')$ such that $\L_{\cD'}(v') = \L_{\cD}(v_1) \cdot \L_{\cD}(v_2)$; and the operation $\union(\cD,v_3,v_4)$ returns a pair $(\cD',v')$ such that $\L_{\cD'}(v') = \L_{\cD}(v_3)\cup \L_{\cD}(v_4)$. To improve the presentation of the algorithm, we include a simple procedure called \textsc{IfProd} (lines 19-25). Basically, this procedure receives a node $v$, an output symbol $\oout$, and a position $k$, and computes $(\cD', v')$ such that $\L_{\cD'}(v') = \L_{\cD}(v) \cdot \{(\oout, k)\}$ if $\oout \neq \eps$, and $\L_{\cD'}(v') = \L_{\cD}(v)$ otherwise.

In {\sc OpenStep}, $\aS^k$ is created (i.e. $S'$), and an empty table is pushed onto $\aT^{k-1}$ to form~$\aT^k$~(line~27). 
Then, all nodes in $\aS^{k-1}$ (i.e. $S$) are checked to see if the runs they represent can be extended with a transition in $\Delta$ (lines 28-29). 
If this is the case (lines 30 onwards), a node $v_\eps$  with the $\eps$-output is added in $\aS^k$ to start a new level (lines 30-32). 
Then, if the transition had a non-empty output, the node $\aS^k_{p,p'}$ is connected with a new label node to form the node $v$ (lines 33-34). 
This node is stored in $\aT^k_{p,\gamma,q}$, or united with the node that was already present there (lines 35-36).

In {\sc CloseStep}, $\aS^k$ is initialized as empty (line 41). 
Then, the procedure looks for all of the valid ways to join a node in $\aT^{k-1}$, a node in $\aS^{k-1}$, and a transition in $\Delta$ to form a new node in $\aS^k$. 
More precisely, it looks for quadruples $(p,\gamma,p',q')$ for which $\aT^{k-1}_{p,\gamma,p'}$ and $\aS^{k-1}_{p',q'}$ are defined, and there is a close transition that starts on $q'$ that reads $\gamma$ (lines 42-43). 
These nodes are joined and connected with a new label node if it corresponds (lines 44-45), and stored in~$\aS^k_{p,q}$ or united with the node that was already present there (lines 46-47). Finally, the top of the stack $T$ is popped after all tuples $(p,\gamma,p',q')$ are checked (line 48).

As it was already mentioned, in each step the construction of $\cD$ follows the ideas of the determinization of a visibly pushdown automata.
As such, Figure~\ref{fig:delta-schema} also aids to illustrate how the table $\aS^k$ and the top of the stack $\aT^k$ are constructed.

The way how the table $\aS^k$ and the stack $\aT^k$ are constructed is formalized in the following result. Recall that a run of $\cT$ over a well-nested word $w = a_1\cdots a_n$ is a sequence of the form
$
\rho = (q_1, \sigma_1) \xrightarrow{a_1/\!\ooutscr_1} \ldots  \xrightarrow{a_n/\!\ooutscr_n} (q_{n+1}, \sigma_{n+1})
$.
Given a span $\spanc{i}{j}$, define a subrun of $\rho$ as a subsequence
$
\rho\spanc{i}{j} = (q_i, \sigma_i) \xrightarrow{a_{i}/\!\ooutscr_{i}}  \ldots  \xrightarrow{a_{j-1}/\!\ooutscr_{j-1}} (q_j, \sigma_j).
$
We also extend the function $\out$ to receive a subrun $\rho\spanc{i}{j}$ in the following way:
$
\out(\rho\spanc{i}{j}) \ = \  \out(\oout_{i},{i}) \cdot \ldots \cdot \out(\oout_{j-1}, j-1)
$.
Finally, define $\Runs(\cT,w)$ as the set of all runs of $\cT$ over $w$. 

\begin{lemma}\label{vpt:steps}
	Let $\cT$ be a \vpt and $w = a_1\cdots a_n$ be a well-nested word. While running the procedure {\sc UpdatePhase} of Algorithm~\ref{alg:preprocessing}, for every $k\in[1,n+1]$, every pair of states $p,q$ and stack symbol $\gamma$ the following hold:
	\begin{enumerate}
		\item $\L_{\D}(\aS^k_{p,q})$ has exactly all sequences $\out(\rho\spanc{j}{k})$ such that $\rho\in\Runs(\cT,w\spanc{1}{k})$, $\clevel(k) = \spanc{j}{k}$, and $\rho\spanc{j}{k}$ starts on $p$ and ends on $q$.
		\item If $\llevel(k)$ is defined,  then $\L_{\D}(\aT^k_{p,\gamma,q})$ has exactly all sequences $\out(\rho\spanc{i}{j})$ such that $\rho\in\Runs(\cT,w\spanc{1}{j})$, $\llevel(k) = \spanc{i}{j-1}$, and $\rho\spanc{i}{j}$ starts on $p$, ends on $q$, and the last symbol pushed onto the stack was $\gamma$.
	\end{enumerate}
\end{lemma}

Since $w$ is well nested, then  $\clevel(|w|+1) = \spanc{1}{|w|+1}$, and so, the lemma implies that the nodes in $\aS^{|w|+1}$ represent all runs of $\cT$ over $w$. 
Then, whenever $\Stream[1, k]$ is well-nested, the stack $T$ is empty (i.e., $T = \emptyset$) and there may be something to enumerate (line 14). 
By taking the union of all pairs
in $\aS^{k+1}$ that represent accepting runs (as is done in lines 15-16), we can conclude the following result:


\begin{theorem}\label{eval:prep}
	Given a \vpt $\cT$ and a stream $\Stream$, {\sc UpdatePhase}$(\cT, \Stream)$ fulfils the conditions of a streaming evaluation algorithm and, after reading the $k$-th symbol, produces a pair $(\cD,v_{\operatorname{out}})$ such that $\L_{\cD}(v_{\operatorname{out}}) = \sem{\cT}(\Stream[1,k])$.
\end{theorem}

At this point we address the fact that $\cD$ needs to be unambiguous in order to enumerate all the outputs from $(\cD,v_{\text{out}})$ without repetitions.
This is guaranteed, essentially, by the fact that $\cT$ is I/O-unambiguous as well.
Indeed, the previous result holds even if $\cT$ is not I/O-unambiguous.
The next result guarantees that the output can be enumerated efficiently.

\begin{lemma}\label{eval:unambiguous}
	Let $\cT$ be an I/O-unambiguous \vpt. While running {\sc UpdatePhase} procedure of Algorithm~\ref{alg:preprocessing}, the \dsepsabbr $\cD$ is unambiguous at every step.
\end{lemma}

The complexity of this algorithm can be easily deduced from the fact that the \dsepsabbr operations we use take constant time (Theorem~\ref{theo:data-structure-eps}). For a \vpt $\cT = (Q, \Sigma, \Gamma, \oalph, \Delta, \qinit, F)$, in each of the calls to {\sc OpenStep}, lines 29-36 perform a constant number of instructions, and they are visited at most $\vert Q\vert\vert\Delta\vert$ times. In each of the calls to {\sc CloseStep}, lines 43-47 perform a constant number of instructions, and they are visited at most $\vert Q\vert^2\vert\Delta\vert$ times. 
Combined with Theorem~\ref{eval:prep}, Lemma~\ref{eval:unambiguous}, and Theorem~\ref{theo:data-structure-eps}, this proves our main result (i.e. Theorem~\ref{theo:main}).


%
%
%
%
%
%
%


\section{Future work}\label{sec:conclusions}

This paper offers several directions for future work. One direction is to find a streaming evaluation algorithm with polynomial update-time for non-deterministic VPT (i.e., in the size of the VPT). In~\cite{AmarilliBMN19}, the authors provided a polynomial-time offline algorithm for non-deterministic word transducers (called vset automata). They extended this result to trees in~\cite{AmarilliBMN19pods}. One could use these techniques in Algorithm~\ref{alg:preprocessing}; however, it is unclear how to extend ECS to deal with ambiguity in a natural way. Regarding space resources, another direction is to find an ``instance optimal'' streaming evaluation algorithm for VPT. As we mentioned, this problem generalizes the weak evaluation problem stated in~\cite{SegoufinV02}, given that it also considers the space to represent the output compactly. Finally, it would be interesting to explore practical implementations. Our view is that the data structure and algorithm presentation aid in reaching this goal, and it leaves space for suitable optimizations.


\bibliography{biblio}

\newpage
\onecolumn
\appendix

\newcommand{\atitle}{\ref{sec:vpawo}}
\section{Proofs from Section~\atitle}\label{sec:appendixvpawo}

\subsection{Proof of Proposition \ref{prop:msovpt}}

We encode words as logical structures as stated in the paper.
In the following, we will define the semantics of $\msom$ in a somewhat different, and more precise way so we can provide a clearer, and more formal statement of the proposition. To this end, we will also show how to encode the output sets to establish an exact equivalence between the logic and our transducer model.

Let $\varphi$ be a $\msom$ formula. We write $\varphi(\bar{x}, \bar{X})$ where $\bar{x}$ and $\bar{X}$ are the sets of free first-order and monadic second-order variables of $\varphi$, respectively. 
An assignment $\sigma$ for $w$ is a function $\sigma\colon \bar{x}\cup \bar{X}\to 2^{[1,n]}$ such that $|\sigma(x)| = 1$ for every $x \in \bar{x}$ (note that we treat first-order variables as a special case of monadic second-order variables). 
As usual, we denote by $\textsf{dom}(\sigma) = \bar{x}\cup \bar{X}$ the domain of the function $\sigma$. 
Then we write $(w, \sigma) \models \varphi(\bar{x}, \bar{X})$ when $\sigma$ is an assignment over $w$, $\textsf{dom}(\sigma) = \bar{x}\cup \bar{X}$, and $w$ satisfies $\varphi(\bar{x}, \bar{X})$ when each variable in $\bar{x}\cup \bar{X}$ is instantiated by $\sigma$. 
Given a formula $\varphi(\bar{x},\bar{X})$, we define $\sem{\varphi}(w) = \{\sigma \mid (w, \sigma) \models \varphi(\bar{x},\bar{X})\}$. 
For the sake of simplification, from now on we will only use $\bar{X}$ to denote the free variables of $\varphi(\bar{X})$ and use $X \in \bar{X}$ for an first-order or monadic second-order variable.

For any assignment $\sigma$ over $w$, we define the support of $\sigma$, denoted by $\textsf{supp}(\sigma)$, as the set of positions mentioned in $\sigma$; formally, $\textsf{supp}(\sigma) = \{i \mid \exists v \in  \textsf{dom}(\sigma)\text{ s.t. } i \in  \sigma(v)\}$. 
Furthermore, we encode assignments as sequences over the support as follows:
Let $\textsf{supp}(\sigma) = \{i_1,\ldots, i_m\}$ such that $i_j < i_{j+1}$ for every $j < m$. 
Then, we define the (word) encoding of $\sigma$ as:
$$
\textsf{enc}(\sigma) = (\bar{X}_1, i_1)(\bar{X}_2, i_2) \ldots (\bar{X}_m, i_m)
$$
such that $\bar{X}_j = \{X \in \textsf{dom}(\sigma) \mid i_j \in \sigma(X)\}$ for every $j \leq m$. 
That is, we represent $\sigma$ as an increasing sequence of positions, where each position is labeled with the variables of $\sigma$ where it belongs.

The statement of the proposition can be formulated as follows:
\begin{proposition}[Proposition~\ref{prop:msovpt}]
	Fix a structured alphabet $\Sigma$. Let $\bar{X}$ be a set of MSO variables and ${\pazocal X} = 2^{\bar{X}}$. 
	\begin{enumerate}
		\item For any $\msom$ formula $\varphi(\bar{X})$ there exists a VPT $\cT$ with output alphabet ${\pazocal X}$ such that for every $w\in \wnS$:
		$$
		\sem{\cT}(w) \ \ = \ \ \{\textsf{enc}(\sigma)\mid \sigma\in\sem{\varphi}(w)\}.
		$$
		\item For any VPT $\cT$ with output alphabet ${\pazocal X}$ there exists a  $\msom$ formula $\varphi(\bar{X})$ such that for every $w\in \wnS$:
		$$
		\{\textsf{enc}(\sigma)\mid \sigma\in\sem{\varphi}(w)\} \ \ = \ \  \sem{\cT}(w).
		$$
	\end{enumerate}
\end{proposition}

The proof of this proposition is largely based on the proof of Theorem 4 in~\cite{AlurM04}.
To prove (1) we can follow the exact same argument as the {\em if} direction of the proof and be left with a VPA $\cA$ over the input alphabet $\Sigma^{\bar{X}} = \Sigma \times {\pazocal X}$ whose language is the set of words which encode a valuation $\sigma$ of $\bar{X}$ along with a word $w$ for which $(w, \sigma)\models\varphi(\bar{X})$. 
We define a straighforward transformation of transitions from this VPA to VPT as follows: $f(t) = t'$ iff $t$ has input symbol $(a, V)$ and $t'$ has input symbol $a$ and output symbol $V$.
We obtain the desired VPT $\cT$ by replacing solely the transition relation $\Delta$ in $\cA$ by $\{f(t)\mid t\in\Delta\}$.

To prove (2) we convert $\cT$ into a VPA $\cA$ with input alphabet $\Sigma^{\bar{X}}$ in the opposite way as in (1) and use the result of~\cite{AlurM04} itself to obtain a $\msom$ formula with no free variables $\varphi'$ over the same input alphabet. We replace any instance of $P_{(a, V)}(x)$ in $\varphi$ by the expression $P_a(x) \wedge \bigwedge_{X\in V}x\in X\wedge \bigwedge_{X\in\bar{X}\setminus V}x\not\in X$ to obtain a formula $\varphi(\bar{X})$ over $\Sigma$ which proves the statement.

\subsection{XPath query examples}
In this section we show two examples of XPath queries and their translations into VPT. 
The type of XPath query we focus on here are {\em full-fledged evaluation} queries, where the expected output set contains the nodes selected by the query. 
The way we translate an XPath query $\cQ$ into a VPT $\cT$ is as follows:  
Let $\tau$ be a function which encodes unranked trees as nested strings by a depth-first traversal. 
In our setting, a node labeled $a$ is encoded as the pair of open/close symbols $\op{a}$, $\cl{a}$. 
Let $\Sigma_{\cQ}$ be the set of labels on trees mentioned in $\cQ$, and let $\opS$ and $\clS$ be the sets of open and close symbols that encode the labels in $\Sigma_{\cQ}$. 
Let $\cQ(D)$ be the set of nodes in $D$ that match $\cQ$. 
Furthermore, consider the set $\cI_{\cQ, D}$ of positions in $\tau(D)$ that correspond to nodes in $\cQ(D)$.
A VPT $\cT$ is a translation of a query $\cQ$ if and only if its input alphabet is $(\opS, \clS)$, and for a given tree $D$, the set $\br{\cT}(\tau(D))$ contains exactly the strings $({\sf L}, i)$ for which $i\in\cI_{\cQ, D}$. 
This notion of translation into VPT is quite natural since the output set of the VPT can be used straightforwardly to reconstruct $\cQ(D)$ with a single pass over $D$.

The VPT in this section are shown graphically using the following notation: An open transition $(p,\op{s},\oout,q,\gamma)$ is represented by an edge from $p$ to $q$ with the label $\op{s} / \gamma$ if $\oout = \eps$ and with the label $\op{s} / \gamma : \oout$ if $\oout \neq \eps$. A close transition $(p,\cl{s},\eps,\gamma,q)$ is represented with the label $\cl{s}, \gamma$. As is customary, we extend this notation by representing multiple transitions that differ only by their input symbol as a single transition over the set of these symbols. We also use the symbol $|$ to group transitions that start and end in the same states.

As a first example, consider the XPath query $\cQ_1 = {\tt //}a{\tt /}b$. This query can be translated into the VPT shown in Figure~\ref{fig-xpath-vpt-small}. 


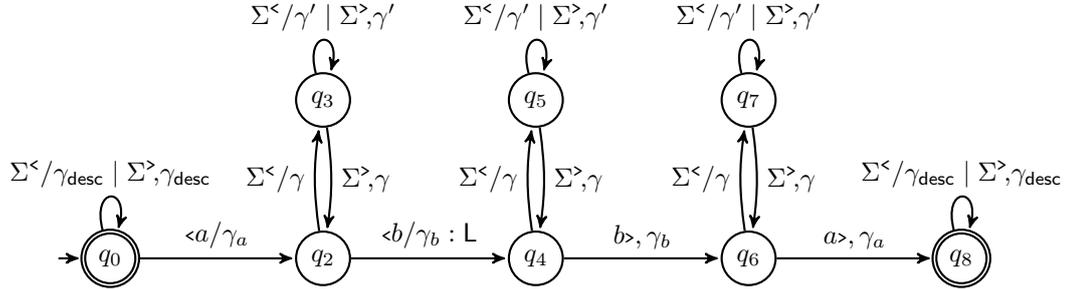
\begin{figure}[h]
	\centering
	
	\begin{tikzpicture}[scale=0.7,->,>=stealth',shorten >=1pt,auto,node distance=2cm,thick,state/.style={circle,draw}, color=black]
		\node[state,draw=none,scale=0.1] (in) at (-1,0) {};
		\node[state, accepting] (q0) at (0,0) {$q_0$};
		\node[state] (q2) at (4,0) {$q_2$};
		\node[state] (q3) at (4,3) {$q_3$};
		\node[state] (q4) at (8,0) {$q_4$};
		\node[state] (q5) at (8,3) {$q_5$};
		\node[state] (q6) at (12,0) {$q_6$};
		\node[state] (q7) at (12,3) {$q_7$};
		\node[state, accepting] (q8) at (16,0) {$q_8$};
		\draw (in) to (q0);
		\draw (q0) to[loop above] node {$\opS / \gamma_{\textsf{desc}}\mid \clS\!,\! \gamma_{\textsf{desc}}$} (q0);
		\draw (q0) to node [above] {$\op{a} / \gamma_a$} (q2);
		\draw (q2) to[out=98,in=-98] node [left] {$\opS / \gamma$} (q3);
		\draw (q3) to[loop above] node {$\opS / \gamma'\mid \clS\!,\! \gamma'$} (q3);
		\draw (q3) to[out=-82,in=82] node [right] {$\clS\!,\! \gamma$} (q2);
		\draw (q2) to node [above] {$\op{b} / \gamma_b : {\sf L}$} (q4);
		\draw (q4) to[out=98,in=-98] node [left] {$\opS / \gamma$} (q5);
		\draw (q5) to[loop above] node {$\opS / \gamma' \mid \clS\!,\! \gamma'$} (q5);
		\draw (q5) to[out=-82,in=82] node [right] {$\clS\!,\! \gamma$} (q4);
		\draw (q4) to node [above] {$\cl{b} , \gamma_b$} (q6);
		\draw (q6) to[out=98,in=-98] node [left] {$\opS / \gamma$} (q7);
		\draw (q7) to[loop above] node {$\opS / \gamma' \mid \clS\!,\! \gamma'$} (q7);
		\draw (q7) to[out=-82,in=82] node [right] {$\clS\!,\! \gamma$} (q6);
		\draw (q6) to node [above] {$\cl{a}, \gamma_a$} (q8);
		\draw (q8) to[loop above] node {$\opS / \gamma_{\textsf{desc}}\mid \clS\!,\! \gamma_{\textsf{desc}}$} (q8);
	\end{tikzpicture}
	
	\caption{A VPT that translates the XPath query $\cQ_1 = {\tt //}a{\tt /}b$. Its input alphabet consists of the sets $\opS = \{\op{a}, \op{b}\}$ and $\clS = \{\cl{a}, \cl{b}\}$.}
	
	\label{fig-xpath-vpt-small}
\end{figure}

As a more involved example, consider the following XPath query over the tree alphabet $\{a, b, c\}$: 
$$
\cQ_2 = \texttt{child:}a\texttt{/descendant:}b[\texttt{following-sibling:}c]
$$

A VPT that translates this query is shown in Figure~\ref{fig-xpath-vpt}.
	

\begin{figure}[h]
	\centering
	
	\begin{tikzpicture}[scale=0.85,->,>=stealth',shorten >=1pt,auto,node distance=2cm,thick,state/.style={circle,draw},color=black]
		\node[state,draw=none,scale=0.1] (in) at (-1,10) {};
		\node[state, accepting] (q0) at (0,10) {$q_0$};
		\node[state] (q1) at (4.5,10) {$q_1$};
		\node[state] (q2) at (3,8) {$q_2$};
		\node[state] (q3) at (6,6) {$q_3$};
		\node[state] (q4) at (3,4) {$q_4$};
		\node[state] (q5) at (8,4) {$q_5$};
		\node[state] (q6) at (6,2) {$q_6$};
		\node[state] (q7) at (3,0) {$q_7$};
		\node[state, accepting] (q8) at (0,4) {$q_8$};
		\draw (in) to (q0);
		\draw (q0) to[out=20,in=160] node [above] {$\opS / \gamma_0$} (q1);
		\draw (q1) to[out=-160,in=-20] node [above] {$\clS , \gamma_0$} (q0);
		\draw (q1) to[loop right] node [right] {$\opS / \gamma'_0$ | $\clS, \gamma'_0$} (q1);
		\draw (q0) to node [right=5pt, pos=0.45] {$\op{a} / \gamma_a$} (q2);
		\draw (q2) to[loop right] node [right] {$\opS / \gamma_{\textsf{desc}}$ | $\clS, \gamma_{\textsf{desc}}$} (q2);
		\draw (q2) to node [right=5pt, pos=0.45] {$\op{b} / \gamma_b : {\sf L}$} (q3); 
		\draw (q2) to node [left] {$\cl{a}, \gamma_a$} (q8);
		\draw (q3) to[loop right] node [right] {$\opS / \gamma'_b$ | $\clS, \gamma'_b$} (q3);
		\draw (q3) to node [above=1.5pt, pos=0.65, yshift=3pt] {$\cl{b}, \gamma_b$} (q4);
		\draw (q4) to[out=15,in=165] node [above=-1pt] {$\{\op{a}, \op{b}\} / \gamma_{\textsf{sib}}$} (q5);
		\draw (q5) to[out=-165,in=-15] node [above] {$\{\op{a}, \op{b}\}, \gamma_{\textsf{sib}}$} (q4);
		\draw (q5) to[loop right] node [right] {$\opS / \gamma'_{\textsf{sib}}$ | $\clS, \gamma'_{\textsf{sib}}$} (q5);
		\draw (q4) to node [above, pos=0.7, yshift=3pt] {$\op{c} / \gamma_c$} (q6);
		\draw (q6) to[loop right] node [right] {$\opS / \gamma'_c$ | $\clS, \gamma'_c$} (q6);
		\draw (q6) to node [above, pos=0.6, yshift=4pt] {$\cl{c} / \gamma_c$} (q7);
		\draw (q7) to[loop below] node [below] {$\opS / \gamma_{\textsf{desc}}$ | $\clS, \gamma_{\textsf{desc}}$} (q3);
		\draw (q7) to node [left=1pt] {$\cl{a}, \gamma_a$} (q8);
		\draw (q8) to[loop left] node [left] {$\opS / \gamma'$ | $\clS, \gamma'$} (q8);
	\end{tikzpicture}

\caption{The VPT that translates the XPath query $\cQ_2$.  Its input alphabet consists of the sets $\opS = \{\op{a}, \op{b}, \op{c}\}$ and $\clS = \{\cl{a}, \cl{b}, \cl{c}\}$.}
	
	\label{fig-xpath-vpt}
\end{figure}
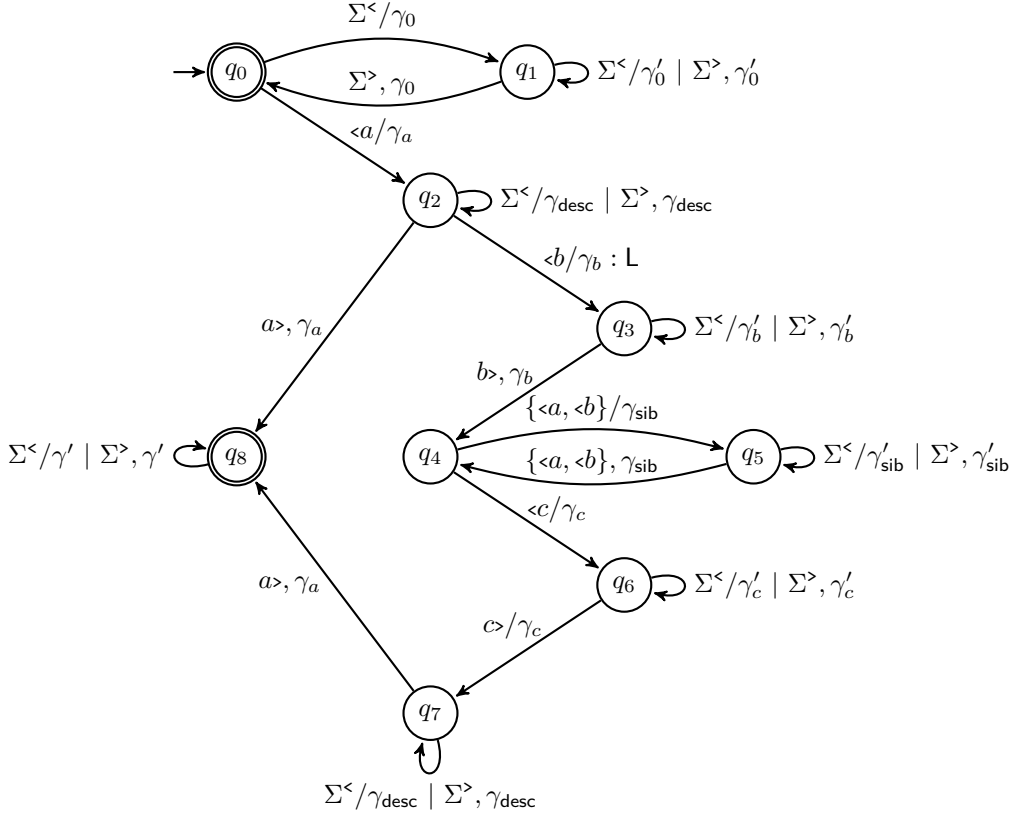

\subsection{Proof of Lemma \ref{vpawo:det}}

Let $\cT = (Q, \Sigma, \Gamma, \oalph, \Delta, \qinit, F)$.
We will construct an input-output deterministic \vpt $\cT' = ( Q', \Sigma, \Gamma', \oalph, \delta^\dets, S_{I}, F')$ as follows:
Let $Q' = 2^{Q\times Q}$ and $\Gamma' = 2^{Q\times\Gamma\times Q}$. 
Let $S_I = \{(q,q)\mid q\in \qinit\}$ and let $F' = \{S\mid (p,q)\in S\text{ for some }p\in I \text{ and } q\in F \}$. 
Let $\delta$ be defined as follows:
\begin{itemize}
	\item For $\op{a}\in \opS$ and $\oout\in\Omega$, $\delta(S,\op{a},\oout) = (S',T)$, where:
	\begin{align*}
		T &= \{(p,\gamma,q)\mid (p,p')\in S \text{ and } (p',\op{a},\oout,\gamma,q)\in\Delta\text{ for some }q\in Q \},\\
		S'&= \{(q,q)\mid(p,p')\in S\text{ and }(p',\op{a},\oout,\gamma,q)\in\Delta\text{ for some }p,p'\in Q \text{ and } \gamma\in\Gamma\}
	\end{align*}
	\item For $\cl{a}\in\clS$ and $\oout\in\Omega$, $\delta(S,\cl{a},\oout,T) = S'$ where, if $T\subseteq Q\times\Gamma\times Q$, then: 
	\begin{align*}
		S' &= \{(p,q)\mid(p,\gamma,p')\in T\text{ and }(p',q')\in S\text{ and }(q',\cl{a},\oout,\gamma,q)\in\Delta\ \ \ \ \ \ \ \ \ \ \ \\ &\ \ \ \ \ \ \ \ \ \ \ \ \ \ \ \ \ \ \ \ \ \ \ \ \ \ \ \ \ \ \ \ \ \ \ \ \ \ \ \ \ \ \ \ \ \ \ \ \ \ \ \ \ \text{ for some }p',q'\in Q,\gamma\in\Gamma\},
	\end{align*}
	\item For $a\in\noS$ and $\oout\in\Omega$, $\delta(S,a) = S'$ where:
	\begin{align*}
		S' &= \{(q,q'')\mid (q,q')\in S\text{ and }(q',a,\oout,q'')\in\Delta\text{ for some }q'\in Q\}.
	\end{align*}
\end{itemize}
One can immediately check that this automaton is input-output determinstic since the transition relation is modelled as a partial function.

We will prove that $\cT$ and $\cT'$ are equivalent by induction on well-nested words. To aid our proof, we will introduce a couple of ideas. First, we extend the definition of a run to include sequences that start on an arbitary configuration. Also, given a run 

$$
\rho = (q_1, \sigma_1) \xrightarrow{s_1/\!\ooutscr_1} (q_2, \sigma_2) \xrightarrow{s_2/\!\ooutscr_2} \cdots  \xrightarrow{s_n/\!\ooutscr_n} (q_{n+1}, \sigma_{n+1}),
$$

\noindent and a span $\spanc{i}{j}$, define a subrun of $\rho$ as the subsequence

$$
\rho\spanc{i}{j} = (q_i, \sigma_i) \xrightarrow{s_{i}/\!\ooutscr_{i}} (q_{i+1}, \sigma_{i+1}) \xrightarrow{s_{i+1}/\!\ooutscr_{i+1}} \cdots  \xrightarrow{s_{j-1}/\!\ooutscr_{j-1}} (q_j, \sigma_j).
$$

In this proof, we only consider subruns such that $w\spanc{i}{j} = s_{i}s_{i+1}\cdots s_{j-1}$ is a well-nested word. 
A second definition we will use is that of a \vpt with arbitrary initial states. 
Formally, let $S \subseteq Q$. 
We define $\cT_{q}$ as the \vpt that simulates $\cT$ by starting on the configuration $(q,\eps)$. 
Note that for a run $\rho = (q_1, \sigma_1) \xrightarrow{s_1/\!\ooutscr_1} \cdots  \xrightarrow{s_n/\!\ooutscr_n} (q_{n+1}, \sigma_{n+1})$ of $\cT$ over $w = s_1\cdots s_n$ and a well-nested span $\spanc{i}{j}$, the subrun $\rho\spanc{i}{j}$ is one of the runs of $\cT_{q}$ over $w\spanc{i}{j}$ modulo $\sigma_i$, which is present in all of the stacks in $\rho$ as a common suffix.

We shall prove first that $\br{\cT}(w) \subseteq \br{\cT'}(w)$ for every well-nested word $w$. This is done with the aid of the following result:

\begin{claim}
	For a well-nested word $w$, output sequence $\mu$, states $p,q\subseteq Q$, and a set $S$ that contains $(p,q)$, if there is a run of $\cT_q$ over $w$ and $\mu$ such that its last state is $q'$, the (only) run of $\cT'_{S}$ over $w$ and $\mu$ ends in a state $S'$ which contains $(p,q')$.
\end{claim}

\begin{proof}
	We will prove the claim by induction on $w$.
	
	If $w = \eps$, the proof is trivial since $q = q'$. If $w = a\in\noS$ the proof follows straightforwardly from the construction of $\delta$.
	
	If $w,v\in\wnS$, and $\mu, \kappa \in \oalph^*$, let $p,q\in Q$, let $S$ be a set that contains $(p,q)$, and let $\rho$ be a run of $\cT_q$ over $wv$ and $\mu\kappa$, which ends in a state $q'$. 
	Our goal is to prove that the run $\rho'$ of $\cT'_S$ over $wv$ and $\mu\kappa$ ends in a state that contains $(p,q')$. 
	Let $n = \vert w\vert$, $m = \vert v \vert$, and let $q^w$ be the last state of the subrun $\rho[1,n+1\rangle$. 
	Consider as well $\rho[n+1,n+m+1\rangle$, which is a run of $\cT_{q^w}$ over $v$ and $\varkappa$ that ends in $q'$. 
	From the hypothesis two conditions follow: 
	(1) In the run of $\cT'_S$ over $w$ and $\mu$ the last state $S'$ contains $(p,q^w)$, and 
	(2) in the run of $\cT'_{S'}$ over $v$ and $\kappa$ the last state contains $(p,q')$. 
	It can be seen that $\rho'$ is the concatenation of these two runs, so this proves the claim.
	
	If $w\in\wnS$, $\op{a}\in\opS$, $\cl{b}\in\clS$, $\mu\in\oalph^*$, and $\oout_1,\oout_2\in\oalph$, let $p,q\in Q$, let $S$ be a set that contains $(p,q)$, and let $\rho$ be a run of $\cT_q$ over $\op{a}w\cl{b}$ and $\oout_1\!\mu\oout_2$. 
	Let $n = \vert w\vert$, and let $q, q_2, \ldots, q_{n+2}, q_{n+3}$ be the states of $\rho$ in order. 
	Our goal is to prove that the run $\rho'$ of $\cT'_S$ over $\op{a}w\cl{b}$ and $\oout_1\!\mu\oout_2$ ends in a state that contains $(p,q_{n+3})$. 
	Let $(q_2,\mathtt{x})$ be the second configuration of $\rho$. 
	This implies that $(q,\op{a},\oout_1,q_2,\gamma)\in\Delta$ and $(q_{n+2},\cl{b},\oout_2,\gamma,q_{n+3})\in\Delta$.
	Let $S'$ and $T$ be such that $\delta(S,\op{a},\oout_1) = (S',T)$.
	Therefore, $(q_2,q_2)\in S'$ and $(p,\gamma,q_2)\in T$.
	Consider the subrun $\rho[2,n+2\rangle$, which is a run of $\cT_{q_2}$ over $w$ and $\mu$ that ends in $q_{n+2}$ modulo the stack suffix $\gamma$.
	Since $(q_2,q_2)\in S'$, from the hypothesis it follows that the run of $\cT'_{S'}$ over $w$ and $\mu$ ends in a state $S''$ that contains $(q_2,q_{n+2})$.
	This run starts on the configuration $(S',\eps)$ and ends in $(S'',\eps)$, so a run on the same automaton that starts on $(S',T)$ and reads the same symbols will end in $(S'', T)$, which is the case for the subrun $\rho'[2,n+2\rangle$.
	Therefore, the construction of $\delta$ implies that $(p,q_{n+3})$ is contained in the last state of $\rho'$, which proves the claim.
\end{proof}

Let now $w$ be a well-nested word and $\mu$ be an output sequence such that $\cT$ accepts $(w, \mu)$. Let $\rho$ be an accepting run of $\cT$ over $(w,\mu)$ which starts on a state $p\in I$ and ends in a state $q\in F$. Note that $\cT_p$ also accepts $(w,\mu)$. Note that $\cT = \cT'_{S_{I}}$, and since $(p,p)\in S_{I}$ the claim implies that the run of $\cT$ over $(w,\mu)$ ends in a state which contains $(p,q)$, and so this run is accepting. This proves that $\br{\cT}(w) \subseteq \br{\cT'}(w)$.

To prove that $\br{\cT'}(w) \subseteq \br{\cT}(w)$ we use a similar result:

\begin{claim}
	For a well-nested word $w$, output sequence $\mu$, states $q,p,q'\subseteq Q$, and a set $S$ that contains $(p,q)$, if the run of $\cT'_S$ over $w$ and $\mu$ ends on a state $S'$ that contains $(p,q')$, then there is a run of $\cT_q$ over $w$ and $\mu$ such that its last state is $q'$.
\end{claim}
\begin{proof}
	We will prove the claim by induction on $w$.
	
	If $w = \eps$, the proof is trivial since $q = q'$. If $w = a\in\noS$ the proof follows straightforwardly from the construction of $\delta$.
	
	If $w,v\in\wnS$, and $\mu, \kappa \in \oalph^*$, let $p,q,q'\in Q$, let $S$ be a set that contains $(p,q)$, and let $\rho$ be the run of $\cT'_S$ over $wv$ and $\mu\kappa$, which ends in a state $S'$ that contains $(p,q')$. 
	Our goal is to prove that there is a run $\rho'$ of $\cT_q$ over $wv$ and $\mu\kappa$ such that its last state is $q'$.
	Let $n = \vert w\vert$, $m = \vert v \vert$, and let $S^w$ be the last state of the subrun $\rho[1,n+1\rangle$. 
	Consider as well $\rho[n+1,n+m+1\rangle$, which is a run of $\cT'_{S^w}$ over $v$ and $\kappa$ that ends in $S'$.
	From the construction of $\delta$, it is clear that if a non-empty state $S'$ follows from $S$ in a run of $\cT'$, then $S$ is not empty.
	Let $(p,q^w)\in S^w$.
	From the hypothesis two conditions follow: 
	(1) There is a run $\rho_1$ of $\cT_q$ over $w$ and $\mu$ such that its last state is $q^w$
	(2) There is a run $\rho_2$ of $\cT_{q^w}$ over $v$ and $\kappa$ such that its last state is $q'$. 
	We then construct $\rho'$ by concatenating $\rho_1$ and $\rho_2$ which ends in $q'$, and this proves the claim.
	
	If $w\in\wnS$, $\op{a}\in\opS$, $\cl{b}\in\clS$, $\mu\in\oalph^*$, and $\oout_1,\oout_2\in\oalph$, let $p,q,q'\in Q$, let $S$ be a set that contains $(p,q)$, and let $\rho$ be the run of $\cT'_S$ over $\op{a}w\cl{b}$ and $\oout_1\!\mu\!\oout_2$. 
	Let $n = \vert w\vert$, let $S, S_2, \ldots, S_{n+2}, S_{n+3}$ be the states of $\rho$ in order, and suppose there is a pair $(p,q')\in S_{n+3}$.
	Our goal is to prove that there is a run $\rho'$ of $\cT_q$ over $\op{a}w\cl{b}$ and $\oout_1\!\mu\!\oout_2$ that ends in $q'$. 
	Let $(S_2,T)$ be the second configuration of $\rho$.
	From the construction of $\delta$, there exist $q_2,q_{n+2}\in Q$ and $x\in\Gamma$ such that $(q_{n+2},\cl{b},\oout_2,\gamma,q_{n+3})\in\Delta$, $(p,\gamma,q_2)\in T$ and $(q_2,q_{n+2})\in S_{n+2}$.
	Since $w$ is well-nested, this $T$ could only have been pushed after reading $\op{a}/\!\oout_1$, which implies that $(q,\op{a},\oout_1,q_2,\gamma)\in\Delta$. This, in turn, means that $(q_2,q_2)\in S_2$.
	Let us consider the subrun $\rho[2,n+2\rangle$, which is a run of $\cT'_{S_2}$ over $w$ and $\mu$ that ends in $S_{n+2}$ modulo the common stack suffix $T$.
	We now have that $(q_2,q_2)\in S_2$ and $(q_2,q_{n+2})\in S_{n+2}$, and so, from the hypothesis it follows that there is a run $\rho''$ of $\cT_{q_2}$ over $w$ and $\mu$ such that its last state is $q_{n+2}$.
	In a similar fashion as in the previous claim, we modify the run slightly to obtain one that starts and ends on the stack $\gamma$.
	This new run can be easily extended with the transitions $(q,\op{a},\oout_1,q_2,\gamma),(q_{n+2},\cl{b},\oout_2,\gamma,q_{n+3})\in\Delta$, and as a result, we obtain a run $\rho'$ of $\cT_q$ that fulfils the conditions of the claim.
\end{proof}

Let now $w$ be a well nested word and let $\mu$ be an output sequence such that $\cT'$ accepts $(w, \mu)$. Since $\cT' = \cT_{S_I}$ and the run of $\cT'$ over $(w, \mu)$ ends in a state $S\in F'$, we have that $S$ contains an element $(p,q)$ such that $p\in I$ and $q\in F$. Moreover, $(p,p)\in S_I$. From the prevous claim, it follows that there is an accepting run of $\cT_p$ over $(w, \mu)$ such that its last state is $q$. Therefore, $\cT$ accepts $(w, \mu)$. This proves that $\br{\cT'}(w) \subseteq \br{\cT}(w)$.

We conclude that $\br{\cT}(w) = \br{\cT'}(w)$ for every well-nested word $w$.\qed

\subsection{Proof of Theorem~\ref{vpawo:deltamain}}

The proof of the theorem is a consequence of the following lemma.

\begin{lemma}\label{vpawo:delta}
	For every I/O-unambiguous \vpt $\cT$ there exists an I/O-unambiguous \vpt $\cT'$ such that $\br{\cT'}(w) = \br{\cT}(w) \setminus \bigcup_{i < |w|} \br{\cT}(	w[1, i])$ for every $w\in\wnS$. Furthermore, the size of $\cT'$ is linear on the size of $\cT$.
\end{lemma}

Let $\cT = (Q, \Sigma, \Gamma, \oalph, \Delta, \qinit, F)$ be an I/O-unambiguous \vpt. We construct a VPT $\cT' = (Q', \Sigma, \Gamma, \oalph, \Delta', \qinit, F')$ such that $Q' = Q \times \{1, 2\}$, $\qinit' = \qinit \times\{1\}$, $F' = F \times\{1\}$ and $\Delta'$ is as follows: 
\begin{align*}
	\Delta' =\ & \{((p,1),\op{a},\oout,(q,1),\gamma)\mid \op{a}\in\opS\text{ and }(p,\op{a},\oout,q,\gamma)\in\Delta\text{ where either }\!\oout\in\oalph\text{ or }p\not\in F\}\ \cup\\
	& \{((p,1),\op{a},\eps,(q,2),\gamma)\mid \op{a}\in\opS\text{ and }(p,\op{a},\eps,q,\gamma)\in\Delta\text{ where } p\in F\}\ \cup\\
	& \{((p,2),\op{a},\oout,(q,1),\gamma)\mid \op{a}\in\opS\text{ and }(p,\op{a},\oout,q,\gamma)\in\Delta\text{ where }\!\oout\in\oalph\}\ \cup\\
	& \{((p,2),\op{a},\eps,(q,2),\gamma)\mid \op{a}\in\opS\text{ and }(p,\op{a},\eps,q,\gamma)\in\Delta\}.
\end{align*}
This construction was only shown for symbols in $\opS$, but it should include analogous constructions for symbols in $\clS$ and $\noS$, which are omitted for convenience.
The idea behind this construction is to separate the \vpt in two halves. Each run starts on the first half (marked 1) and once it reaches a final state, it changes into the second half (marked 2). The run then stays on the second half until it sees an output symbol, with which it returns to the first half.

To show that $\br{\cT'}(w) = \br{\cT}(w) \setminus \bigcup_{i < |w|} \br{\cT}(w[1, i\rangle)$ consider a $w\in \wnS$. Let $\mu$ be an output in $\br{\cT'}(w)$ and consider an accepting run $\rho'$ such that $\out(\rho') = \mu$. We can construct an accepting run $\rho$ of $\cT$ over $w$ by starting from $\rho'$ and replacing any appearance of a state $(q, k)$ by $q$. From this it follows that $\mu \in \br{\cT}(w)$. Assume now that $\mu \in \br{\cT}(w[1, i\rangle)$ for some $i < |w|$. From the construction of $\Delta'$ it can be seen that the $i$-th and following states in $\rho'$ are of the form $(q, 2)$, as all of the following transitions in $\rho'$ have $\eps$ as their output symbols. Therefore, $\rho'$ cannot be an accepting run, and we reach a contradiction, from which we conclude that $\mu \in \br{\cT}(w) \setminus \bigcup_{i < |w|} \br{\cT}(w[1, i\rangle)$. Let $\mu$ now be an output in $\mu \in \br{\cT}(w) \setminus \bigcup_{i < |w|} \br{\cT}(w[1, i\rangle)$ and let $\rho$ be the accepting run of $\cT$ over $w$ such that $\out(\rho) = \mu$. It can be seen from the construction of $\Delta'$ that the run of $\cT'$ over $w$ is identical to $\rho$ except each state $q$ in $\rho$ appears as $(q, k)$ in $\rho'$. We will show that the last state in $\rho'$ is of the form $(q,2)$. Towards a contradiction, assume that it is not. Therefore, in $\rho'$ there is a transition where the first state is of the form $(p,1)$ and the second is of the form $(q,2)$, and furthermore, every transition following this one has $\eps$ as its output symbol. Let $i$ be the step where this happens. From the construction of $\Delta$ we see that the $i$-th state is in $F$, from which it follows that the run $\rho'_i$ built from the first $i$ steps in $\rho'$ is an accepting run of $\cT'$ over $w[1,i\rangle$ and that $\out(\rho'_i) = \mu$. We can do a similar process as a above and construct an accepting run of $\cT$ over $w[1,i\rangle$ that renders the same output $\mu$, which contradicts our assumption that $\mu\not\in\bigcup_{i < |w|} \br{\cT}(w[1, i\rangle)$. We conclude that $\br{\cT'}(w) = \br{\cT}(w) \setminus \bigcup_{i < |w|} \br{\cT}(w[1, i\rangle)$.

To show that $\cT'$ is unambiguous, consider a $w\in \wnS$. Let $\mu \in \br{\cT'}(w)$ and consider two accepting runs $\rho_1$ and $\rho_2$ such that $\out(\rho_1) = \out(\rho_2) = \mu$. Let us build a run $\rho$ of $\cT$ over $w$ as in the previous part of the proof, which is the same for $\rho_1$ and $\rho_2$ since $\cT$ is unambiguous. This implies that both $\rho_1$ and $\rho_2$ contain the same sequence of states in $Q$. Suppose now that the runs are different, which is only possible if at some step $i$, the $i$-th state in $\rho_1$ and $\rho_2$ are the same, and in the $(i+1)$-th state in $\rho_1$ and $\rho_2$ are different. This cannot the case since from the construction of $\cT'$, for a given transition $t\in \Delta$ that starts in a state $p$, and some $k\in\{1, 2\}$, there exists exactly one transition $t'\in\Delta'$ that starts in $(p,k)$. This is a contradiction, so we prove that $\cT'$ is unambiguous.

\subsection{Proof of Proposition  \ref{alg:spacebound}}

{\bf Part 1.} This proof is a corollary of Theorem 4.5 in~\cite{BarYossefFJ07}. The proof of this result implies that for the XPath query $Q = {\tt / / a[b\ and\ c]}$, any streaming algorithm that verifies if an XML document matches $Q$ (the problem $\textsc{booleval}_Q$) and any integer $r \geq 1$, there exists a document of depth at most $r + C$, where $C$ is a constant value, on which the algorithm requires $\Omega(r)$ bits of space.

Our proof will show a \vpa $\cA$ which can simulate the query $Q$ for a direct mapping $\nu$ of the documents that are constructed in~\cite{BarYossefFJ07}, where $\nu({\tt \langle a \rangle}) = \op{a}$, $\nu({\tt \langle / a \rangle}) = \cl{a}$, $\nu({\tt \langle / b \rangle}) = b$, and $\nu({\tt \langle / c \rangle}) = c$. The \vpa is shown in Figure~\ref{fig-vpt-depthlowerbound}. We convert this \vpa into a \vpt $\cT$ by adding an $\eps$ output symbol on each transition, so the problem of deciding if $\cA$ accepts $w$ is equivalent to deciding if $\br{\cT}(w)$ is empty, or the set $\{\eps\}$. The theorem follows by taking this $\cT$ as the one in the statement, considering an arbitrary streaming evaluation algorithm $\enumE$ that solves \enumvpt{} with input $\cT$, and using this algorithm along with the mapping $\nu$ to solve $\textsc{booleval}_Q$.

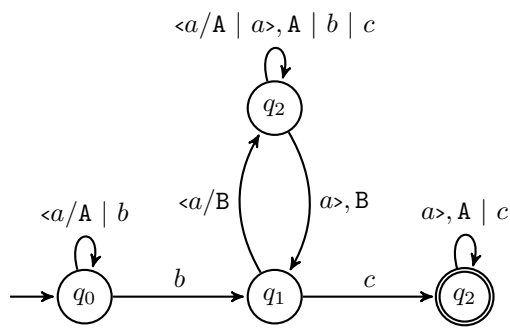
\begin{figure}[t]
	\centering
	
	\begin{tikzpicture}[scale=1,->,>=stealth',shorten >=1pt,auto,node distance=2cm,thick,state/.style={circle,draw}]
		\node[state,draw=none,scale=0.1] (in) at (-1,0) {};
		\node[state] (q0) at (0,0) {$q_0$};
		\node[state] (q1) at (2.5,0) {$q_1$};
		\node[state] (q2) at (2.5,2.5) {$q_2$};
		\node[state, accepting] (q3) at (5,0) {$q_2$};		
		\draw (in) to (q0);
		\draw (q0) to[loop above] node [above, align=center] {$\op{a} / {\tt A}$ | $b$} (q0);
		\draw (q0) to node [above] {$b$} (q1);
		\draw (q1) to[out=120,in=-120] node [left] {$\op{a} / {\tt B}$} (q2);
		\draw (q2) to[loop above] node [above, align=center] {$\op{a}/ {\tt A}$ | $\cl{a}, {\tt A}$ | $b$ | $c$} (q2);
		\draw (q2) to[out=-60,in=60] node [right] {$\cl{a}, {\tt B}$} (q1);
		\draw (q1) to node [above] {$c$} (q3);
		\draw (q3) to[loop above] node [above, align=center] {$\cl{a}, {\tt A}$ | $c$} (q3);
	\end{tikzpicture}
	
	\caption{\vpa $\cA$ used in the proof. An open transition $(p,\op{s},q,\gamma)$ is represented by an edge from $p$ to $q$ with the label $\op{s} / \gamma$. A close transition $(p,\cl{s},\gamma,q)$ is represented with the label $\cl{s}, \gamma$. An neutral transition $(p,s,q)$ is represented with the label $s$.}
	\label{fig-vpt-depthlowerbound}
\end{figure}

{\bf Part 2.} This proof uses the main ideas of the proof of Theorem 1 in~\cite{BarYossefFJ05}. Here, the authors describe a set-computing communication complexity problem. In the problem $\pazocal P$, Alice and Bob compute a two-argument function $p(\cdot, \cdot)$, defined as follows. Alice's input is a subset $A\subseteq\{1,\ldots,k\}$, Bob's input is a bit $b \in \{0,1\}$, and $p(A, b)$ is defined to be $A$, if $b = 1$, and $\emptyset$ otherwise. Proposition 1 in~\cite{BarYossefFJ05} proves that the one-way communication complexity of $\pazocal P$ is at least $k$.

Let $\cT = (Q, \Sigma, \Gamma, \oalph, \Delta, \qinit, F)$ is defined over the alphabets $\noS = \{a, b, \$\}$, and $\oalph = \{{\sf x}\}$ and have its sets $Q$ $\Delta$, $\qinit$, $F$ be as presented in Figure~\ref{fig-vpt-lowerbound}. It can be seen that it satisfies
$$
\sem{\cT}(w) = \begin{cases}
	\{({\sf x}, i) \mid w[i] = b\}	 &\text{if } w \text{ ends in } \$\\
	\emptyset &\text{otherwise}.
\end{cases}	
$$	
\begin{figure}[t]
	\centering
	
	\begin{tikzpicture}[scale=1,->,>=stealth',shorten >=1pt,auto,node distance=2cm,thick,state/.style={circle,draw}]
		\node[state,draw=none,scale=0.1] (in) at (-1,0) {};
		\node[state] (q0) at (0,0) {$q_0$};
		\node[state, accepting] (q1) at (3,0) {$q_1$};		
		\draw (in) to (q0);
		\draw (q0) to[loop above] node [above, align=center] {$a, \eps$ | $b, {\sf x}$} (q0);
		\draw (q0) to node [above] {$\$,\eps$} (q1);
	\end{tikzpicture}
	
	\caption{\vpt $\cT$ used in the proof. A neutral transition $(p, s, \oout, q)$ is represented by an edge from $p$ to $q$ labeled with $s, \oout$.}
	\label{fig-vpt-lowerbound}
\end{figure}
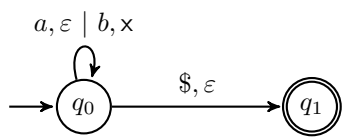

Consider an arbitrary algorithm $\enumE$ that solves $\enumvpt$ with input $\cT$. We will now present a reduction that creates a protocol for ${\pazocal P}$ which makes use of the algorithm $\enumE$. Here, Alice receives the set $A$ and generates a word $w$ of size $k$ such that $w[i] = b$ if $i\in A$ and $w[i] = a$ otherwise. Alice then executes $\enumE$ on input $\cT$ and $w$ as the first $k$ characters of a stream. She sends the state of the algorithm to Bob, who receives the bit $b$, and does the following: If $b = 1$ he continues running $\enumE$ as if the last character of the input was $\$$. If $b = 0$, he stops executing $\enumE$ immediately. In either case, the output given by $\enumE$ contains all the information necessary to compute the set $p(A,b)$, so the reduction is correct.  This proves that $\enumE$ requires at least $k$ bits for an input of size less than $k$, and so $\enumE$ for any $n \geq 1$, requires at least $n$ bits of space in a worst-case stream $\Stream$, which is in $\Omega(\outgap(\cT, S[1,n]))$.

\subsection{Proof of Proposition  \ref{prop:space}}

The time bounds are implied by Theorem~\ref{theo:main}, so we will prove the space bounds. The algorithm has a update phase and an enumeration phase, and the enumeration phase only processes the data structure that was built on the update phase, using at most linear extra space, as is explained in Section~\ref{sec:ds}. As such, we will prove that Algorithm~\ref{alg:preprocessing} on input $(\cT, w)$ uses $\cO((\depth(w) + \outgap(\cT,w))\times|Q|^2|\Delta|)$ space at every point in its execution, which implies the statement of the proposition, where $w = \Stream[1,n]$ for some stream $\Stream$ and $n$.

As it is explained in Section~\ref{sec:eval}, Algorithm~\ref{alg:preprocessing} uses a hash table $S$, and a stack $T$ that stores hash tables. The size of the stack at each point is bounded $\depth(w)$, and the size of each hash table is bounded by $|Q|^2|\Gamma|$, so the size of $S$ and $T$ combined is in $\cO(\depth(w)|Q|^2|\Delta|)$. The rest of the space used is related to the \dsepsabbr $\cD$, which we will now bound by $\cO(\outgap(\cT,w[1,k])|Q|^2|\Delta|)$ at each step $k$.

For every step $k$ of the algorithm, consider an \dsepsabbr $\cD^{\textsf{trim}}_k$ which is composed solely of the nodes that are reachable from of the ones stored in $S^k$, or the ones stored in some hash table in $T^k$ (borrowing the notation from Section~\ref{sec:eval}). A simple induction argument on $k$ shows that the rest of the nodes in $\cD$ can be discarded with no effect over the correctness of the algorithm, so they are not considered in the memory used by it. Therefore, proving that at each step $| \cD^{\textsf{trim}}_k | \in \cO(\outgap(\cT,w[1,k])|Q|^2|\Delta|)$ is enough to complete the proof. 

Let $\cI$ be the set of positions less than $k$ that appear in some output of $\sem{\cT}(w[1,k] \cdot w')$ for some $w \cdot w'\in \pwnS$. We now refer to Lemma~\ref{vpt:steps} since it implies that for each node $v$ stored in $S^k$ or the topmost hash table in $T^k$, each sequence in $\L_{\D}(v)$ corresponds to at least one valid run of $\cT$ over $w[1,k]$, and since $\cT$ is trimmed, each one of these runs is part of an accepting run of $\cT$ over $w[1,k] \cdot w'$, for some word $w'$. Therefore, each of the positions that appear in some of these sets is in $\cI$. Furthermore, we can use this lemma to characterize the positions in the rest of the hash tables in $T^k$, since appending any close symbol $\cl{a}$ to $w[1,k]$ will make the algorithm pop an element from $T$, which will make the next hash table the topmost. This argument can be extended to any of the hash tables in $T^k$, so in all, Lemma~\ref{vpt:steps} implies that all of the positions that appear in some non-empty leaf in $\cD^{\textsf{trim}}_k $ are in $\cI$. Theorem~\ref{theo:main} implies that the set of these positions corresponds exactly to $\cI$, since if there was any position in $\cI$ missing from the leaves in $\cD$, the algorithm would not be correct.

Lastly, we will show that $| \cD^{\textsf{trim}}_k | \leq | \cI | \times |Q|^2|\Delta| \times d$, where $d$ is a constant. Towards this goal, we will bound the number of $\eps$-leaves, non-empty leaves, and product nodes by $\cO(| \cI | \times |Q|^2|\Delta|)$ independently. Union nodes can be bounded by counting the other types of nodes: The only cases where a union node is created are (1) in line 35, only after a product node had been created, (2) during the creation of a product node (as described in Theorem~\ref{theo:data-structure-eps}), (3) in line 46, but only whenever one of the previous lines had created either a product node or a non-empty leaf node, and (4) in line 15, which only happens once at the end of the update phase, and iterates by nodes in $S$, so the number of union nodes created at this for loop at most $| \cD^{\textsf{trim}}_k |$. The number of $\eps$-nodes is at most one, owing to Theorem~\ref{theo:data-structure-eps}, since its proof shows that at the end of step $k$, each of the nodes in $\cD^{\textsf{trim}}_k$ is $\eps$-safe. The number of non-empty leaves can be straightforwardly shown to be $\cO(| \cI | \times |Q|^2|\Delta|)$ since each of these leaves was introduced in some step in $\cI$, and in each one of these steps, the number of operations that the algorithm does is in $\cO(|Q|^2|\Delta|)$. 

To show a bound over the number of product nodes, consider a slight modification of Algorithm~\ref{alg:preprocessing}: 
product nodes that are created in line 44 are labeled with the step $k$ in which the algorithm is at the moment. 
Now, for a set of nodes $A$ let $\cD^{\textsf{trim}}_A$ be the \dsepsabbr that is obtained by removing all of the nodes that are not reachable from some node in $A$ from $\cD$.
Let $\cI_A$ be the set of positions that appear in some non-empty leaf node in $\cD^{\textsf{trim}}_A$, and let ${\pazocal P}_A$ be the set of step labels that appear in some product node in $\cD^{\textsf{trim}}_A$ excluding the steps in $\cI_A$. Also, let $V_k$ be the the set of nodes in $\cD^{\textsf{trim}}_k$. We will show by induction on $k$ that $| {\pazocal P}_A | \leq | \cI_A | - 1$ for any $A\subseteq V_k$ which contains at least one node that is not an $\eps$-node. Consider any set $A \subseteq V_k$. The first observation we make here is that we can partition the nodes in $A$ to a collection $\{A_H\}$ of sets of nodes depending on the hash table $H$ they are reachable from, given that they are in $\cD^{\textsf{trim}}_k$. Let ${\pazocal Q}_A = {\pazocal P}_A \cup \cI_A$. From Lemma~\ref{vpt:steps} we get that for two different sets $A_{H_1}$ and $A_{H_2}$ in the collection, the sets ${\pazocal Q}_{H_1}$ and ${\pazocal Q}_{H_2}$ are disjoint. Therefore, in step $k$, if the algorithm enters $\textsc{CloseStep}$, we only need to focus on the set $A_{S}$, and if the algorithm enters $\textsc{OpenStep}$ on the set and $A_{T^k}$ (note that in this case, $S^k$ is composed only of $\eps$-nodes). The rest of the hash tables were reachable  on a previous step, so the inequality can be reached by adding up the inequalities that held in those steps. First, note that if none of the product nodes in $A$ were created in step $k$, then we can consider the set $B$ of nodes reachable from $A$ that were created in a previous step and notice that ${\pazocal P}_A = {\pazocal P}_B$ and $\cI_B \subseteq \cI_A$, so the statement follows since $B \subseteq V_{k-1}$. Also, note that if the algorithm in step $k$ enters $\textsc{OpenStep}$, all of the product nodes created in this step are directly connected to a non-$\eps$ leaf created in this same step, so the statement also follows. From this point on, we can assume that the algorithm enters $\textsc{CloseStep}$ on step $k$, and all of the nodes in $A$ are reachable from some node in $S^k$, and there is at least one product node in $A$ that was created in step $k$. Let $P$ be the set of product nodes in $A$ that were created on step $k$. Consider the span $\clevel(k) = \spanc{j}{k}$. The \textsf{prod} operation in line 44 either creates a new product node, or makes $v$ reference a node that already existed in $S^{k-1}$ or the topmost table in $T^j$. Furthermore, if a product node is created in line 44, then Theorem~\ref{theo:data-structure-eps} tells us that it must be connected to a node in $S^{k-1}$ that is not an $\eps$-node, and to a node in the topmost table in $T^j$ that is also not an $\eps$-node. Consider now the set of nodes $B$ that is made up of (1) nodes in $A$ that are reachable from $S^{k-1}$ and (2) nodes in $S^{k-1}$ that are connected to a product node in $P$. Consider also the set of nodes $C$ that is made up of (1) nodes in $A$ that are reachable from the topmost table in $T^j$, and nodes in the topmost table in $T^j$ that are connected to a node in $P$. Note that both sets $B$ and $C$ contain a non-$\eps$ node, and are composed of nodes created in a previous step, so assume that $|{\pazocal P}_B| \leq |\cI_B| -1$ and that $|{\pazocal P}_C| \leq |\cI_C| -1$. It can be seen that every node in $\cD^{\textsf{trim}}_A$ is either in $B$, $C$, or was created on step $k$, so we get that ${\pazocal P}_A = {\pazocal P}_B \cup {\pazocal P}_C \cup \{k\}$ and $\cI_A \supseteq \cI_B \cup \cI_C$. From Lemma~\ref{vpt:steps} we get that ${\pazocal Q}_B$ and ${\pazocal Q}_C$ are disjoint, and putting these facts to together gives us that $| {\pazocal P}_A | =  |{\pazocal P}_B| + |{\pazocal P}_C| + 1\leq |\cI_B| + |\cI_C| - 1 \leq |\cI_A|-1$.

Having proven this statement, we can deduce that the number of product nodes in $\cD^{\textsf{trim}}_k$ is in $\cO(| \cI | \times |Q|^2|\Delta|)$ since the number of steps where they are created is bounded by $|\cI|$. Therefore, $| \cD^{\textsf{trim}}_k | \leq | \cI | \times |Q|^2|\Delta| \times d$, for some constant $d$. This concludes the proof.

\subsection{Counterexample that the algorithm is not instance optimal}
In this section, we show a \vpt for which only logarithmic space in $\outgap(\cT,w)$ is enough for any stream $\Stream$. Let $\oout$ be any output symbol and consider a \vpt $\cT$ for which the output set is $\sem{\cT}(w)  = \{\{(\oout, i)\}\mid 1\leq i\leq \vert w \vert \}$ if the last symbol in $w$ is $\$$ and the empty set otherwise. Clearly, the \ogapname of any $w$ with respect to $\cT$ is linear in $\vert w\vert$. However, one could design a streaming evaluation algorithm that has only a counter that stores the length of the input so far, and produces the correct output set after reading the last symbol in $w$. The enumeration phase can easily be done with output-linear delay (i.e., by counting from $1$ to $\vert w\vert$). This completes the example.

\renewcommand{\atitle}{\ref{sec:ds}}
\section{Proofs from Section~\atitle}\label{sec:appendixds}

\subsection{Proof of Proposition \ref{ds:lindelay}}

Let $\D = (\Sigma, V, I, \ell, r, \lambda)$ be a $k$-bounded \dsabbr\ and $v\in V$. We will show that the set $\L_{\D}(v)$ can be enumerated with output-linear delay. To show that this is possible we use a data structure we call an {\em output tree}. This is a dynamic binary tree $T$ which appends itself to an \dsabbr \ $\cD$. We define it as follows: If $v$ is a leaf node in $\cD$, then $v$ is an output tree of $\cD$. If $T$ is an output tree and $v$ is a union node, then $T' = v(T)$ is an output tree of $\cD$. If $T_1$ and $T_2$ are output trees and $v$ is a product node, then $v(T_1,T_2)$ is an output tree of $\cD$. In either case, we say that $T$ is rooted in $v$, and we notate it by ${\sf root}(T) = v$. For an output tree $T$ we define the functions ${\sf child}_T, {\sf lchild}_T$ and ${\sf rchild}_T$ as follows: If $v(T')$ is a subtree of $T$, then ${\sf child}_T(v) = T'$. If $v(T_1,T_2)$ is a subtree of $T$, then ${\sf lchild}_T(v) = T_1$ and ${\sf rchild}_T(v) = T_2$. These functions are not defined in any other case.

\begin{definition}
	Let $\D = (\Sigma, V, I, \ell, r, \lambda)$ be an \dsabbr. An output tree $T$ of $\cD$ is full if for each node $v$ in $T$ the following hold: If $v$ is an union node in $\cD$, then ${\sf child}_T(v)$ is either rooted in $\ell(v)$ or in $r(v)$. If $v$ is a product node in $\cD$, then ${\sf lchild}_T(v)$ is rooted in $\ell(v)$ and ${\sf rchild}_T(v)$ is rooted in $r(v)$.
\end{definition}

We define the function ${\sf print}(T)$ as follows: If $v$ is a leaf node $v$, then ${\sf print}(T) = \lambda(v)$. If $T = v(T')$ then ${\sf print}(T) = {\sf print}(T')$. If $T = v(T_1, T_2)$ then ${\sf print}(T) = {\sf print}(T_1)\cdot {\sf print}(T_2)$.

\begin{lemma}\label{appendix:output-tree-print}
	Let $\cD$ be an \dsabbr and let $v$ be a node in $\cD$. For a full output tree $T$ of $\cD$ rooted on $v$ it holds that ${\sf print}(T) \in \L_{\cD}(v)$.
\end{lemma}
\begin{proof}
	We prove this by induction on the size of $T$. The case $T = v$ where $v$ is a leaf node is trivial. If $T = v(T')$, $v$ is an union node, so the proof follows since ${\sf print}(T)$ is equal to ${\sf print}(T')$ which is either in $\L(\ell(v))$ or $\L(r(v))$, and therefore in $L(v)$. If $T = v(T_1,T_2)$ then $v$ is a product node. We have that ${\sf print}(T_1)\in\L(\ell(v))$ and ${\sf print}(T_2)\in\L(r(v))$, from which it follows that ${\sf print}(T)\in\L(v)$.
\end{proof}

\begin{lemma}\label{appendix:output-tree-unique}
	Let $\cD$ be an unambiguous \dsabbr and let $v$ be a node in $\cD$. For each $\mu\in\L_{\cD}(v)$ there exists exactly one full output tree $T_{\mu}$ of $\cD$ rooted in $v$ such that ${\sf print}(T) = \mu$.
\end{lemma}
\begin{proof}
	Let ${\sf reach_{\cD}}(v)$ be the number of nodes reachable from $v$ in $\cD$, including itself. We will prove this lemma by induction in ${\sf reach_{\cD}}(v)$. If ${\sf reach_{\cD}}(v) = 1$, then $v$ is a leaf node and the proof follows directly since the only output tree rooted in $v$ is $v$ itself. Assume that it holds for every node $v$ such that ${\sf reach_{\cD}}(v)< s$. Let $v$ be a node such that ${\sf reach_{\cD}}(v) = s$ and let $\mu\in\L(v)$. If $v$ is a union node suppose without loss of generality that $\mu\in\L(\ell(v))$. Note that since $\cD$ is unambiguous we have that $\mu\not\in\L(r(v))$. If $T_{\mu} = v(T')$ and $T'$ was rooted in $r(v)$, Lemma~\ref{appendix:output-tree-print} would imply that ${\sf print}(T_{\mu}) = {\sf print}(T') \in \L(r(v))$ which leads to a contradiction. Therefore, $T_{\mu}$ could only be of the form $v(T')$ where $T'$ is rooted in $\ell(v)$. From our hypothesis, there exists only one full output tree $T_{\mu}'$ such that ${\sf print}(T_{\mu}') = \mu$, so the proof follows from taking $T_{\mu} = v(T_{\mu}')$. If $v$ is a product node note that any full output tree $T$ rooted in $v$ is of the form $v(T_1,T_2)$, where $T_1$ and $T_2$ are rooted in $\ell(v)$ and $r(v)$ respectively. Since $\cD$ is unambiguous, there exists only two strings $\mu_1$ and $\mu_2$ such that $\mu = \mu_1\cdot\mu_2$ and $\mu_1\in\L(\ell(v))$ and $\mu_2\in\L(r(v))$. Let $T_{\mu_1}$ and $T_{\mu_2}$ be the only full output trees that are rooted in $\ell(v)$ and $r(v)$ respectively for which the hypothesis hold. The proof follows by taking $T_{\mu} = v(T_{\mu_1},T_{\mu_2})$.
\end{proof} 

For an \dsabbr $\cD$ and node $v$ we define a total order over the full output trees rooted in $v$ recursively: If $v$ is a leaf node there exists only one tree rooted in $v$ so the order is trivial. If $v$ is a union node then let $T_1 = v(T_1')$ and $T_2 = v(T_2')$ be full output trees. We have that $T_1 < T_2$ if and only if ${\sf root}(T_1') = \ell(v)$ and ${\sf root}(T_2') = r(v)$, or $T_1' < T_2'$. If $v$ is a product node then let $T = v(T_1,T_2)$ and $T' = v(T_1',T_2')$. We have that $T < T'$ if and only if $T_1 < T_1'$, or $T_1 = T_1'$ and $T_2 < T_2'$.

For an \dsabbr $\cD$ and an output tree $T$ on $\cD$ we define the operation ${\sf tilt}(T)$ as follows: If $T = v$, then ${\sf tilt}(T) = v$. If $T = v(T')$ where ${\sf root}(T') = \ell(v)$, then ${\sf tilt}(T) = v({\sf tilt}(T))$. If $T = v(T')$ where ${\sf root}(T') = r(v)$, then ${\sf tilt}(T) = {\sf tilt}(T')$. If $T = v(T_1,T_2)$, then ${\sf tilt}(T) = v({\sf tilt}(T_1),{\sf tilt}(T_2))$. Intuitively, what this operation does is to bypass any union node in $T$ whose child is a right child in $\cD$.

\begin{definition}
	For an \dsabbr $\cD$, an output tree $T$ of $\cD$ is left-tilted if it can be obtained as $T = {\sf tilt}(T')$ where $T'$ is a full output tree.
\end{definition}

Two left-tilted output trees can be seen in Figure~\ref{fig-output-trees}. The first tree in the figure is also full. Note that since the root could be a union node whose child is a right child, the root of ${\sf tilt}(T)$ could be a different node than the root of $T$. We also notice the following result.

\begin{lemma}\label{appendix:enum-first}
	Let $\cD$ \dsabbr with a node $v$. The first tree $T$ in the ordered sequence of full output trees rooted in $v$ is also left-tilted. In other words, ${\sf tilt}(T) = T$.
\end{lemma}
\begin{proof}
	We define the operation ${\sf build}(v)$ as follows. If $v$ is a leaf node, then ${\sf build}(v) = v$. If $v$ is a union node then ${\sf build}(v) = v({\sf build}(\ell(v))$. If $v$ is a product node then ${\sf build}(v) = v({\sf build}(\ell(v)),{\sf build}(r(v)))$. Let $T'$ be a different full output tree rooted in $v$. A straightforward induction shows that $T < T'$.
\end{proof}

\begin{lemma}
	Let $\D$ be an \dsabbr with an output tree $T$. We have that ${\sf print}({\sf tilt}(T)) = {\sf print}(T)$.
\end{lemma}
\begin{proof}
	The proof follows by a straightforward induction on the tree. 
\end{proof}


\begin{algorithm}[t]
	\caption{Enumeration of the set $\L_{\D}(v)$ for a CE $\D$ and a node $v$.}\label{alg:dsenum2}
	\begin{varwidth}[t]{0.6\textwidth}
		\begin{algorithmic}[1]
			\Procedure{{Enumerate}}{$\D, v$}
			\State $T\gets\textsc{BuildTree}(\D, v)$
			\State ${\tt Output}\ \#$
			\While{$T\neq\emptyset$}
			\State ${\tt Output}$ ${\sf print}(T)$
			\State ${\tt Output}\ \#$
			\State $T \gets \textsc{NextTree}(D,T)$
			\EndWhile
			\EndProcedure
			\Procedure{{BuildTree}}{$\D, v$}
			\If{$\lambda(v)\in\Sigma$}
			\State {\bf return} $v$ 
			\ElsIf{$\lambda(v) = \odot$}
			\State $T_1\gets \textsc{BuildTree}(\D,\ell(v))$
			\State $T_2\gets \textsc{BuildTree}(\D,r(v))$
			\State {\bf Return} $v(T_1,T_2)$
			\ElsIf{$\lambda(v) = \cup$}
			\State $T \gets \textsc{BuildTree}(\D,\ell(v))$
			\State {\bf Return} $v(T)$
			\EndIf
			\EndProcedure
			\algstore{myalg}
		\end{algorithmic}	
	\end{varwidth}
	\hspace{1em}
	\begin{varwidth}[t]{0.6\textwidth}
		\begin{algorithmic}[1]
			\algrestore{myalg}
			\Procedure{{NextTree}}{$\D,T$}
			\If{$T = v$}
			\State {\bf return} $\emptyset$ 
			\ElsIf{$T = v(T_1,T_2)$}
			\State $T_2\gets\textsc{NextTree}(\D,T_2)$
			\If{$T_2$ is empty}
			\State $T_1\gets\textsc{NextTree}(\D,T_1)$
			\If{$T_1$ is empty}
			\State {\bf return} $\emptyset$
			\EndIf
			\State $T_2 \gets\textsc{BuildTree}(\D,r(v))$
			\EndIf
			\State {\bf return} $T$
			\ElsIf{$T = v(T')$}
			\State $T'\gets\textsc{NextTree}(\D,T')$
			\If{$T' = \emptyset$}
			\State $T\gets\textsc{BuildTree}(\D,r(v))$
			\EndIf
			\State {\bf return} $T$
			\EndIf
			\EndProcedure
		\end{algorithmic}
	\end{varwidth}
\end{algorithm}

\begin{figure}[t]
	\centering
	\begin{tikzpicture}[->,>=stealth',roundnode/.style={circle,draw,inner sep=1.2pt},squarednode/.style={rectangle,inner sep=3pt}]
		\node [squarednode] (0) at (0, 0) {$a_1$};
		\node [squarednode] (1) at (2, 0) {$a_2$};
		\node [squarednode] (2) at (4, 1) {$a_3$};
		\node [squarednode] (3) at (1, 1) {$\cup$};
		\node [squarednode] (4) at (2, 2) {$\odot$};
		\node [squarednode] (5) at (3, 3) {$\cup$};
		\node [roundnode] (6) at (0, 0.3) {};
		\node [roundnode] (8) at (4, 1.3) {};
		\node [roundnode] (9) at (1, 1.3) {};
		\node [roundnode] (10) at (2, 2.3) {};
		\node [roundnode] (11) at (3, 3.3) {};
		\draw (3.south west) to (0);
		\draw (3.south east) to (1);
		\draw (4.south west) to (3);
		\draw (4.south east) to (2);
		\draw (5.south west) to (4);
		\draw (5.south east) to (2);
		\draw [dashed, -] [out=200,in=70] (9) to (6);
		\draw [dashed, -] [out=200,in=70] (10) to (9);
		\draw [dashed, -] [out=-20,in=130] (10) to (8);
		\draw [dashed, -] [out=200,in=70] (11) to (10);
	\end{tikzpicture}
	\begin{tikzpicture}[->,>=stealth',roundnode/.style={circle,draw,inner sep=1.2pt},squarednode/.style={rectangle,inner sep=3pt}]
		\node [squarednode] (0) at (0, 0) {$a_1$};
		\node [squarednode] (1) at (2, 0) {$a_2$};
		\node [squarednode] (2) at (4, 1) {$a_3$};
		\node [squarednode] (3) at (1, 1) {$\cup$};
		\node [squarednode] (4) at (2, 2) {$\odot$};
		\node [squarednode] (5) at (3, 3) {$\cup$};
		\node [roundnode] (7) at (2, 0.3) {};
		\node [roundnode] (8) at (4, 1.3) {};
		\node [roundnode] (10) at (2, 2.3) {};
		\node [roundnode] (11) at (3, 3.3) {};
		\draw (3.south west) to (0);
		\draw (3.south east) to (1);
		\draw (4.south west) to (3);
		\draw (4.south east) to (2);
		\draw (5.south west) to (4);
		\draw (5.south east) to (2);
		\draw [dashed, -] [out=-150,in=110] (10) to (7);
		\draw [dashed, -] [out=-20,in=130] (10) to (8);
		\draw [dashed, -] [out=200,in=70] (11) to (10);
	\end{tikzpicture}
	\caption{An example iteration of an output tree. The subjacent \dsabbr\ $\D$ is represented by solid edges, and the output tree with curve dashed lines. The next tree would be the single node $v$ for which $\lambda(v) = a_3$.}
	\label{fig-output-trees}
\end{figure}
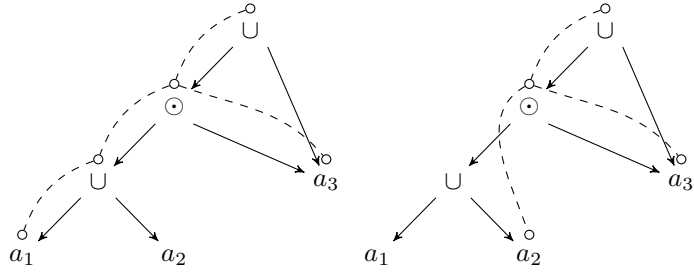

We are ready to discuss the enumeration algorithm. Our algorithm receives an unambiguous $k$-bounded \dsabbr $\cD$ along with one of its nodes $v$ and prints the elements in $\L_{\cD}(v)$ one by one. The way this is done is by generating the sequence of left-tilted output trees ${\sf tilt}(T_1),\ldots,{\sf tilt}(T_m)$ for which $T_1 < \cdots < T_m$ is the complete sequence of full output trees rooted in $v$. After generating each tree $T$, the procedure outputs the string ${\sf print}(T)$ which can be easily done with a depth-first traversal on the tree. The procedure is detailed in Algorithm~\ref{alg:dsenum2}.

The procedure {\sc BuildTree} builds a completely embedded output tree rooted in $u$. The procedure {\sc NextTree} receives a tree rooted in $u$ and recursively builds the next tree in the sequence ${\sf tilt}(T_1),\ldots,{\sf tilt}(T_m)$ for which $T_1 < \cdots < T_m$ is the sequence of full output trees rooted in $u$. 

We can deduce the following from Lemma~\ref{appendix:enum-first}:
\begin{corollary}
	Let $\D$ be an \dsabbr and let $v$ be one of its nodes. {\sc BuildTree}$(\cD,v)$ builds a full output tree $T$ that is the first in the ordered sequence of full output trees rooted in $v$.
\end{corollary}

We prove the correctness of the algorithm in the following results.

\begin{lemma}\label{appendix:enum-correctlemma}
	Let $\D$ be an \dsabbr and let $v$ be one of its nodes. Let $T_1<\ldots<T_m$ be the sequence of full output trees rooted in $v$. If the procedure {\sc NextTree} receives $(\cD,{\sf tilt}(T_i))$ it returns ${\sf tilt}(T_{i+1})$, or $\emptyset$ if $i = m$.
\end{lemma}
\begin{proof}
	We prove this by induction in ${\sf reach}_{\cD}(v)$. If $v$ is a leaf node, the sequence consists only of the tree $v$, so the proof follows directly. 
	Assume it holds for nodes $v'$ such that ${\sf reach}_{\cD}(v') < s$ and let $v$ be such that ${\sf reach}_{\cD}(v) = s$. 
	If $v$ is a union node notice that there exists an $e$ such that the sequence of full output trees rooted in $v$ is $T_1<\ldots<T_e<T_{e+1}<\ldots<T_m$ where $T_e = v(T_e')$ and $T_{e+1} = v(T_{e+1}')$, and ${\sf root}(T_e') = \ell(v)$ and ${\sf root}(T_{e+1}') = r(v)$. 
	If $i < e$ or $i > e$, then the proof follows by induction. Otherwise, if $i = e$, note that the procedure {\sc BuildTree}$(\cD,r(v))$ builds the first full output tree rooted in $r(v)$, which is $T_{e+1}'$, and is equal to ${\sf tilt}(T_{e+1})$. 
	If $v$ is a product node the proof follows straightforwardly by induction over the algorithm.
\end{proof}

From the previous results, correctness of the algorithm follows:

\begin{claim}
	{\sc Enumerate} receives an \dsabbr $\D$ and one of its nodes $v$ and outputs all of the elements in $\L_{\cD}(v)$ one by one without repetition.
\end{claim}
\begin{proof}
	Let $T_1 <\cdots <T_m$ be the sequence of full output trees rooted in $v$. The algorithm starts by generating $T_1 = {\sf tilt}(T_1)$ as proven by Corollary~\ref{appendix:enum-first}. Then on each step $i$, the algorithm iterates $T$ as ${\sf tilt}(T_i)$ to transform it into ${\sf tilt}(T_i)$, as proven by Lemma~\ref{appendix:enum-correctlemma}. In each step, an element in $\L_{\cD}(v)$ is given as output as proven by Lemma~\ref{appendix:output-tree-print}. Moreover, the sequence $T_1 <\cdots <T_m$ allows the set $\L_{\cD}(v)$ to be produced exhaustively without repetitions, as proven by Lemma~\ref{appendix:output-tree-unique}.
\end{proof}

The following results ensure that each tree in the sequence can be generated efficiently.

\begin{lemma}\label{appendix:next-tree-almost-linear-delay}
	Let $\cD$ be an \dsabbr, let $v$ one of its nodes, and let $T_1<\cdots<T_m$ be the sequence of full output trees rooted in $v$. If the procedure {\sc NextTree} receives $(\cD,T)$ it returns the tree $T'$ in at most $c(\vert T\vert + \vert T'\vert)$ time, for some constant $c$.
\end{lemma}
\begin{proof}
	We choose $c$ as a factor of the number of steps that are taken in {\sc NextTree} without taking into account recursion. That is, the time that it takes to run steps 19-33 without calls. A first observation that we make is that {\sc BuildTree} builds a tree $T$  in time at most $c\vert T\vert$, since each call to {\sc BuildTree} takes less than $c$ steps, and exactly one call to {\sc BuildTree} is done per node in $T$. We prove the lemma by induction on the tree. If $v$ is a leaf node, then the proof is trivial. If $v$ is a product node, let $T = v(T_1,T_2)$, and let $T'$ be the output of {\sc NextTree} such that $T' = v(T_1',T_2')$ or $T' = \emptyset$. If the call in line 22 returns an nonempty tree, then the procedure takes time $c + c(\vert T_2\vert + \vert T_2'\vert)$. Otherwise, line 22 takes time $c\vert T_2\vert$. Then, if the call in line 24 returns a nonempty tree, it takes time $c(\vert T_1\vert + \vert T_1'\vert)$, and then the call in line 27 takes time $c\vert T_2'\vert$; otherwise, it takes time $c\vert T_1\vert$. In each of the routes where $T'$ is not empty, the execution time is bounded by $c(\vert T_1\vert + \vert T_2\vert + \vert T_1'\vert + \vert T_2'\vert+1) \leq c(\vert T\vert + \vert T' \vert)$, and if $T' = \emptyset$, it is bounded by $c(\vert T_1\vert + \vert T_2\vert + 1) = c\vert T\vert$ which proves the statement. If $v$ is a union node, let $T = v(T')$ and let $T_{\out}$ be the output of {\sc NextTree}. If the call in line 30 returns a nonempty tree, it takes time $c(\vert T'\vert + \vert T_{\out}'\vert)$, where $T_{\out} = v(T_{\out}')$, and the procedure takes total time $c + c(\vert T'\vert + \vert T_{\out}'\vert)\leq c(\vert T\vert + \vert T_{\out}\vert)$, which proves the statement. Otherwise,  the call in line 30 takes time $c\vert T'\vert$, and then line 32 takes time $c\vert T_{\out}\vert$, which adds to a total time $c + c(\vert T'\vert + \vert T_{\out}\vert)= c(\vert T\vert + \vert T_{\out}\vert)$, which also proves the statement.
\end{proof}

\begin{lemma}\label{appendix:tree-size}
	Let $\cD$ be a $k$-bounded \dsabbr and $T$ be a left-tilted output tree in $\cD$. The size of $T$ is at most $2k\vert {\sf print}(T)\vert$.
\end{lemma}
\begin{proof}
	Note that $\vert {\sf print}(T)\vert$ is equal to the number of leaves in $T$. Since $T$ is left-tilted, then  for each union node $v$ in $T$ we have that ${\sf child}_{T}(v)$ is rooted in $\ell(v)$. We also have that $\cD$ is $k$-bounded, so there are at most $k$ nodes between each pair of product nodes in $T$. We know that a binary tree with $e$ leaves has $2e-1$ nodes and $2e-2$ edges. Therefore, if we replace each edge by $k-1$ nodes we obtain a tree whose size is an upper bound for the size of $T$, and the proof follows.
\end{proof}

From these lemmas we obtain a result that ensures nearly output-linear delay.

\begin{claim}\label{appendix:ds-almost-linear-delay}
	Let $\cD$ be a $k$-bounded \dsabbr and let $v$ be a node in $\cD$. For some sequence $\mu_1,\ldots,\mu_m$ that contains exactly the elements in $\L_{\cD}(v)$ without repetition, {\sc Enumerate} can produce each element $\mu_i$ for $i\in[2,m]$ with delay $c(\vert \mu_{i-1}\vert + \vert \mu_i\vert)$, and $\mu_1$ with delay $c\vert\mu_1\vert$, where $c$ is a constant.
\end{claim}
\begin{proof}
	The sequence in question is the one given by the total order $T_1<\cdots <T_m$ of total output trees rooted in $v$, for which $\mu_i = {\sf print}(T_i)$. Let $c'$ be the constant in Lemma~\ref{appendix:next-tree-almost-linear-delay} and let $d$ be a constant such that ${\sf print}(T)$ can be produced in time $d\vert T \vert$. We have that {\sc BuildTree} can build a tree $T$ in size. In Lemma~\ref{appendix:next-tree-almost-linear-delay} it is shown that the first tree $T_1$ in the sequence can be generated in time $c'\vert T_1\vert$, and in Lemma~\ref{appendix:tree-size} we show that $\vert T_1\vert \leq 2k\vert \mu_1\vert$. From this, it follows that $\mu_1$ can be produced in time $2k(c'+d)\vert \mu_1\vert$. For each $i\in[2,m]$ Lemma~\ref{appendix:next-tree-almost-linear-delay} shows that $T_i$ can be generated in time $c'(\vert T_{i-1}\vert + \vert T_i\vert)$. We can bound this number by $2kc'(\vert \mu_{i-1}\vert + \vert \mu_i\vert)$ using Lemma~\ref{appendix:tree-size}. Printing the output takes time $d\vert T_i\vert$, so the total time is $2kc'(\vert \mu_{i-1}\vert + \vert \mu_i\vert) + 2kd\vert \mu_i\vert$, which is bounded by $2k(c'+d)(\vert \mu_{i-1}\vert + \vert \mu_i\vert)$. We conclude the proof by taking $c = 2k(c'+d)$.
\end{proof}

We optimize this result to obtain the desired statement.

\begin{proposition}[Proposition~\ref{ds:lindelay}]
	Fix $k\in\nat$. Let $\cD$ be an unambiguous and $k$-bounded \dsabbr. Then the set $\L_{\cD}(v)$ can be enumerated with output-linear delay for any node $v$ in $\cD$.
\end{proposition}
\begin{proof}
	Let $c$ be the constant from Claim~\ref{appendix:ds-almost-linear-delay}, and let $\mu_1,\ldots,\mu_m$ be the elements in $\L_{\cD}(v)$ in the order that the algorithm from Claim~\ref{appendix:ds-almost-linear-delay} produces them.  We have that $\mu_1$ can be produced in $c|\mu_1|$ steps, whereas each other $\mu_i$ can be produced in $c(|\mu_{i-1}|+|\mu_i|)$ steps. 
	Our algorithm consists in printing the output set in order in an auxiliary tape, and to simply wait $2c\cdot|\mu_i|$ steps to print each output $\mu_i$ to the actual output tape.
	To see why this is possible to do, note that each output $\mu_i$ will be printed in the auxiliary tape after at most $c|\mu_1| + c(|\mu_1|+|\mu_2|) + c(|\mu_2|+|\mu_3|) + \cdots + c(|\mu_{i-1}|+|\mu_i|)$ steps, which is less than $2c\cdot(|\mu_1|+|\mu_2|+\cdots+|\mu_i|)$. This guarantees that at the moment each output $\mu_i$ need to be printed in the output tape, it will be available in the auxiliary tape.
	Since this clearly works with output-linear delay, the statement follows.
\end{proof}

\subsection{Proof of Theorem \ref{theo:data-structure}}

%
%
%
%

The construction of the operators and the reasoning why each partial result $(\D',v')$ is 2-bounded is stated in the paper. 
By adding the condition that $\D'$ is unambiguous we can deduce that $\L_{\D'}(v')$ can be enumerated with output-linear delay using Proposition~\ref{ds:lindelay}.

\subsection{Proof of Theorem \ref{theo:data-structure-eps}}

In \dsepsabbr, $\eps$-nodes are treated quite particularly. 
For a given \dsepsabbr $\cD$, we require that any node $v\in\cD$ satisfies exactly one of the following:
(1) $\lambda(v) \neq \eps$ and for any node $u$ which is reachable from $v$ it holds that $\lambda(u) \neq \eps$, (2) $\lambda(v) = \eps$ or (3) $\lambda(v) = \cup$, $\lambda(\ell(v)) = \eps$, and $r(v)$ satisfies (1). 
In other words, $\eps$ can only be child of a union node with in-degree 0.
For the rest of the proof, we will refer to a node $v$ that satisfies each case as a node such that (1) $\eps\not\in\cL_{\cD}(v)$, (2) $\lambda(v) = \eps$ or (3) $v$ is {\em in Case 3}, respectively.
Note that this construction ensures that if $\eps\in\cL_{\cD}(v)$, it can be retrieved in constant time.

With these conditions in mind, we can address output-depth, $k$-bounded and safeness. The definition of output-depth is unchanged for nodes $v$ for which $\eps\not\in\cL_{\cD}(v)$, if $\lambda(v) = \eps$, then $\odepth(v) = 0$, and if $v$ is in Case 3, $\odepth(v) = 1$. The definition of $k$-bounded is unchanged. The definition of safe nodes is unchanged except for the additional restriction that a node $v$ can only be safe if $\eps\not\in\cL_{\cD}(v)$.
\begin{claim}\label{appendix:eps-enum}
For a $k$-bounded unambiguous \dsepsabbr $\D$, the set $\cL_{\cD}(v)$ can be enumerated with output-linear delay for every node $v$ in $\cD$.
\end{claim}
\begin{proof}
	To prove this, we formalize the idea behind the construction of an \dsepsabbr.
	Let $\D_{v}$ be the \dsepsabbr induced by the nodes that are reachable from $v$. Formally, let $V_v$ be this set of nodes. 
	Then $\D_{v} = (\Sigma, V_v, I_v, \ell_v, r_v, \lambda_v)$ where $I_v = I \cap V_v$, and also $\ell_v$, $r_v$ and $\lambda$ are the functions $\ell$, $r$ and $\lambda$ induced by $V_v$. 
	It is straightforward to check that $\L_{\D_v}(v) = \L_{\D}(v)$.
	Note that if $\eps\not\in\cL_{\cD}(v)$, then $\cD_v$ is a regular \dsabbr, and if $v$ is in Case 3, then $\cD_{r(v)}$ is a regular \dsabbr as well.
	Furthermore, if $\cD$ is unambiguous and $k$-bounded, then the \dsabbr in each of these cases is also unambiguous and $k$-bounded.
	From here, the proof follows straightforwardly by using Proposition~\ref{ds:lindelay} over these \dsabbr.
\end{proof}

One last notion we make use of is $\eps$-safe nodes. For a given \dsepsabbr $\cD$ and $v\in\cD$ we say that $v$ is $\eps$-safe if either (1) $v$ is safe, (2) $\lambda(v) = \eps$, or (3) $v$ is in Case 3 and $r(v)$ is safe.

We define the operations $\add$, $\prod$ and $\union$ over $\D$ to return a pair $(\D',v')$ such that $\D' = (\Sigma, V', I', \ell', r', \lambda')$ as follows:

For $\add(\D,a)\to(\D',v')$ we define $V' := V \cup \{v'\}$, $I' := I$, and $\lambda'(v') = a$.

Assume $v_1$ and $v_2$ are $\eps$-safe. 
Further, assume that for every word in $w\in\L_{\D}(v_1)\cdot\L_{\D}(v_2)$ there exist only two non-empty words $w_1$ and $w_2$ such that $w_1\in\L_{\D}(v_1)$, $w_2\in\L_{\D}(v_2)$ and $w = w_1w_2$.
Since both $v_1$ and $v_2$ may fall in one of three cases, we define $\prod(\D,v_1,v_2) \to (\D',v')$ by separating into nine cases, of which the first six are straightforward: 
\begin{itemize}
	\item If $\eps\not\in\L_{\D}(v_1)$ and $\eps\not\in\L_{\D}(v_2)$,  we use the construction given for a regular \dsabbr.
	\item If $\eps\not\in\L_{\D}(v_1)$ and $\lambda(v_2) = \eps$, we define $v' = v_1$, and $\D' = \D$.
	\item If $\lambda(v_1) = \eps$ and $\eps\not\in\L_{\D}(v_2)$, we define $v' = v_2$, and $\D' = \D$.
	\item If $\lambda(v_1) = \eps$ and $\lambda(v_2) = \eps$, we define $v' = v_1$, and $\D' = \D$.
	\item If $\lambda(v_1) = \eps$ and $v_2$ is in Case 3, we define $v' = v_2$, and $\D' = \D$.
	\item If $v_1$ is in Case 3 and $\lambda(v_2) = \eps$, we define $v' = v_1$, and $\D' = \D$.
\end{itemize}


\begin{figure}[t]
	\centering
\begin{tikzpicture}[->,>=stealth',roundnode/.style={circle,draw,inner sep=1.2pt},squarednode/.style={rectangle,inner sep=3pt}]
	\node [squarednode] (0) at (1, 2.75) {$\cup$};
	\node [squarednode] (6) at (1.5, 2.75) {$= v_a'$};
	\node [squarednode] (1) at (0, 1) {$v_1$};
	\node [squarednode] (2) at (2, 1.75) {$\odot$};
	\node [squarednode] (3) at (2, 1) {$v_2$};
	\node [squarednode] (4) at (1, 0) {$\eps$};
	\node [squarednode] (5) at (3, 0) {$r(v_2)$};
	\draw (0) to (1);
	\draw (0) to (2);
	\draw (2) to (5);
	\draw (3) to (4);
	\draw (3) to (5);
	\draw (2) to (1);
\end{tikzpicture}
\hspace{1em}
\begin{tikzpicture}[->,>=stealth',roundnode/.style={circle,draw,inner sep=1.2pt},squarednode/.style={rectangle,inner sep=3pt}]
		\node [squarednode] (0) at (2.5, 2.75) {$\cup$};
		\node [squarednode] (6) at (3, 2.75) {$= v_b'$};
\node [squarednode] (1) at (3, 1) {$v_2$};
\node [squarednode] (2) at (1.75, 1.75) {$\odot$};
\node [squarednode] (3) at (0.75, 1) {$v_1$};
\node [squarednode] (4) at (0, 0) {$\eps$};
\node [squarednode] (5) at (1.5, 0) {$r(v_1)$};
	\draw (0) to (1);
	\draw (0) to (2);
	\draw (2) to (5);
	\draw (3) to (4);
	\draw (3) to (5);
	\draw (2) to (1);
\end{tikzpicture}
\hspace{1em}
\begin{tikzpicture}[->,>=stealth',roundnode/.style={circle,draw,inner sep=1.2pt},squarednode/.style={rectangle,inner sep=3pt}]
		\node [squarednode] (0) at (1.5, 2.75) {$\cup$};
			\node [squarednode] (11) at (2, 2.75) {$= v_c'$};
\node [squarednode] (1) at (0.75, 2.25) {$\eps^*$};
\node [squarednode] (2) at (2.25, 2.25) {$\textsf{union}$};
\node [squarednode] (3) at (3, 1.75) {$\cup$};
\node [squarednode] (4) at (2, 1) {$\odot$};
\node [squarednode] (5) at (3, 1) {$v_2$};
\node [squarednode] (6) at (2.5, 0) {$\eps$};
\node [squarednode] (7) at (4, 0) {$r(v_2)$};
\node [squarednode] (8) at (0.5, 1) {$v_1$};
\node [squarednode] (9) at (0, 0) {$\eps$};
\node [squarednode] (10) at (1, 0) {$r(v_1)$};
		\draw (0) to (1);
\draw (0) to (2);
\draw (2) to (10);
\draw (4) to (10);
\draw (4) to (7);
\draw (5) to (6);
\draw (5) to (7);
\draw (3) to (7);
\draw (3) to (4);
\draw (2) to (3);
\draw (8) to (10);
\draw (8) to (9);
\end{tikzpicture}
	\caption{Gadgets for $\prod$ as defined for an \dsepsabbr. Nodes $v_a'$, $v_b'$ and $v_c'$ correspond to $v'$ as is defined for cases (a), (b) and (c) respectively.}
	\label{fig-prod-multi-gadget}
\end{figure}
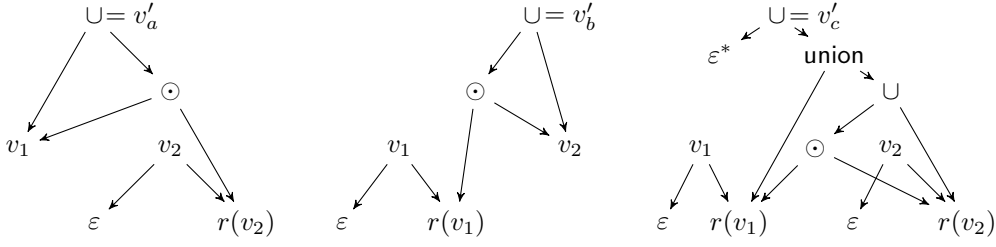

The other three cases are more involved and they are presented graphically in Figure~\ref{fig-prod-multi-gadget}. 
Formally, they are defined as follows:

\begin{itemize}
	\item[(a)] If $\eps\not\in\L_{\D}(v_1)$ and $v_2$ is in Case 3, then $V' = V\cup\{v',v''\}$, $I' = I\cup\{v',v''\}$, $\ell'(v') = v_1$, $r'(v') = v''$, $\ell(v'') = v_1$,  $r(v'') = r(v_2)$, $\lambda'(v') = \cup$ and $\lambda'(v'') = \odot$. 
	\item[(b)] If $v_1$ is in Case 3 and $\eps\not\in\L_{\D}(v_1)$, then $V' = V\cup\{v',v''\}$, $I' = I\cup\{v',v''\}$, $\ell'(v') = v''$, $r'(v') = v_2$, $\ell(v'') = r(v_1)$,  $r(v'') = v_2$, $\lambda'(v') = \cup$ and $\lambda'(v'') = \odot$. 
	\item[(c)] If both $v_1$ and $v_2$ are in Case 3, we do a slightly more delicate construction. 
	First, we define a $\cD''$ with $V'' = V\cup\{v^{3},v^{4}\}$, $I'' = I\cup\{v^{3},v^{4}\}$, $\ell''(v^3) = v^4$, $r''(v^3) = r(v_2)$, $\ell''(v^4) = r(v_1)$, $r''(v^4) = r(v_2)$, $\lambda''(v^3) = \cup$, $\lambda''(v^4) = \odot$.
	Now, let $(\cD^3, v^2) \gets \union(\cD'', r(v_1), v_3)$.
	Lastly, let $V' = V^3 \cup \{v^{*}, v'\}$, $I' = I^3\cup\{v'\}$, $\ell'(v') = v^*$, $r(v') = v_2$, $\lambda(v') = \cup$ and $v^* = \eps$.
\end{itemize}
Note that the $\union$ operation in case (c) does not recurse since $r(v_1)$ is safe. In particular, it does not reach any $\eps$-leaf.

Assume $v_1$ and $v_2$ are $\eps$-safe nodes. 
Further, assume that $\L_{\D}(v_1) \setminus \{\eps\}$ and $\L_{\D}(v_2)\setminus \{\eps\}$ are disjoint. 
We define $\union(\D,v_1, v_2) \to (\D',v')$ as follows:
\begin{itemize}
	\item If $\eps\not\in\L_{\D}(v_1)$ and $\eps\not\in\L_{\D}(v_2)$, we use the construction given for a regular \dsabbr.
	\item If $\eps\not\in\L_{\D}(v_1)$ and $\lambda(v_2) = \eps$, we define $V' = V \cup\{v'\}$, $I' = I\cup\{v'\}$ and $\lambda(v') = \cup$. We connect $\ell(v') = v_2$ and $r(v') = v_1$.
	\item If $\eps\not\in\L_{\D}(v_1)$ and $v_2$ is in Case 3, let $(\D'',v'') = \union(\D,v_1,r(v_2))$ as defined for a regular \dsabbr. We define $V' = V'' \cup\{v'\}$, $I' = I''\cup\{v'\}$ and $\lambda'(v') = \cup$ where $\lambda'$ is an extension of $\lambda''$. We connect $\ell'(v') = \ell(v_2)$ and $r'(v') = v''$.
	\item If $\lambda(v_1) = \eps$ and $\eps\not\in\L_{\D}(v_2)$, we define $V' = V \cup\{v'\}$, $I' = I\cup\{v'\}$ and $\lambda(v') = \cup$. We connect $\ell(v') = v_1$ and $r(v') = v_2$.
	\item If $\lambda(v_1) = \eps$ and $\lambda(v_2) = \eps$, we define $\D' = \D$ and $v' = v_1$.
	\item If $\lambda(v_1) = \eps$ and $v_2$ is in Case 3, we define $\D' = \D$ and $v' = v_2$.
	\item If $v_1$ is in Case 3 and $\eps\not\in\L_{\D}(v_2)$, let $(\D'',v'') = \union(\D,r(v_1),v_2)$ as defined for a regular \dsabbr.  We define $V' = V'' \cup\{v'\}$, $I' = I''\cup\{v'\}$ and $\lambda'(v') = \cup$ where $\lambda'$ is an extension of $\lambda''$. We connect $\ell'(v') = \ell(v_2)$ and $r'(v') = v''$. (*)
	\item If $v_1$ is in Case 3 and $\lambda(v_2) = \eps$, we define $\D' = \D$ and $v' = v_1$.
	\item If both $v_1$ and $v_2$ are in Case 3, let $(\D',v') = \union(\D,r(v_1),v_2)$ by using the construction of case (*).
\end{itemize}

Whenever $\D''$ is mentioned it is assumed to be equal to $(\Sigma, V'', I'', \ell'', r'', \lambda'')$.

It is straightforward to check that each operation behaves as expected. 
That is, if $\add(\D,a)\to(\D',v')$, then $\L_{\D}(v') = \{a\}$, if $\prod(\D,v_1,v_2)\to(\D',v')$, then $\L_{\D}(v') = \L_{\D}(v_1)\cdot\L_{\D}(v_2)$, and if $\union(\D,v_1,v_2)\to(\D',v')$, then $\L_{\D}(v_1)\cup\L_{\D}(v_2)$. 
Moreover, if both $v_1$ and $v_2$ are $\eps$-safe, then the resulting node $v'$ is $\eps$-safe as well for each operation.

Note that each operation falls into a fixed number of cases which can be checked exhaustively, and each construction has a fixed size, so they take constant time. 
Furthermore, each operation is fully persistent.

Finally, let $(\D',v')$ be a partial result obtained from applying the operations $\add$, $\prod$ and $\union$ such that $\D'$ is unambiguous. 
The proof follows from Claim~\ref{appendix:eps-enum}.

\renewcommand{\atitle}{\ref{sec:eval}}
\section{Proofs from Section~\atitle}\label{sec:appendixeval}

\subsection{Proof of Lemma \ref{vpt:steps}}

	We will prove the lemma by induction on $k$. 
	The case $k = 0$ is trivial since $\clevel(0) = \spanc{0}{0}$, $S^0_{p,q}$ is empty and $\llevel(0)$ is not defined. 
	We assume that statements 1 and 2 of the lemma are true for $k-1$ and below. 
	
	If $a_k\in\opS$, the algorithm proceeds into {\sc OpenStep} to build $S^k$ and $T^k$. 
	Statement 1 can be proved trivially since $\clevel(k) = \spanc{k}{k}$, similarly as for the base case.
	For statement 2 let $\llevel(k) = \spanc{i}{k-1}$, and consider a run $\rho\in\Runs(\cT,w\spanc{1}{k})$ such that $\rho\spanc{i}{k}$ starts on $p$ and ends on $q$ for some $p,q$ and $\gamma$, and let $p'$ be its second-to-last state. 
	Since $a_k$ is an open symbol, then the string $a_{i+1}\cdots a_{k-1}$ is well-nested, so it holds that $\clevel(k-1) = \spanc{i}{k-1}$. 
	Therefore, from our hypothesis it holds that $\L_{\D}(S^{k-1}_{p,p'})$ contains $\out(\rho\spanc{i}{k-1})$, and so, $\out(\rho\spanc{i}{k})$ is included in $\L_{\D}(T^{k}_{p,\gamma,q})$ at some iteration of $T^{k}_{p,\gamma,q}$ at line 36. 
	To show that every element in $\L_{\D}(T^k_{p,\gamma,q})$ corresponds to some run $\rho\in\Runs(\cT,w\spanc{1}{k})$, we note that the only step that modifies $T^k_{p,\gamma,q}$ is line 36, which is reached only when a valid subrun from $i$ to $k$ can be constructed.
	
	If $a_k\in\clS$, the algorithm proceeds into {\sc CloseStep} to build $S^k$ and $T^k$. 
	Let $\clevel(k) = \spanc{j}{k}$.
	In this case, statement 2 can be deduced directly from the hypothesis since $j < k$ and the table on the top of $T^k$ is the same as $T^{j}$.
	To prove statement 1 notice that since $a_k$ is a close symbol it holds that $\clevel(k-1) = \spanc{j'}{k-1}$ and $\llevel(k-1) = \spanc{j}{j'-1}$ for some $j'$. 
	Consider a run $\rho\in\Runs(\cT,w)$ such that $\rho\spanc{j}{k}$ starts on $p$, ends on $q$, and the last symbol pushed onto the stack is $\gamma$.
	This run can be subdivided in three subruns from $p$ to $p'$, from $p'$ to $q'$, and a transition from $q'$ to $q$ as it is illustrated in Figure~\ref{fig:delta-schema} (Right). 
	The first  two subruns correspond to $\rho\spanc{j}{j'+1}$ and $\rho\spanc{j'}{k-1}$, for which $\out(\rho\spanc{j}{j'+1})\in\L_{\D}(T^{k-1}_{p,\gamma,q})$ and $\out(\rho\spanc{j'}{k-1})\in\L_{\D}(S^{k-1}_{p',q'})$.
	Therefore, $\out(\rho\spanc{j}{k}) \in \L_{\D}(S^k_{p,q})$ at some iteration of line 47.
	To show that every element in $S^k_{p,q}$ corresponds to some run $\rho\in\Runs(\cT, w\spanc{1}{k})$, note that the only line at which $S^k_{p,q}$ is modified are is line 47, which is reached only when a valid run from $j$ to $k$ has been constructed.
	
\subsection{Proof of Theorem \ref{eval:prep}}

This theorem is a straightforward consequence of Lemma~\ref{vpt:steps}.

\subsection{Proof of Lemma \ref{eval:unambiguous}}

\begin{proof}
	For the sake of simplification, assume that $\cT$ is I/O-unambiguous on subruns as well. Formally, we extend the condition so that for every well-nested word $w$, span $\spanc{i}{j}$ and $\mu\in\Omega^*$ there exists only one run $\rho\in\Runs(\cT,w)$ such that $\mu = \out(\rho\spanc{i}{j})$.
	Towards a contradiction, we assume that $\cD$ is not I/O-unambiguous. Therefore, at least one of these conditions must hold: (1) There is some union node $v$ in $\cD$ for which $\L_{\cD}(\ell(v))$ and $\L_{\cD}(r(v))$ are not disjoint, or (2) there is some product node $v$ for which there are at least two ways to decompose some $\mu\in\L_{\cD}(v)$ in non-empty strings $\mu_1$ and $\mu_2$ such that $\mu = \mu_1\cdot\mu_2$ and $\mu_1\in\L(\ell(v))$ and $\mu_2\in\L_{\D}(r(v))$. 
	
	Assume the first condition is true and let $v$ be an union node that satisfies it, and let $k$ be the step in which it was added to $\cD$. If this node was added on {\sc OpenStep}, then the node $v$ represents a subset of the subruns defined in condition 1 of Lemma~\ref{vpt:steps}. Consider two different iterations of lines 35-36 on step $k$ where two nodes $v$ and $v'$ were united for which there is an element $\mu\in\L_{\cD}(v)\cap\L_{\cD}(v')$. Since these nodes were assigned to $T_{p,\gamma,q}$ on different iterations, the states $p'$ that were being considered must have been different. Therefore, if $\llevel(k) = \spanc{i}{j}$, $\mu = \out(\rho\spanc{i}{k}) = \out(\rho'\spanc{i}{k})$ for two runs $\rho$ and $\rho'$ where the $(k-1)$-th state is different. This violates the condition that $\cT$ is unambiguous. If this node was added on {\sc CloseStep}, we can follow an analogous argument. Note that union nodes created on a $\prod$ operation are unambiguous by construction (see Theorem~\ref{theo:data-structure-eps}).
	
	Assume now that the second condition is true and let $v$ be a node for which the condition holds and let $k$ be the step where it was created. We note that this node could not have been created in {\sc OpenStep} since the only step that creates product nodes is line 36, where $v_{\lambda}$ has the label $(\oout,k)$, and $S_{p,p'}$ is connected to nodes that were created in a previous step, so all of the elements $\mu\in\L(S_{p,p'})$ only contain pairs $(\oout,j)$ where $j < k$. We can follow a similar argument to prove that this node could not have been created in line 45 of {\sc CloseStep}. We now have that $v$ was created in line 44 of {\sc OpenStep}, and therefore $\ell(v) = T^{k-1}_{p,\gamma,q}$ and $r(v) = S^{k-1}_{p',q'}$ unless either of these indices were empty. However, that is not possible since we assumed that the step where $v$ was created was $k$, and if either were empty, no node would have been created. Now let $\mu\in\L(v)$ be such that there exist strings $\mu_1,\mu_1'\in\L(T^{k-1}_{p,\gamma,q})$ and $\mu_2,\mu_2'\in\L(S^{k-1}_{p',q'})$ such that $\mu = \mu_1\mu_2 = \mu_1'\mu_2'$ and $\mu_1 \neq \mu_1'$. Without loss of generality, let $\mu''$ be the non-empty suffix in $\mu_1$ such that $\mu_1'\mu'' = \mu_1$. Here we reach a contradiction since $\mu''$ is a prefix of $\mu_2$ and thus it must contain a pair $(\oout,j)$ such that and $j \in\llevel(k)$ and $j\in\clevel(k)$, which is not possible.
	
	The fact that all nodes in $\cD$ are $\eps$-safe carries easily from Theorem~\ref{theo:data-structure-eps}.
\end{proof}

\section{Applications in document spanners}\label{sec:appendixspanners}

This section presents an application of our enumeration algorithm to the evaluation of recursive spanners~\cite{PeterfreundCFK19}. Practical formalisms to define document spanner for information extraction with recursion was only proposed very recently. In~\cite{liatpaper}, the author suggests using extraction grammars to specify document spanners, which is the natural extension of regular spanners to a controlled form of recursion. Furthermore, the author gives an enumeration algorithm for unambiguous functional extraction grammars that outputs the results with constant-delay after quintic time preprocessing (i.e., in the document). We can show a streaming enumeration algorithm with update-time that is independent of the document, and output-linear delay by restricting to the class of visibly pushdown extraction grammars. We proceed by recalling the framework of document spanners and extraction grammars to define the class of visibly pushdown extraction grammars and state the main algorithmic result. 

We start by recalling the basics of document spanners~\cite{FaginKRV15}. Fix an alphabet~$\Sigma$ and a set of variables $\varset$ such that $\Sigma \cap \varset = \emptyset$. A document $d$ over $\Sigma$ is basically a word in $\Sigma^*$. A span $s$ of a document $d$ is a pair $\spanc{i}{j}$ of natural numbers $i$ and $j$ with $1 \leq i \leq j \leq |d|+1$. Intuitively, a span represents a substring of $d$ by identifying the starting and ending position. 
We denote by $\spanset(d)$ the set of all possible spans of $d$.
Let $X \subseteq \varset$ be a finite set of variables.
An $(X, d)$-mapping $\smap\colon X \rightarrow \spanset(d)$ assigns variables in $X$ to spans of $d$. An $(X, d)$-relation is a finite set of $(X, d)$-mappings. Then a document spanner $P$ (or just spanner) is a function associated with a finite set $X$ of variables that maps documents $d$ into $(X, d)$-relations.  

We use the framework of extraction grammars, recently proposed in~\cite{liatpaper}, to specify document spanners. 
For $X \subseteq \varset$, let $\varcaptures{X} = \{\varop{x}, \varcl{x}\mid x\in X\}$ be the set of captures of $X$ where, intuitively, $\varop{x}$ denotes the opening of $x$, and $\varcl{x}$ its closing. 
An \emph{extraction context-free grammar}, or \emph{extraction grammar} for short, is a tuple $G = (X, V, \Sigma, S, P)$ such that $X \subseteq \varset$, $V$ is a finite set of non-terminals symbols with $V\cap \varset = \emptyset$, $\Sigma$ is the alphabet of terminal symbols with $\Sigma \cap V = \emptyset$, $S \in V$ is the start symbol, and $P \subseteq V \times (V \cup \Sigma \cup \varcaptures{X})^*$ is a finite relation. In the literature, the elements of $V$ are also referred as ``variables'', but we call them non-terminals to distinguish $V$ from~$\varset$.
Each pair $(A, \alpha) \in P$ is called a production and we write it as $A \rightarrow \alpha$. The set of productions $P$ defines the (left) derivation relation $\gprod{G} \ \subseteq \, (V \cup \Sigma \cup \varcaptures{X})^* \times (V \cup \Sigma \cup \varcaptures{X})^*$ such that $w A \beta \gprod{G} w \alpha \beta$ iff $w \in (\Sigma \cup \varcaptures{X})^*$, $A \in V$, $\alpha, \beta \in (V \cup \Sigma \cup \varcaptures{X})^*$, and $A \rightarrow \alpha \in P$. We denote by $\gprod{G}^*$ the reflexive and transitive closure of $\gprod{G}$. Then the language defined by $G$ is $\cL(G) = \{w \in (\Sigma \cup \varcaptures{X})^* \mid S \gprod{G}^* w\}$. 
A word $w \in \cL(G)$ is called a \emph{ref-word} produced by $G$. 

In order to define a spanner from $G$, we need to interpret ref-words as mappings~\cite{Freydenberger19}. Formally, a ref-word $r = a_1 \ldots a_n \in (\Sigma \cup \varcaptures{X})^*$ is called valid for $X$ if, for every $x \in X$, there exists exactly one position $i$ with $a_i = \varop{x}$ and exactly one position $j$ with $a_j = \ \varcl{x}$, such that $i < j$. In other words, a valid ref-word defines a correct match of open and close captures. Moreover, each $x \in X$ induces a unique factorization of $r$ of the form $r = r_x^p \cdot \varop{x} \, \cdot \, r_x \, \cdot\, \varcl{x} \cdot r_x^s$. 
This factorization defines an $(X,d)$-mapping as follows. 
Let $\splain: (\Sigma \cup \varcaptures{X})^* \rightarrow \Sigma^*$ be the morphism that removes the captures from ref-words, namely, $\splain(a) = a$ when $a \in \Sigma$ and $\splain(c) = \eps$ when $c \in \varcaptures{X}$.
We extend $\splain$ to operate over strings in the obvious way.
Furthermore, let $r$ be a valid ref-word for $X$, $d$ be a document, and assume that $\splain(r) = d$.
Then we define the $(X,d)$-mapping $\smap^r$ such that $\smap^r(x) = \spanc{i}{j}$ iff $r = r_x^p \cdot \varop{x} \, \cdot \, r_x \, \cdot\, \varcl{x} \cdot r_x^s$, $i = |\splain(r_x^p)|+1$, and $j = i + |\splain(r_x)|$. 
Finally, the spanner $\sem{G}$ associated to an extraction grammar $G$ is defined over any document $d \in \Sigma^*$ as follows:
$$
\sem{G}(d) \ = \ \{ \, \mu^r \ \mid \ r \in \cL(G), \text{ $r$ is valid for $X$}, \text{ and } \splain(r) = d\, \}.
$$
There are two classes of extraction grammars that are relevant for our discussion.
The first class of grammars are called functional extraction grammars. 
An extraction grammar $G$ is \emph{functional} if every $r \in \cL(G)$ is valid for $X$.  In~\cite{liatpaper} it was shown that for any extraction grammar $G$ there exists an equivalent functional grammar $G'$ (i.e. $\sem{G} = \sem{G'}$). Non-functional grammars are problematic given that, even for regular spanners, their decision problems easily become intractable~\cite{MaturanaRV18,FreydenbergerKP18}. For this reason, we restrict to functional extraction grammars without loss of expressive power. 
The second class of grammars are called unambiguous extraction grammars. An extraction grammar $G$ is \emph{unambiguous} if for every $r \in \cL(G)$ there exists exactly one path from $S$ to $r$ in the graph $((V \cup \Sigma \cup \varcaptures{X})^*, \gprod{G})$. In other words, there exists exactly one left-most derivation.

We consider now a sub-class of extraction grammars for nested words. Let $\Sigma = (\opS, \clS, \noS)$ be a structured alphabet. Then a \emph{visibly pushdown extraction grammar} (VPEG) is a functional extraction grammar $G = (X, V, \Sigma, S, P)$ in which $\Sigma = (\opS, \clS, \noS)$ is a structured alphabet, and all the productions in $P$ are of one of the following forms: (1) $A \rightarrow \eps$; (2) $A \rightarrow a B$ such that $a \in \noS \cup \varcaptures{X}$ and $B \in V$; (3) $A \rightarrow \op{a} \, B \, \cl{b} \, C$ such that $\op{a} \in \opS$, $\cl{b} \in \clS$, and $B, C \in V$. Intuitively, rules $A \rightarrow a B$ allow to produce arbitrary sequences of neutral symbols, where rules $A \rightarrow \op{a} \, B \, \cl{b} \, C$ forces the word to be well-nested. 

Visibly pushdown extraction grammars are a subclass of extraction grammars that works for well-nested documents. In fact, the reader can notice that the visibly pushdown restriction for extraction grammars is the analog counterpart of visibly pushdown grammars\footnote{The definition of visibly pushdown grammars in~\cite{AlurM04} is slightly more complicated given that they consider nested words that are not necessary well-nested (see the discussion in Section~\ref{sec:prelim}).} introduced in~\cite{AlurM04}. Therefore, one could expect that VPEGs are less expressive than extraction grammars. 
Interestingly, we can use Theorem~\ref{theo:main} to give an efficient streaming enumeration algorithm for evaluating VPEG. 
\begin{theorem}\label{theo:spanners}
	Fix a set of variables $X$. The problem of, given a visibly pushdown extraction grammar $G = (X, V, \Sigma, S, P)$ and a stream $\Stream$, enumerating all $(X,\Stream[1,n])$-mappings of $\sem{G}(w)$ can be solved with update-time $\cO(2^{|G|^3})$, and output-linear delay. Furthermore, if $G$ is restricted to be unambiguous, then the problem can be solved with update-time $\cO(|G|^3)$.
\end{theorem} 
This result goes by constructing an extraction pushdown automata~\cite{liatpaper} from $G$, and reduce it to a visibly pushdown transducers. Note that, although the update-time of the algorithm is exponential in the size of the grammar, in terms of data-complexity the update-time is constant. Furthermore, for the special case of unambiguous grammars the update-time is even polynomial. Unambiguous grammars are very common in parsing tasks~\cite{aho1986compilers} and, thus, this restriction could be useful in practice.

\subsection{Proof of Theorem~\ref{theo:spanners}}

To link the model of visibly pushdown extraction grammars and visibly pushdown automata we define another class of automata based on the ideas in~\cite{liatpaper}. Let $\cA$ be an {\em extraction visibly pushdown automaton} (EVPA) if $\cA = (X,Q,\Sigma,\Gamma,\Delta,I,F)$ where $X$ is a set of variables, $Q$ is a set of states, $\Sigma = (\opS,\clS,\noS)$ is a visibly pushdown alphabet, $\Gamma$ is a stack alphabet, $\Delta \subseteq
(Q \times \opS \times Q \times \Gamma) \ \cup (Q \times \clS \times \Gamma \times Q) \ \cup (Q \times (\noS\cup\varcaptures{X}) \times Q)$, $I$ is a set of initial states, and $F$ is a set of final states. Note that this is a simple extension of \vpa where neutral transitions are allowed to read neutral symbols or captures in $X$. We define the runs as in \vpa except the input in a EVPA is a ref-word $w\in(\Sigma\cup\varcaptures{X})$, and we say that $w\in\L(\cA)$ if and only if there is an accepting run of $\cA$ on $w$. Furthermore, $\cA$ is unambiguous if for every ref-word $w$ there exists at most one accepting run of $\cA$ over $w$. It is straightforward to see that this is a direct counterpart to visibly pushdown extraction grammars. Therefore, we can use the ideas in \cite{AlurM04} to obtain a one-to-one conversion from one to another.

\begin{claim}\label{appendix:spannerclaim}
	For a given VPEG $G$ there exists an EVPA $\cA_G$ such that $\L(G) = \L(\cA_G)$. Moreover, $\cA_G$ is unambiguous iff $G$ is unambiguous, and $\cA_G$ can be constructed in time $\cO(\vert G\vert)$.
\end{claim}
\begin{proof}
	Let $G = (X, V, \Sigma, S, P)$ be a VPEG. We construct a EVPA $\cA_G = (X,Q,\Sigma,\Gamma,\Delta,I,F)$ such that $\L(G) = \L(\cA_G)$ using an almost identical construction to the one in Theorem 6 of~\cite{AlurM04}. The only differences arise in that our structure is defined for well-nested words, so it can be slightly simpified, and in the case where a production is of the form $X\to aY$, for which we add the possibility that $a\in\varcaptures{X}$. This construction provides one transition in $\Delta$ per production in $P$, and in some cases it needs to check if a variable is nullable. Checking if a single variable is nullable is costly, but by a constant number of traversals in $P$ it is possible to check which variables in $X$ are nullable or not, which can be done before building $\Delta$. Therefore, this construction can be done in time $\cO(\vert P\vert)$. Finally, $\cA_G$ is unambiguous if and only if $G$ is unambiguous, which is another consequence of Theorem 6 of~\cite{AlurM04}.	
\end{proof}

Here we define the spanner $\br{\cA}$ for a given EVPA $\cA$ identically to the definition for an extraction grammar. Note that from the proof it also follows that if $G$ is functional, then $\cA_G$ is functional as well.

For the next part of the proof assume that $\cA_G$ is unambiguous. We will show that for an EVPA $\cA$ and stream $\Stream$, the set $\br{\cA}(d)$, can be enumerated with output-linear delay and update-time $\cO(\vert\cA_G\vert^3)$, for $d = \Stream[1,n]$. Towards this goal, we will start with an unambiguous $\cA_G = (X,Q,\Sigma,\Gamma,\Delta,I,F)$ and convert it into a VPT $\cT_G$ with output symbol set $2^{\varcaptures{X}}$ and use our algorithm to enumerate the set $\br{\cT_G}(w)$ where $d' = d\#$, using a dummy symbol $\#$. Each element $w\in\br{\cT_G}(d')$ can then be converted into a mapping $\mu\in\br{G}(d)$ after it is given as output in time $\cO(\vert\mu\vert)$.

Let $\cT_G = (Q', \Sigma', \Gamma, \oalph, \Delta', \qinit, F')$ 
where $Q' = Q\cup\{q_f\}$, 
$\Sigma' = (\opS,\clS,\noS_{\#})$ 
such that $\noS_{\#} = \noS\cup\{\#\}$, $\oalph = 2^{\varcaptures{X}}\cup\{\eps\}$ and $F' = \{q_f\}$. 
To define $\Delta'$ we introduce a {\sf merge} operation on a path over $\cA_G$. 
This is defined for any non-empty sequence of transitions $t = (p_1,v_1,q_1)(p_2,v_2,p_2)\cdots(p_m,v_m,q_m)\in\Delta^*$ such that $v_i\in\varcaptures{X}$ for $i\in[1,m]$, and $q_i = p_{i+1}\in[i,m-1]$. 
If these conditions hold, we say that $t$ is a v-path ending in $p_m$. 
Let $t$ be such a v-path and let $S = \{v_1,\ldots,v_m\}$. For $\op{a}\in\opS$, and a transition $(p,\op{a},\gamma,q)$ such that $p = q_m$, we define ${\sf merge}(t, (p,\op{a},\gamma,q)) := (p_1,\op{a},S,\gamma,q)$. For $\cl{a}\in\clS$ and a transition $(p,\cl{a},q,\gamma)$ such that $p = q_m$, we define ${\sf merge}(t,(p,\cl{a},q,\gamma)) := (p_1,\cl{a},S,q,\gamma)$. For $a\in\noS$ and a transition $(p,a,q)$ such that $p = q_m$, we define ${\sf merge}(t,(p,a,q)) := (p_1,a,S,q)$. We now define $\Delta'$ as follows:
\begin{align*}
	\Delta' = \,&\{(p,\op{a},\eps,\gamma,q)\mid (p,\op{a},\gamma,q)\in\Delta\}\,\cup\\
	&\{{\sf merge}(t,(p,\op{a},\gamma,q))\mid\text{there is a v-path $t\in\Delta^*$ ending in $p$ and }(p,\op{a},\gamma,q)\in\Delta\}\,\cup\\
	&\{(p,\cl{a},\eps,q,\gamma)\mid (p,\cl{a},q,\gamma)\in\Delta\}\,\cup\\
	&\{{\sf merge}(t,(p,\cl{a},q,\gamma))\mid\text{there is a v-path $t\in\Delta^*$ ending in $p$ and }(p,\cl{a},q,\gamma)\in\Delta\}\,\cup\\
	&\{(p,a,\eps,q)\mid (p,a,q)\in\Delta\}\, \cup\\
	&\{{\sf merge}(t,(p,a,q))\mid\text{there is a v-path $t\in\Delta^*$ ending in $p$ and }(p,a,q)\in\Delta\}\,\cup\\
	&\{{\sf merge}(t,(p,\#,q_f))\mid\text{there is a v-path $t\in\Delta^*$ ending in $p$ and } p\in F\}.
\end{align*}
Since $\cA_G$ is unambiguous, and therefore, the transitions in $\Delta$ define a DAG over $Q$, from which we deduce that $\Delta$ is well-defined.
By the definition of {\sf merge} it is straightforward to check that every accepting path in $\cA_G$ is preserved in $\cT_G$, in the sense that if $r\in\L(\cA_G)$ then there exists an accepting path of $\cT_G$ over $({\sf plain}(r)\#, \omega)$, where $\omega$ is a sequence of elements in $2^{\varcaptures{X}}\cup\{\eps\}$ built from the captures present in $r$.

To show accepting pairs for $\cT_G$ correspond to a valid counterpart in $\cA_G$ let $(d,\omega)$ be an input/output pair that is accepted by $\cT_G$. Note that $d = d'\#$ from our definition of $\Delta'$. It can be seen that for every accepting path of $\cT_G$ over $(d,\omega)$ there exists at least one ref-word $r$ built from $d$ and $\omega$. However, note that for every such ref-word $r$ the only difference may be in the order of the elements inside each group of contiguous captures, which will be asociated to the same position in $\mu^r$. From this, it follows that for each accepting pair $(d,\omega)$ there exists only one mapping $\mu\in\br{\cA_G}(d')$ that can be built from $(d,\omega)$.

The size of $\Delta$ is bounded by the number of valid v-paths there could exist in $\cA_G$. Recall that $\cA_G$ is functional, an thus every v-path in $\cA_G$ contains at most one instance of each element in $\varcaptures{X}$. From this it follows that the size of $\cT_G$ is in $\cO(\vert\Delta\vert\vert 2^{\varcaptures{X}}\vert)$. Furthermore, since the transitions in $\Delta$ form a DAG over $Q$, each of these v-paths can be found by a single traversal over $\cA_G$, so building $\cT_G$ takes time $\cO(\vert\Delta\vert)$.

By using the algorithm detailed in Section~\ref{sec:eval} we can enumerate the set $\br{\cT_G}(d)$ with update-time $\cO(\vert\cT_G\vert^3)$ and output-linear delay. 
However, with a more fine-grained analysis of the algorithm, we note that the update-time is bounded by $\vert Q'\vert^2\vert\Delta'\vert\in \cO(\vert Q\vert^2\vert\Delta\vert\vert 2^{\varcaptures{X}}\vert)$. We modify the enumeration algorithm slightly so that for each output $\omega\in\br{\cT_G}(d)$ we build the expected output in $\br{G}(d)$. We do this by checking $w$ symbol by symbol and building a mapping $\mu\in\br{G}(d)$, and this can be done in time $\cO(\vert\mu\vert)$. As the set $X$ is fixed, it follows that this enumeration can be done with update-time $\cO(\vert G\vert^3)$ and output-linear delay.

Finally, we adress the case where $G$ is an arbitrary VPEG. The way we deal with this case is by determinizing the EVPA constructed in Claim~\ref{appendix:spannerclaim}. This can be done in time $\cO(2^{\vert \cA_G\vert})$. From here, we can follow the reasoning given for the unambiguous case to prove the statement.

\end{document}